\documentclass[aos]{imsart}

\RequirePackage{amsthm,amsmath,amsfonts,amssymb}
\RequirePackage[authoryear]{natbib} 
\RequirePackage[colorlinks,citecolor=blue,urlcolor=blue]{hyperref}
\RequirePackage{graphicx}

\usepackage[greek, english]{babel}

\usepackage{amsmath,amssymb,amsfonts}
\usepackage{mathtools} 
\mathtoolsset{showonlyrefs=true}
\usepackage{framed} 
\usepackage[ntheorem, framemethod=TikZ]{mdframed} 
\usepackage{bm}
\usepackage{isomath}
\usepackage{alphabeta}

\newcommand{\qedwhite}{\hfill \ensuremath{\Box}}

\DeclareMathOperator{\probMeasSet}{\mathcal{P}}%
\DeclareMathOperator{\borel}{\mathcal{B}}%
\newcommand{\iidSim}{\mathrel{\overset{\mathsc{iid}}{\sim}}}

\newcommand{\nParticles}{N} 
\newcommand{\nDimensions}{D} 
\newcommand{\nTimeSteps}{T} 
\newcommand{\particleIndex}{n} 
\newcommand{\particleIndexAlt}{m} 
\newcommand{\particleIndexAltAlt}{l} 
\newcommand{\dimensionIndex}{d} 
\newcommand{\timeIndex}{t} 

\newcommand{\state}{\mathbf{x}} 
\newcommand{\State}{\mathbf{X}}

\newcommand{\OutputParticleIndex}{K} 
\newcommand{\outputParticleIndex}{k}

\newcommand{\logWeightSingle}{w}
\newcommand{\logBackwardWeightSingle}{v}
\newcommand{\logWeightRandomWalkSingle}{\bar{\logWeightSingle}}

\newcommand{\logWeight}{\mathbf{\logWeightSingle}}
\newcommand{\logBackwardWeight}{\mathbf{\logBackwardWeightSingle}}
\newcommand{\logWeightRandomWalk}{\widebar{\mathbf{\logWeightSingle}}}
\newcommand{\logBackwardWeightRandomWalk}{\widebar{\mathbf{\logBackwardWeightSingle}}}

\newcommand{\ResamplingKernel}[2]{R_{#1}^{#2}}
\newcommand{\BackwardKernel}[2]{B_{#1}^{#2}}

\newcommand{\acceptanceRate}[2]{\alpha_{#1}^{#2}}

\newcommand{\ExpectationCsmcKernel}[2]{\mathbb{E}_{#1}^{#2}}

\newcommand{\IteratedCsmcKernel}[2]{\mathbb{P}_{#1}^{#2}}

\newcommand{\InducedIteratedCsmcKernel}[2]{\mathbf{P}_{#1}^{#2}}

\newcommand{\selectionFunctionBoltzmann}[1]{\varPsi^{#1}}
\newcommand{\selectionFunctionRosenbluth}[1]{\varPhi^{#1}}

\newcommand{\ResamplingKernelRandomWalk}[2]{\widebar{R}_{#1}^{#2}}
\newcommand{\BackwardKernelRandomWalk}[2]{\widebar{B}_{#1}^{#2}}

\newcommand{\acceptanceRateRwCsmc}[2]{\bar{\alpha}_{#1}^{#2}}

\newcommand{\ExpectationCsmcKernelRandomWalk}[2]{\widebar{\mathbb{E}}_{#1}^{#2}}
\newcommand{\IteratedCsmcKernelRandomWalk}[2]{\widebar{\mathbb{P}}_{#1}^{#2}}

\newcommand{\InducedIteratedCsmcKernelRandomWalk}[2]{\widebar{\mathbf{P}}_{#1}^{#2}}

\newcommand{\ExpectationEhmmKernelRandomWalk}[2]{\widetilde{\mathbb{E}}_{#1}^{#2}}
\newcommand{\IteratedEhmmKernelRandomWalk}[2]{\widetilde{\mathbb{P}}_{#1}^{#2}}

\newcommand{\InducedIteratedEhmmKernelRandomWalk}[2]{\widetilde{\mathbf{P}}_{#1}^{#2}}

\newcommand{\acceptanceRateRwEhmm}[2]{\tilde{\alpha}_{#1}^{#2}}

\newcommand{\esjd}{\mathrm{ESJD}}
\newcommand{\esjdRandomWalk}{\smash{\widebar{\mathrm{ESJD}}}}
\newcommand{\esjdRandomWalkEhmm}{\smash{\widetilde{\mathrm{ESJD}}}}

\newcommand{\ConcentrationSet}{\mathbf{F}}

\newcommand{\SamplingKernelRandomWalk}[2]{S_{#1}^{#2}}

\newcommand{\particle}{\mathbf{z}}
\newcommand{\Particle}{\mathbf{Z}}
\newcommand{\particleSingle}{z}

\newcommand{\Target}{\bm{\pi}}

\newcommand{\Mutation}{\mathbf{M}}
\newcommand{\mutation}{\mathbf{m}}
\newcommand{\Potential}{\mathbf{G}}

\newcommand{\TargetSingle}{\pi}
\newcommand{\MutationSingle}{M}
\newcommand{\mutationSingle}{m}
\newcommand{\PotentialSingle}{G}

\newcommand{\spaceState}{\mathbf{E}} 
\newcommand{\spaceStateSingle}{\reals} 
\newcommand{\sigFieldState}{\bm{\mathcal{E}}}%
\newcommand{\sigFieldStateSingle}{\borel(\reals)}%

\newcommand{\stateSingle}{x}

\DeclareMathOperator{\dN}{N} 
\DeclareMathOperator{\dUnif}{Unif} 

\DeclareMathOperator{\dDirac}{\textnormal{\greektext d\latintext}}
\newcommand{\standardNormalCdf}{\mathrm{\Phi}}
\newcommand{\standardNormalPdf}{\textnormal{\greektext f\latintext}}

\newcommand{\diff}{\mathrm{d}} 
\newcommand{\intDiff}{\,\mathrm{d}} 
\DeclareMathOperator{\Prob}{\mathbb{P}} 
\DeclareMathOperator{\E}{\mathbb{E}} 
\DeclareMathOperator{\var}{var} 
\DeclareMathOperator{\ind}{\mathbb{I}} 
\newcommand{\ccdot}{\,\cdot\,} 
\newcommand{\T}{\mathrm{T}} 
\DeclareMathOperator{\bo}{\mathrm{O}}%
\DeclareMathOperator{\lo}{\mathrm{o}}%
\newcommand{\ProbSymbol}{\mathbb{P}}%
\newcommand{\lip}{\textnormal{\textsc{lip}}}%
\DeclareMathOperator{\cov}{cov}%
\DeclareMathOperator{\corr}{corr}%
\newcommand{\iMat}{\mathrm{I}}
\newcommand{\unitMat}{\bm{1}}
\newcommand{\zeroMat}{\bm{0}}
\newcommand{\eul}{\mathrm{e}}

\newcommand{\convergesInProbability}{\ensuremath{\mathrel{\mathop{\rightarrow}\nolimits_{\ProbSymbol}}}}
\newcommand{\convergesInDistribution}{\ensuremath{\mathrel{\mathop{\rightarrow}\nolimits_{\text{d}}}}}

\newcommand{\calA}{\mathcal{A}}%
\newcommand{\calB}{\mathcal{B}}%
\newcommand{\calF}{\mathcal{F}}%
\newcommand{\calI}{\mathcal{I}}%
\newcommand{\calM}{\mathcal{M}}%
\newcommand{\calR}{\mathcal{R}}%
\newcommand{\calS}{\mathcal{S}}%
\newcommand{\calU}{\mathcal{U}}%
\newcommand{\calV}{\mathcal{V}}%
\newcommand{\calW}{\mathcal{W}}%
\newcommand{\mathsc}[1]{{\normalfont\textsc{#1}}}%

\newcommand{\reals}{\mathbb{R}}%
\newcommand{\naturals}{\mathbb{N}}%
\DeclareFontFamily{U}{mathx}{\hyphenchar\font45}
\DeclareFontShape{U}{mathx}{m}{n}{<-> mathx10}{}
\DeclareSymbolFont{mathx}{U}{mathx}{m}{n}
\DeclareMathAccent{\widebar}{0}{mathx}{"73}

\newcommand{\prodSubstackAligned}[6]{
  \prod_{\substack{\mathllap{#1} #2 \mathrlap{#3}\\
                   \mathllap{#4} #5 \mathrlap{#6}}}}%
                   
\newcommand{\sumSubstackAligned}[6]{
  \sum_{\substack{\mathllap{#1} #2 \mathrlap{#3}\\
                   \mathllap{#4} #5 \mathrlap{#6}}}}%

\newcommand{\IE}{i.e.\ }
\newcommand{\WRT}{w.r.t.\ }

\newcommand{\mendmark}{\ifmmode \eqno \hbox{\ensuremath{\triangleleft}} \else \pushright{\hbox{\quad\ensuremath{\triangleleft}}}\fi}
\newcommand{\tendmark}{\hfill \ensuremath{\triangleleft}}

\usepackage[inline]{enumitem}

\usepackage{xcolor}

\definecolor{pms286}{cmyk}{1,0.6,0,0.06}
\definecolor{pms279}{cmyk}{0.96,0.34,0,0}
\definecolor{cyan}{cmyk}{1,0,0,0}

\definecolor{myblue}{cmyk}{0.95,0.37,0,0.57}
\definecolor{darkgrey}{cmyk}{0.34,0.08,0,0.76}
\definecolor{darkergrey}{cmyk}{0,0,0.01,0.64}
\definecolor{lightergrey}{cmyk}{0,0,0.01,0.45}
\definecolor{lightgrey}{cmyk}{0,0,0,0.16}
\definecolor{erylightgrey}{cmyk}{0,0,0,0.1}

\definecolor{mygray}{rgb}{0.5,0.5,0.5}

\colorlet{lightred}{white!80!red}
\colorlet{lightblue}{white!80!blue}

\definecolor{pms072}{cmyk}{1,0.97,0,0}
\definecolor{pms012}{cmyk}{0,0.2,0,0,1}
\definecolor{pms200}{cmyk}{0,1,0.65,0.15}
\definecolor{pmsgreen}{cmyk}{1,0,0.65,0}
\definecolor{pmspurple}{cmyk}{0.43,0.91,0,0}
\definecolor{pms382}{cmyk}{0.305,0,0.94,0}

\definecolor{plotgreen}{rgb}{0.2509804, 0.5607843, 0.2823529}
\definecolor{plotmagenta}{rgb}{0.7764706, 0.1254902, 0.8117647}
\definecolor{plotblue}{rgb}{0.1764706, 0.1607843, 0.9607843}
\definecolor{plotgrey}{cmyk}{0,0,0.01,0.55}

\usepackage{graphicx}
\usepackage{float}
\usepackage{tabularx}
\usepackage{booktabs}

\usepackage[labelfont = sc, labelsep = period]{caption}
\usepackage[labelfont = normal, labelsep = space]{subcaption}
\makeatletter
\renewcommand{\fnum@figure}{Fig.~\thefigure}
\makeatother

\usepackage{tikz}
\usetikzlibrary{backgrounds, fit,matrix,positioning, external}
\usetikzlibrary{fpu}
\def\scriptsize{\fontsize{7}{7}}

\usepackage{hyperref}%
\hypersetup{%
  breaklinks=true,%
  linktocpage=true,
  colorlinks=true,%
  linkcolor=magenta,%
  pdftitle={},%
  pdfsubject={},%
  pdfauthor={Axel Finke},%
  pdfkeywords={},%
  pdfproducer={}%
  }%

\usepackage{ifthen}%
\usepackage{mfirstuc}
\usepackage{xkeyval}%
\usepackage{xfor}%
\usepackage{amsgen}%
\usepackage{etoolbox}%
\usepackage{datatool-base}%

\usepackage[%
  xindy,%
  style=long,%
  toc=true,%
  nonumberlist,%
  acronym,%
  acronymlists={ensemble_abbreviations},%
  hyperfirst=true
  ]{glossaries}%

\newacronym{PMCMC}{PMCMC}{particle Markov chain Monte Carlo}%
\newacronym{EMCMC}{EMCMC}{ensemble Markov chain Monte Carlo}%
\newacronym{APF}{APF}{auxiliary particle filter}%
\newacronym{FAAPF}{FA-APF}{fully-adapted auxiliary particle filter}%
\newacronym{PF}{PF}{particle filter}%
\newacronym{HMM}{HMM}{hidden Markov model}%
\newacronym{CPF}{CPF}{conditional particle filter}%
\newacronym{CPFBS}{CPF-BS}{conditional particle filter with backward sampling}%
\newacronym{CPFAS}{CPF-AS}{conditional particle filter with backward sampling}%
\newacronym{PG}{PG}{particle Gibbs}%
\newacronym{PGBS}{PG-BS}{particle Gibbs sampler with backward sampling}%
\newacronym{PGAS}{PG-AS}{particle Gibbs sampler with backward sampling}%
\newacronym{SQMC}{SQMC}{sequential quasi Monte Carlo}%
\newacronym{RQMC}{RQMC}{randomised quasi Monte Carlo}%
\newacronym[user1={ancestor-sampling}]{AS}{AS}{ancestor sampling}%
\newacronym[user1={backward-sampling}]{BS}{BS}{backward sampling}%
\newacronym{PDF}{PDF}{probability density function}%
\newacronym{IID}{IID}{independent and identically distributed}%
\newacronym{MCMC}{MCMC}{Markov chain Monte Carlo}%
\newacronym{MH}{MH}{Metropolis--Hastings}%
\newacronym{ESS}{ESS}{effective sample size}%
\newacronym{PMMH}{PMMH}{particle marginal Metro\-po\-lis--Has\-tings}%
\newacronym{MCWM}{MCWM}{Monte Carlo within Metropolis}%
\newacronym{CDF}{CDF}{cumulative distribution function}%
\newacronym{SMC}{SMC}{sequential Monte Carlo}%
\newacronym{CSMC}{CSMC}{conditional sequential Monte Carlo}%
\newacronym{CSMCBS}{CSMC-BS}{CSMC with backward sampling}%
\newacronym{CSMCAS}{CSMC-AS}{CSMC with ancestor sampling}%

\newacronym{SDE}{SDE}{stochastic differential equation}%
\newacronym{WP}{w.p.\@}{with probability}%

\newacronym{ICSMC}{i-CSMC}{iterated conditional sequential Monte Carlo}%
\newacronym{IRWCSMC}{i-RW-CSMC}{iterated random-walk conditional sequential Monte Carlo}%
\newacronym{RWEHMM}{RW-EHMM}{random-walk embedded hidden Markov model}%
\newacronym{RWCSMC}{RW-CSMC}{random-walk conditional sequential Monte Carlo}%
\newacronym{EHMM}{EHMM}{embedded hidden Markov model}%

\newacronym{EPSRC}{EPSRC}{Engineering and Physical Sciences Research Council}%
\newacronym{LW}{LW}{Liu~\&~West}%
\newacronym{PL}{PL}{particle learning}%
\newacronym{FF}{FF}{fertility factor}%
\newacronym{CLT}{CLT}{central limit theorem}%
\newacronym{WLLN}{WLLN}{weak law of large numbers}%
\newacronym{BPF}{BPF}{blocked particle filter}%
\newacronym{IACT}{IACT}{integrated autocorrelation time}%

\newacronym{SIR}{SIR}{sampling--importance resampling}%
\newacronym{CSIR}{CSIR}{conditional sampling--importance resampling}%

\newacronym{MCMCSIR}{MCMC-SIR}{sampling importance resampling with MCMC proposals}%

\newacronym{ESJD}{ESJD}{expected squared jumping distance}%

\newacronym[\glslongpluralkey={particle filters with MCMC moves}]{MCMCPF}{MCMC-PF}{particle filter with MCMC moves}%
\newacronym[\glslongpluralkey={bootstrap particle filters with MCMC moves}]{MCMCBPF}{MCMC-BPF}{bootstrap particle filter with MCMC moves}%
\newacronym[\glslongpluralkey={fully-adapted particle filters with MCMC moves}]{MCMCFAAPF}{MCMC-FA-APF}{fully-adapted auxiliary particle filter with MCMC moves}%
\usepackage{csquotes}%

\usepackage[textwidth=2.5cm, textsize=scriptsize]{todonotes}

\startlocaldefs

\newtheorem{theorem}{Theorem}[section]
\newtheorem{proposition}[theorem]{proposition}
\newtheorem{lemma}[theorem]{Lemma}
\newtheorem{corollary}[theorem]{corollary}
\theoremstyle{remark}

\newtheorem*{example}{Example}

\newtheorem{conjecture}[theorem]{Conjecture}
\newtheorem{remark}[theorem]{remark}
\newtheorem*{namedproof}{Proof}

\newmdtheoremenv[
 ntheorem=true,
 skipbelow = .6\baselineskip plus 1ex minus 1ex,
 skipabove = .6\baselineskip plus 1ex minus 1ex,
 innerleftmargin = 0pt,
 innerrightmargin = 0pt,
 innertopmargin=-8pt, 
 leftline = false,
 rightline = false,
 needspace = 5ex 
]{framedAlgorithm}{Algorithm}

\endlocaldefs

\begin{document}

\begin{frontmatter}
\title{Conditional sequential Monte Carlo in high dimensions}
\runtitle{Conditional sequential Monte Carlo in high dimensions}

\begin{aug}
\author[A]{\fnms{Axel} \snm{Finke}\ead[label=e1]{a.finke@lboro.ac.uk}}
\and
\author[B]{\fnms{Alexandre H.\@} \snm{Thiery}\ead[label=e2]{a.h.thiery@nus.edu.sg}}
\address[A]{Department of Mathematical Sciences, Loughborough University, UK
\printead{e1}}

\address[B]{Department of Statistics and Applied Probability, National University of Singapore, Singapore
\printead{e2}}
\end{aug}

 \glsunset{MCMC}
 
\begin{abstract}
 \noindent{}The \emph{\gls{ICSMC}} algorithm from \citet{andrieu2010particle} is an \gls{MCMC} approach for efficiently sampling from the joint posterior distribution of the $T$ latent states in challenging time-series models, e.g.\ in non-linear or non-Gaussian state-space models. It is also the main ingredient in \emph{particle Gibbs} samplers which infer unknown model parameters alongside the latent states. In this work, we first prove that the \gls{ICSMC} algorithm suffers from a curse of dimension in the dimension of the states, $D$: it breaks down unless the number of samples (`particles'), $N$, proposed by the algorithm grows exponentially with $D$. Then, we present a novel `local' version of the algorithm which proposes particles using Gaussian random-walk moves that are suitably scaled with $D$. We prove that this \emph{\gls{IRWCSMC}} algorithm avoids the curse of dimension: for arbitrary $N$, its acceptance rates and \glsdesc{ESJD} converge to non-trivial limits as $D \to \infty$. If $T = N = 1$, our proposed algorithm reduces to a \glsdesc{MH} or Barker's algorithm with Gaussian random-walk moves and we recover the well known scaling limits for such algorithms.
\end{abstract}

\glsreset{MCMC}
\glsreset{ESJD}
\glsreset{CSMC}
\glsreset{MH}
\glsreset{IRWCSMC}
\glsreset{ICSMC}


\end{frontmatter}
\tableofcontents

\section{Introduction}

\subsection{Summary} 

This work analyses Monte Carlo methods for approximating the joint smoothing distribution (i.e.\ the joint distribution of all latent states) in high-dimensional state-space models. Developing efficient \gls{MCMC} algorithms for this task is challenging if the dimension of the latent states, the `spatial' dimension $\nDimensions$, or the number of observations, the `time horizon' $\nTimeSteps$, is large because of the difficulty of finding good `global' proposal distributions on a large ($\nDimensions\nTimeSteps$-dimensional) space. For this reason, the acceptance rate of \emph{independent \gls{MH}} kernels for this problem is typically $\bo(\eul^{-\nDimensions \nTimeSteps})$ which means that the algorithm suffers from a \emph{`curse of dimension',} \IE its complexity grows exponentially in the size ($TD$) of the problem. Throughout this work, we define \emph{complexity} as the number of full likelihood evaluations needed to control the approximation error of a fixed-dimensional marginal of the joint smoothing distribution.

For the moment, assume that $\nDimensions$ is fixed and sufficiently small. In this scenario, the \emph{\gls{ICSMC}} algorithm \citep{andrieu2010particle, chopin2015particle, andrieu2018uniform} has become a popular Monte Carlo method for approximating the joint smoothing distribution. The algorithm is based around a \emph{\gls{CSMC}} algorithm which builds a proposal distribution sequentially in the `time' direction by propagating $\nParticles+1$ Monte Carlo samples termed `particles' over the $\nTimeSteps$ time steps. One of these lineages is set equal to the current state of the Markov chain and termed the \emph{reference path.} At each time step, some of the remaining $\nParticles$ particle lineages are pruned out if they are unlikely represent good proposals (`selection'). The remaining particle lineages are multiplied and extended to the next time by sampling from the model dynamics (`mutation'). The selection steps prevent the algorithm from wasting computational effort on extending samples which are unlikely to form good proposals. This ensures that the complexity of the algorithm remains linear (and hence avoids a curse of dimension) in $\nTimeSteps$. Specifically, this linear complexity is due to the fact that the number of particles needs to be scaled as $\nParticles \sim \nTimeSteps$ \citep{andrieu2018uniform, lindsten2015uniform, delmoral2016particle, brown2021simple}. Recently, \citet{lee2020coupled} showed that the use of an extension known as \emph{backward sampling} \citep{whiteley2010particle} removes this need so that the overall complexity of the algorithm can be further reduced to $\bo(1)$ (for fixed $\nDimensions$), recalling that `complexity' is the number of likelihood evaluations needed to control approximation errors of fixed-dimensional marginals of the joint smoothing distribution. Empirically, this has also been found to hold for a related extension called \emph{ancestor sampling} \citep{lindsten2012ancestor}.


Due to this favourable scaling in $\nTimeSteps$, the \gls{ICSMC} algorithm has become a popular tool for Bayesian inference in low-dimensional state-space models (and beyond). For instance, it is the main ingredient within so-called \emph{particle Gibbs} samplers \citep{andrieu2010particle} which infer unknown model parameters alongside the latent states. 

Unfortunately, as we show in this work, the \gls{ICSMC} algorithm suffers from a curse of dimension in the `spatial' dimension $\nDimensions$ of the latent states. That is, for any time horizon $\nTimeSteps$, the algorithm breaks down if $\log(\nParticles) = \lo(\nDimensions)$ -- i.e.\ unless the number of particles grows exponentially in $D$ -- and this cannot be overcome through the use of backward sampling.

The main contribution of this work is to propose a novel \gls{CSMC} algorithm, called \emph{\gls{RWCSMC}} algorithm. In contrast to the (standard) \gls{CSMC} algorithm, it scatters the particles \emph{locally} around the reference path using Gaussian random-walk proposals whose variance is suitably scaled with $\nDimensions$. The algorithm is incorporated into a larger \emph{\gls{IRWCSMC}} algorithm which again induces a Markov kernel that leaves the joint smoothing distribution invariant. We prove that this strategy overcomes the curse of dimension in $\nDimensions$, i.e.\ in the sense that the \glsdesc{ESJD} associated with any $D$-dimensional time-marginal distribution is stable as $D \to \infty$. In other words, for any fixed $T$, the algorithm has complexity $\bo(D)$ (and the number of particles does not need to grow with $D$).

We also discuss the complexity in the time horizon $T$. Specifically, if the model factorises over time, we are able to verify that our proposed \gls{IRWCSMC} algorithm has the same scaling as the \gls{ICSMC} algorithm. That is, without backward sampling, we may grow the number of particles as $\nParticles = C \nTimeSteps$, for some constant $C > 0$, to guarantee an overall complexity $\bo(TD)$. The use of backward sampling again removes the need for growing $\nParticles$ with $\nTimeSteps$ so that the overall complexity can be brought down to $\bo(\nDimensions)$. Admittedly, the `factorisation-over-time' assumption is strong. However, we conjecture that the above-described scaling in $T$ holds more generally, i.e.\ -- just as in the \gls{ICSMC} algorithm -- this assumption is not necessary. As evidence for this, we present a slight modification of the \gls{IRWCSMC} algorithm based around the \emph{\gls{EHMM}} method from \citet{neal2003markov, neal2004inferring}, which we term the \emph{\gls{RWEHMM}} algorithm. Without making the `factorisation-over-time' assumption, we prove that this modified algorithm does not require scaling $N$ with $T$. 

Table~\ref{tab:complexity:csmc} summarises the complexity of the algorithms discussed in this work.

\begin{table}
 \centering
  \caption{Complexity of the algorithms in this work. `Complexity' is defined as the number of likelihood evaluations needed to control approximation errors of fixed-dimensional marginals. The {}\textsuperscript{$*$}-symbol indicates that the complexity in $T$ is only proved in for models that factorise over time in this work.}
 \begin{tabular}{lrr} \toprule
  & \multicolumn{2}{c}{With backward sampling?}\\ \cmidrule(lr){2-3}
  \parbox{3.4cm}{}                   & \parbox{2.5cm}{\raggedleft No} &  \parbox{2.5cm}{\raggedleft Yes} \\ \cmidrule(lr){2-2}\cmidrule(r){3-3}
  \gls{ICSMC} &  $\bo(\nTimeSteps \eul^{\nDimensions})$ & $\bo(\eul^\nDimensions)$ \\
  \gls{IRWCSMC}\textsuperscript{*}  &  $\bo(\nTimeSteps \nDimensions)$ & $\bo(\nDimensions)$ \\\cmidrule(lr){2-3}
  \gls{RWEHMM} & \multicolumn{2}{c}{$\bo(D)$}\\ 
  \bottomrule
 \end{tabular}
 \label{tab:complexity:csmc}
\end{table}

 \glsunset{HMM}

\subsection{Related work}

Our work can be viewed as an extension of high-dimensional scaling limits of classical \gls{MCMC} algorithms \citep[e.g.,][]{roberts1997weak}. This is because if $\nParticles=\nTimeSteps=1$, the \gls{IRWCSMC} update reduces to a \gls{MH} (or to Barker's) kernel \citep{metropolis1953equation, hastings1970monte, barker1965monte} with a \emph{random-walk} proposal. In contrast, the \gls{ICSMC} algorithm reduces to a \gls{MH} (or again to Barker's) algorithm with an \emph{independent} proposal (`independent' refers to the fact that the proposed value does not depend on the current state of the Markov chain) which is known to break down in high dimensions.

If $\nTimeSteps=1$ and $\nParticles > 1$, these algorithms can be viewed as a \gls{MH} (or Barker's) kernel with \emph{multiple} proposals. Such methods were introduced in the seminal works of \citet{tjelmeland2004using, neal2003markov}. Classical optimal scaling results were extended to a closely related class of \gls{MCMC} algorithms with multiple proposals in \citet{bedard2012scaling}. 

We limit our analysis to the \gls{IRWCSMC} algorithm. However, alternative ways of constructing (iterated) \gls{CSMC} algorithms with local moves are possible. Indeed, our work was motivated by \citet{shestopaloff2018sampling} who proposed the first such algorithm -- which, incidentally, reduces to a \gls{MH} (or Barker's) kernel with delayed acceptance \citep{christen2005markov} if $N = T = 1$. A generic framework which admits the \gls{ICSMC} algorithm, the \gls{IRWCSMC} algorithm, and the method from \citet{shestopaloff2018sampling} as special cases can be found in \citet[][Section~6]{finke2016embedded}. 

We have recently become aware of \citet[][Chapter~4]{malory2021bayesian} who independently analyse a related class of iterated \gls{CSMC} algorithms with exchangeable particle proposals that is likewise a special case of \citet[][Section~6]{finke2016embedded}. Our work distinguishes itself from theirs by, among others, the following contributions.
\begin{enumerate}
 \item We provide formal proof that the standard \gls{ICSMC} algorithm breaks down in high dimensions, even with backward sampling.
 
 \item Our dimensional-stability guarantees for the \gls{IRWCSMC} algorithm 
 hold even if the state-space model is dependent over time -- \citet{malory2021bayesian} assume that the target distribution factorises into a product of independent marginals over time.
 
 
 
 \item Our methodology and analysis permits a backward-sampling extension which is vital for performing inference for long time series.
 
\end{enumerate}


\glsreset{ESJD}
\subsection{Contributions and structure} 

This work is structured as follows.

\textbf{Section~\ref{sec:csmc}} reviews the \gls{ICSMC} algorithm and shows that it generalises classical \gls{MCMC} kernels with independent proposals. Our main result in this section is the following.
\begin{itemize}
 \item 
 \textit{Proposition~\ref{prop:limiting_csmc_algorithm}} proves that the \gls{ICSMC} algorithm suffers from a curse of dimension in the spatial dimension $D$ and that this cannot be overcome with backward sampling.
\end{itemize}
 
 
 \textbf{Section~\ref{sec:rwehmm}} introduces the novel \gls{RWEHMM} algorithm as a preliminary solution to the curse-of-dimensionality problem and as a precursor to our proposed \gls{IRWCSMC} algorithm. It does not employ resampling and therefore requires $\bo(N^2)$ operations to implement rather than $\bo(N)$ iterations. However, we introduce this algorithm here for didactic reasons because it is simple to understand and shares many features with our main \gls{IRWCSMC} algorithm. For instance, both algorithms scatter particles around the reference path using the same Gaussian random-walk type proposals which are scaled suitably with $\nDimensions$. Our main results in this section are the following.
    \begin{itemize}
 \item 
 \textit{Proposition~\ref{prop:limiting_rwehmm_algorithm}} and Proposition~\ref{prop:lower_bound_of_acceptance_rates_ehmm} prove that the \gls{RWEHMM} algorithm has stable acceptance rates in high dimensions.
 \item 
 \textit{Proposition~\ref{prop:esjd_rwehmm}} establishes a non-trivial limit for the \glsdesc{ESJD} in high dimensions.
 \item
 \textit{Corollary~\ref{cor:stability_of_acceptance_rates_ehmm}} verifies that the number of particles does not need to be scaled with $T$.
 \end{itemize}

 \textbf{Section~\ref{sec:rwcsmc}} introduces our novel \gls{IRWCSMC} algorithm, shows that it generalises classical \gls{MCMC} kernels with Gaussian random-walk proposals, and proves that it avoids the curse of dimension. Our main results in this section are the following.
 \begin{itemize}
  \item \textit{Propositions~\ref{prop:discrete_markov_kernel_rwcsmc_without_backward_sampling}~and~\ref{prop:discrete_markov_kernel_rwcsmc_with_backward_sampling}} show that the \gls{IRWCSMC} algorithm can be viewed as a `perturbed' version of the \gls{RWEHMM} algorithm.

 \item \textit{Proposition~\ref{prop:limiting_rwcsmc_algorithm}} and Corollary~\ref{cor:dimensional_stability_of_acceptance_rates_rwcsmc} prove that the \gls{IRWCSMC} algorithm has stable acceptance rates in high dimensions.
 
 \item \textit{Proposition~\ref{prop:esjd_rwcsmc}} establishes a non-trivial limit for the \glsdesc{ESJD} in high dimensions.
 
 \item \textit{Proposition~\ref{prop:stability_of_acceptance_rates}} verifies that, under the additional assumption that the model is independent over time, without backward sampling, $N$ must grow at least linearly in the time horizon $T$; with backward sampling, $N$ does not need to scale with $T$. We conjecture that this result holds more generally, i.e.\ that the `factorisation-in-time' assumption is not necessary.
 \end{itemize}
 Additionally, \textit{Remark~\ref{rem:relationship_with_unconditional_smc}} explains that whilst the (standard) \gls{CSMC} algorithm that underlies the \gls{ICSMC} algorithm is inextricably linked to the justification of standard `unconditional' \gls{SMC} counterpart, no such `unconditional' \gls{SMC} counterpart exists for the \gls{RWCSMC} algorithm that underlies the \gls{IRWCSMC} algorithm.

\textbf{Section~\ref{sec:numerical_illustration}} provides numerical illustration of our theoretical results. Most of our proofs and further materials can be found in the appendix. In particular, Appendix~\ref{app:sec:parameter_estimation} extends the proposed methodology deal with unknown `static' model parameters -- either via a particle-Gibbs type update or via another novel \gls{MCMC} kernel that is loosely related to correlated pseudo-marginal methods.

\subsection{Notation and conventions} 

Let $(\Omega, \calA, \ProbSymbol)$ be some probability space and denote expectation with respect to $\ProbSymbol$ by $\E$. The symbol $\dN(\mu,\varSigma)$ denotes a normal distribution with mean vector $\mu$ and covariance matrix $\varSigma$; $\dDirac_x$ is the point mass (Dirac measure) at $x$. Unless otherwise indicated, all (transition) densities mentioned in this work are absolutely continuous \WRT a suitable version of the Lebesgue measure.

For $n \in \naturals$, we often write $[n] \coloneqq \{1,2,\dotsc,n\}$ and $[n]_0 \coloneqq \{0,1,2,\dotsc,n\}$ and we let $\unitMat_{n} \in \reals^n$ and $\zeroMat_{n} \in\reals^n$ be a column vectors of length $n$ whose entries are all $1$ and $0$, respectively. The symbol $\iMat_n \in \reals^{n \times n}$ denotes the identity matrix.

Finally, for any $\nParticles \in \naturals$, $\particleIndex \in [\nParticles]_0$ and any $h^{1:\nParticles} \in \reals^\nParticles$, and with convention $h^0 \coloneqq 0$, we define the \emph{Boltzmann selection function}
\begin{align}
   \selectionFunctionBoltzmann{\particleIndex}(\{h^\particleIndexAlt\}_{\particleIndexAlt = 1}^\nParticles) \coloneqq  \dfrac{\exp(h^{\particleIndex})}{1 + \sum_{\particleIndexAlt = 1}^\nParticles \exp(h^\particleIndexAlt)},
\end{align}
as well as the \emph{Rosenbluth--Teller selection function}
\begin{align}
 \selectionFunctionRosenbluth{\particleIndex}(\{h^\particleIndexAlt\}_{\particleIndexAlt = 1}^\nParticles)  \coloneqq 
 \begin{cases}
   \dfrac{\exp(h^{\particleIndex})}{1 - 1 \wedge \exp(h^\particleIndex) + \sum_{\particleIndexAlt = 1}^\nParticles \exp(h^\particleIndexAlt)}, & \text{if $\particleIndex \in [\nParticles]$,}\\
   1 - \sum_{\particleIndexAlt = 1}^\nParticles
   \selectionFunctionRosenbluth{\particleIndexAlt}(\{h^\particleIndexAltAlt\}_{\particleIndexAltAlt = 1}^\nParticles)
   , & \text{if $\particleIndex = 0$.}
 \end{cases} \label{eq:selection_function_rosenbluth}
\end{align}

If $N = 1$, these selection functions reduce to the well-known acceptance functions of Barker's algorithm \citep{barker1965monte}: $\selectionFunctionBoltzmann{1} = \exp / (1 + \exp)$ and of the \gls{MH} algorithm \citep{metropolis1953equation, hastings1970monte}: $\selectionFunctionRosenbluth{1} = 1 \wedge \exp$. The following Peskun-ordering type result \citep{peskun1973optimum} (which follows immediately from the definition) shows that the Rosenbluth--Teller selection function induces a smaller rejection probability than the Boltzmann selection function.
\begin{lemma}\label{lem:peskun}
 For any $\nParticles \in \naturals$ and $h^{1:\nParticles} \in \reals^\nParticles$, $\selectionFunctionBoltzmann{0}(\{h^\particleIndexAlt\}_{\particleIndexAlt = 1}^\nParticles) \geq \selectionFunctionRosenbluth{0}(\{h^\particleIndexAlt\}_{\particleIndexAlt = 1}^\nParticles)$. \tendmark
\end{lemma}

\section{Existing methodology: the i-CSMC algorithm}
\label{sec:csmc}

\glsreset{ESJD}
\glsreset{CSMC}
\glsreset{ICSMC}

\subsection{Feynman--Kac model}
\label{subsec:feynman--kac_model}

For the measurable space $(\spaceState, \sigFieldState) \coloneqq (\reals^{\nDimensions}, \borel(\reals)^{\otimes \nDimensions})$, let $\Mutation_1 \in \probMeasSet(\spaceState)$ be a probability measure with density $\mutation_1 \colon \spaceState \to [0,\infty)$. Furthermore, for $\timeIndex \geq 2$, let $\Mutation_\timeIndex\colon \spaceState \times \sigFieldState \to [0,1]$ be some Markov kernel with transition density $\mutation_\timeIndex\colon \spaceState \times \spaceState \to [0,\infty)$. Furthermore, for $\timeIndex \geq 1$, let $\Potential_\timeIndex \colon \spaceState \to (0,\infty)$ be strictly positive measurable potential functions. The methods discussed in this work target the following probability measure on $(\spaceState_{T,D}, \sigFieldState_{T,D}) \coloneqq (\spaceState^\nTimeSteps, \sigFieldState^{\otimes \nTimeSteps})$:
\begin{align}
 \Target_{\nTimeSteps, \nDimensions}(\diff \state_{1:\nTimeSteps})
 & \propto \Mutation_1(\diff \state_1) \Potential_1(\state_1) \prod_{\timeIndex = 2}^\nTimeSteps \Mutation_\timeIndex(\state_{\timeIndex - 1}, \diff \state_\timeIndex) \Potential_\timeIndex(\state_\timeIndex).
\end{align}

\begin{example}[state-space model]
 Let $(\State_t, \mathbf{Y}_t)_{t \geq 1}$ be a Markov chain on some space $\spaceState \times \mathbf{F}$ with initial distribution  $\Mutation_1(\diff \state_1)\mathbf{H}_1(\state_1, \diff \mathbf{y}_1)$ and transition kernels $\Mutation_t(\state_{t-1}, \diff \state_t)\mathbf{H}_t(\state_t, \diff \mathbf{y}_t)$. Assume that for each time $t\in [T]$, we observe a realisation $\mathbf{y}_t \in \mathbf{F}$ of $\mathbf{Y}_t$ but $\State_t$ is unobserved (`latent'). We are then typically interested in computing (at least approximately) the posterior distribution of the latent `states' $\State_{1:T}$, often called the \emph{joint smoothing distribution}:
 \begin{align}
  \Target_{T,D}(\diff \state_{1:T}) = \Prob(\State_{1:T} \in \diff \state_{1:T}| \mathbf{Y}_{1:T} = \mathbf{y}_{1:T}).
 \end{align}
 Such a model, called \emph{state-space model} or \emph{(general-state) hidden Markov model}, can be viewed as a Feynman--Kac model by considering the observed values $\mathbf{y}_{1:T}$ to be `fixed' (so that they can be dropped from the notation) and assuming that $\Potential_t(\state_t) = \mathbf{h}_t(\state_t, \mathbf{y}_t)$, where $\mathbf{h}_t(\state_t, \ccdot)$ is a density of $\mathbf{H}_t(\state_t, \ccdot)$ w.r.t.\ to a suitable dominating measure. \tendmark
 \end{example}
Such models are routinely used in a wide variety of fields \citep{cappe2005inference}. Unfortunately, with the exception of a few special cases (e.g.\ state-space models that are both linear and Gaussian) the distribution $\Target_{T,D}$ is typically intractable and must be approximated, e.g.\ using \gls{MCMC} methods. It is therefore crucial to design $\Target_{T,D}$-invariant \gls{MCMC} kernels that can efficiently deal with long time horizons (large $T$) and large `spatial' dimension (large $D$) both of which are nowadays often found in the  models of interest to practitioners \citep[see, e.g.,][]{vanleeuwen2009particle, cressie2015statistics}.
 
\subsection{Description of the algorithm}
\label{subsec:csmc}

\subsubsection{Basic algorithm}

In the remainder of this section, we review the \gls{ICSMC} algorithm. For the moment, assume that the `spatial' dimension $\nDimensions$ is fixed (and not too large). For such scenarios, the \gls{ICSMC} algorithm \citep{andrieu2010particle} has become a popular way of constructing an efficient $\Target_{\nTimeSteps, \nDimensions}$-invariant Markov kernel. Specifically, the algorithm employs a collection of $\nParticles$ particles to construct a proposal for the entire state sequence sequentially in the `time' direction. Compared to updating the latent state sequence via an independent \gls{MH} kernel, this strategy brings down the computational complexity from $\bo(\eul^\nTimeSteps)$ to at most $\bo(\nTimeSteps)$ and thus avoids a curse of dimension in the time horizon $\nTimeSteps$.

For any $\timeIndex \in [\nTimeSteps]$ and any $\state_{1:\nTimeSteps} \in \spaceState_{T,D}$, define 
\begin{align}
 \logWeight_\timeIndex(\state_\timeIndex) &\coloneqq \log \Potential_\timeIndex(\state_\timeIndex).
\end{align}
The $l$th update of the \gls{ICSMC} scheme is then outlined in Algorithm~\ref{alg:iterated_csmc}, where we use the convention that any action described for the $\particleIndex$th particle index is to be performed conditionally independently for all $\particleIndex \in [\nParticles]$. We also use the convention that any quantity with time index $\timeIndex < 1$ is to be ignored, e.g.\ so that $\Mutation_1(\particle_{0}^\particleIndex, \ccdot) \equiv \Mutation_1(\ccdot)$.

\noindent\parbox{\textwidth}{
\begin{flushleft}
 \begin{framedAlgorithm}[\gls{ICSMC}] \label{alg:iterated_csmc} Given $\state_{1:T} \coloneqq \state_{1:\nTimeSteps}[l] \in \spaceState_{T,D}$.
 \begin{enumerate}
 \item \label{alg:iterated_csmc:1} For $\timeIndex \in [\nTimeSteps]$,
 \begin{enumerate}
  \item if $\timeIndex > 1$, 
  \begin{enumerate}
    \item set $A_{\timeIndex-1}^0 = a_{\timeIndex-1}^0 \coloneqq 0$,
  
    \item sample $A_{\timeIndex-1}^\particleIndex = a_{\timeIndex-1}^\particleIndex \in [\nParticles]_0$ with probability 
    \begin{align}
     \selectionFunctionBoltzmann{a_{\timeIndex-1}^\particleIndex}(\{\logWeight_{\timeIndex-1}(\particle_{\timeIndex-1}^\particleIndexAlt) - \logWeight_{\timeIndex-1}(\particle_{\timeIndex-1}^0)\}_{\particleIndexAlt = 1}^\nParticles)
     & = 
     \smash{\dfrac{\smash{\Potential_{\timeIndex-1}(\particle_{\timeIndex-1}^{a_{\timeIndex-1}^{\smash{\particleIndex}}})}}{\sum_{\particleIndexAlt = 0}^\nParticles \Potential_{\timeIndex-1}(\particle_{\timeIndex-1}^\particleIndexAlt)},}
    \end{align}
  \end{enumerate}
       
  \item set $\smash{\Particle_\timeIndex^0 = \particle_\timeIndex^0 \coloneqq \state_\timeIndex}$ and sample $\smash{\Particle_\timeIndex^\particleIndex = \particle_\timeIndex^\particleIndex \sim \Mutation_\timeIndex(\particle_{\timeIndex-1}^{a_{\timeIndex-1}^\particleIndex}, \ccdot)}$.
  \end{enumerate}
  
  \item \label{alg:iterated_csmc:2a} Sample $\OutputParticleIndex_\nTimeSteps =\outputParticleIndex_\nTimeSteps \in [\nParticles]_0$ with probability
  \begin{align}
    \selectionFunctionBoltzmann{\outputParticleIndex_\nTimeSteps}(\{\logWeight_{\nTimeSteps}(\particle_{\nTimeSteps}^\particleIndexAlt) - \logWeight_\nTimeSteps(\particle_{\nTimeSteps}^0)\}_{\particleIndexAlt = 1}^\nParticles)
    & = 
    \smash{\dfrac{\Potential_{\nTimeSteps}(\particle_{\nTimeSteps}^{\outputParticleIndex_\nTimeSteps})}{\sum_{\particleIndexAlt = 0}^\nParticles \Potential_{\nTimeSteps}(\particle_{\nTimeSteps}^\particleIndexAlt)}}.
  \end{align}
    
  \item \label{alg:iterated_csmc:2b} Set $\smash{\OutputParticleIndex_\timeIndex =\outputParticleIndex_\timeIndex \coloneqq a_\timeIndex^{\outputParticleIndex_{\timeIndex+1}}}$, for $t = T-1,\dotsc,1$.

  \item \label{alg:iterated_csmc:3} Set $\smash{\State_{1:\nTimeSteps}' \coloneqq \state_{1:T}' \coloneqq (\particle_1^{\outputParticleIndex_1}, \dotsc, \particle_\nTimeSteps^{\outputParticleIndex_\nTimeSteps})}$. 
  
  \item \label{alg:iterated_csmc:4} Return $\smash{\state_{1:\nTimeSteps}[l+1] \coloneqq \state_{1:T}'}$. 
 \end{enumerate}
\end{framedAlgorithm}
\end{flushleft}
}

\glsreset{CSMC}
Step~\ref{alg:iterated_csmc:1} of Algorithm~\ref{alg:iterated_csmc} which (a) performs (conditional) multinomial resampling by drawing the ancestor indices $A_t^n$ and (b) generates the particles $\Particle_t^n$, is known as the \emph{\gls{CSMC}} algorithm.

The following running example illustrates how the algorithms discussed in this work reduce to versions of well known classical \gls{MCMC} kernels if $\nParticles=\nTimeSteps=1$.

\begin{example}[classical \gls{MCMC} kernels]
 If $\nTimeSteps=1$ and $N=1$, the target distribution is given by $\Target_{1, \nDimensions}(\diff \state_1) \propto \Mutation_1(\diff \state_1)  \Potential_1(\state_1)$ and Algorithm~\ref{alg:iterated_csmc} proposes $\Particle_1^1 = \particle_1^1 \sim \Mutation_1$ and accepts this proposal as the new state of the Markov chain with probability
 \begin{align}
  \selectionFunctionBoltzmann{1}(\logWeight_{1}(\particle_{1}^1) - \logWeight_1(\particle_{1}^0))
   =
  \frac{\Potential_1(\particle_1^1)}{\Potential_1(\particle_1^0) + \Potential_1(\particle_1^1)},
 \end{align}
 where $\particle_1^0 = \state_1$. This can be recognised as a version of Barker's kernel \citep{barker1965monte} with independence proposal $\Mutation_1$ (in the sense that the proposed state does not depend on the current state). \tendmark
\end{example}

\begin{example}[multi-proposal \gls{MCMC} kernels]
 If $T = 1$ but $N > 1$, Algorithm~\ref{alg:iterated_csmc} (termed \emph{conditional sampling--importance resampling} in \citealt{andrieu2018uniform}) reduces to an \gls{MCMC} algorithm with multiple proposals (all being independent of each other and of the current state of the Markov chain). Multi-proposal \gls{MCMC} algorithms were introduced in \citet{neal2003markov, tjelmeland2004using, frenkel2004speed};  \citet{delmas2009does, yang2017parallelizable, schwedes2018quasi} analyse Rao--Blackwellisation strategies for re-using all $N$ proposed samples to estimate expectations of interest. \tendmark 
\end{example}

\subsubsection{Extensions}
\label{subsec:csmc_extensions}

\paragraph*{Forced move} 
To reduce the probability of sampling $\OutputParticleIndex_\nTimeSteps =\outputParticleIndex_\nTimeSteps = 0$ in Step~\ref{alg:iterated_csmc:2a} of Algorithm~\ref{alg:iterated_csmc} (and hence improve the \gls{ICSMC} kernel in the Peskun order -- see Lemma~\ref{lem:peskun}) \citet{chopin2015particle} proposed to replace the Boltzmann selection function in Step~\ref{alg:iterated_csmc:2a} by the Rosenbluth--Teller selection function, i.e.\ instead sample $\OutputParticleIndex_\nTimeSteps= k_\nTimeSteps \neq 0$ with probability
\begin{align}
  \MoveEqLeft 
  \selectionFunctionRosenbluth{\outputParticleIndex_\nTimeSteps}(\{\logWeight_{\nTimeSteps}(\particle_{\nTimeSteps}^\particleIndexAlt) - \logWeight_\nTimeSteps(\particle_{\nTimeSteps}^0)\}_{\particleIndexAlt = 1}^\nParticles)
  = \frac{\Potential_\nTimeSteps(\particle_\nTimeSteps^{\outputParticleIndex_\nTimeSteps})}{\sum_{\particleIndexAlt=1}^\nParticles \Potential_\nTimeSteps(\particle_\nTimeSteps^\particleIndexAlt) - \Potential_\nTimeSteps(\particle_\nTimeSteps^{\outputParticleIndex_\nTimeSteps}) \wedge \Potential_\nTimeSteps(\particle_\nTimeSteps^0)}.
\end{align}
This so called \emph{forced move} approach can be recognised as an application of the \emph{modified discrete-state Gibbs sampler} kernel from \citet{liu1996peskun}. See also \citet{tjelmeland2004using} for an iterative algorithm for optimising the selection function.

\begin{example}[classical \gls{MCMC} kernels, continued]
 With the forced-move extension, Step~\ref{alg:iterated_csmc:2a} of Algorithm~\ref{alg:iterated_csmc} accepts $\Particle_1^1 = \particle_1^1 \sim \Mutation_1$ with probability
 \begin{align}
    \selectionFunctionRosenbluth{1}(\logWeight_{1}(\particle_{1}^1) - \logWeight_1(\particle_{1}^0))
    =
  1 \wedge \frac{\Potential_1(\particle_1^1)}{\Potential_1(\particle_1^0)},
 \end{align}
 where $\particle_1^0 = \state_1$. This can be recognised as a version of an independent \gls{MH} kernel \citep{metropolis1953equation, hastings1970monte}. \tendmark
\end{example}

\paragraph*{Backward sampling}
Steps~\ref{alg:iterated_csmc:2a} and \ref{alg:iterated_csmc:2b} of Algorithm~\ref{alg:iterated_csmc} sample a final-time particle index $\OutputParticleIndex_\nTimeSteps$ and then trace back its ancestral lineage. This limits the new state $\state_{1:T}[l+1]$ to one of the $\nParticles+1$ particle lineages generated under the \gls{CSMC} algorithm in Step~\ref{alg:iterated_csmc:1} which often
coalesce with the old reference path $\state_{1:T}[l]$. To ensure good mixing, we must therefore control the probability of such coalescence events by growing $N$ linearly with $T$. This can be costly if $T$ is large. To circumvent this problem, the\emph{backward-sampling} extension \citep{whiteley2010particle} instead samples $\OutputParticleIndex_\timeIndex =\outputParticleIndex_\timeIndex \in [\nParticles]_0$ in Step~\ref{alg:iterated_csmc:2b} with probability
\begin{align}
  \selectionFunctionBoltzmann{\outputParticleIndex_\timeIndex}(\{
  \logBackwardWeight_\timeIndex(\particle_{\timeIndex}^\particleIndexAlt, \particle_{\timeIndex+1}^{\outputParticleIndex_{\timeIndex+1}}) - \logBackwardWeight_\timeIndex(\particle_{\timeIndex}^0, \particle_{\timeIndex+1}^{\outputParticleIndex_{\timeIndex+1}})
  \}_{\particleIndexAlt = 1}^\nParticles)
  & =
  \frac{\Potential_\timeIndex(\particle_\timeIndex^{\outputParticleIndex_\timeIndex}) \mutation_{\timeIndex+1}(\particle_\timeIndex^{\outputParticleIndex_\timeIndex}, \particle_{\timeIndex+1}^{\outputParticleIndex_{\timeIndex+1}})}{\sum_{\particleIndexAlt=0}^\nParticles \Potential_\timeIndex(\particle_\timeIndex^\particleIndexAlt) \mutation_{\timeIndex+1}(\particle_\timeIndex^\particleIndexAlt, \particle_{\timeIndex+1}^{\outputParticleIndex_{\timeIndex+1}})},
\end{align}
for $t = T-1, \dotsc, 1$, where we have defined 
\begin{align}
 \logBackwardWeight_\timeIndex(\state_{\timeIndex:\timeIndex + 1}) &\coloneqq \logWeight_\timeIndex(\state_\timeIndex) + \log \mutation_{\timeIndex+1}(\state_\timeIndex, \state_{\timeIndex+1}).
\end{align}
\citet{lee2020coupled} show that backward sampling allows us to keep $N$ constant in $\nTimeSteps$ thus reducing the complexity of the algorithm from $\bo(\nTimeSteps)$ to $\bo(1)$. A closely related method, \emph{ancestor sampling,} was proposed in \citet{lindsten2012ancestor}.

\paragraph*{Further extensions}
Step~\ref{alg:iterated_csmc:1} of Algorithm~\ref{alg:iterated_csmc} proposes particles from the `prior' $\Mutation_t(\state_{t-1}, \ccdot)$ and draws the parent indices $A_{t-1}^n$ via (conditional) multinomial resampling. Other proposal kernels \citep{doucet2000sequential}, other resampling schemes \citep{douc2005comparison} or even auxiliary particle filter ideas \citep{pitt1999filtering, shestopaloff2019replica} could be employed. However, since none of these extensions overcome the curse of dimension proved below, we refrain from including them here to keep the presentation simple.

\subsubsection{Induced \texorpdfstring{$\Target_{\nTimeSteps,\nDimensions}$-}{}invariant Markov kernel}


Given $\smash{\State_{1:T} = \state_{1:\nTimeSteps} = \state_{1:\nTimeSteps}[l]}$, let
\begin{align}
 \smash{\IteratedCsmcKernel{\nTimeSteps,\nDimensions,\state_{1:T}}{\nParticles}(\diff \particle_{1:\nTimeSteps} \times \diff a_{1:\nTimeSteps-1} \times \diff\outputParticleIndex_{1:\nTimeSteps} \times \diff \state_{1:T}')}
\end{align}
be the law of all the random variables $\smash{(\Particle_{1:T}, A_{1:T-1}, K_{1:T}, \State_{1:T}')}$ generated in Steps~\ref{alg:iterated_csmc:1}--\ref{alg:iterated_csmc:3} of Algorithm~\ref{alg:iterated_csmc} (with or without the forced-move extension and with or without backward sampling). Appendix~\ref{app:subsec:joint_law_csmc} gives a formal definition of this law. 

Let $\smash{\ExpectationCsmcKernel{\nTimeSteps,\nDimensions,\state_{1:\nTimeSteps}}{\nParticles}}$ denote expectation \WRT $\smash{\IteratedCsmcKernel{\nTimeSteps,\nDimensions,\state_{1:T}}{\nParticles}}$. Algorithm~\ref{alg:iterated_csmc} induces a Markov kernel 
\begin{align}
 \InducedIteratedCsmcKernel{\nTimeSteps,\nDimensions}{\nParticles}(\state_{1:\nTimeSteps}, \diff \state_{1:T}') 
 & \coloneqq \ExpectationCsmcKernel{\nTimeSteps,\nDimensions,\state_{1:\nTimeSteps}}{\nParticles}[\ind\{\State_{1:T}' \in \diff \state_{1:T}'\}],
 \label{eq:iterated_csmc_kernel}
\end{align}
for $(\state_{1:\nTimeSteps},\diff \state_{1:T}') \in \spaceState_{T,D} \times \sigFieldState_{T,D}$. The following proposition shows that this Markov kernel leaves $\Target_{\nTimeSteps, \nDimensions}$ invariant. It was proved \citet{andrieu2010particle} for the basic algorithm and in \citet{chopin2015particle, whiteley2010particle} for the forced-move and backward-sampling extensions. For completeness, we provide an alternative, simple proof in Appendix~\ref{app:subsec:csmc_invariance}.

\begin{proposition}\label{prop:csmc_invariance}
 For any $N, T, D \in \naturals$, $\smash{\Target_{T,D}\InducedIteratedCsmcKernel{\nTimeSteps,\nDimensions}{\nParticles} = \Target_{T,D}}$. \tendmark
\end{proposition}

For any $t \in [T]$, we call
\begin{align}
 \smash{\acceptanceRate{\nTimeSteps,\nDimensions,\state_{1:\nTimeSteps}}{\nParticles}(t) 
  \coloneqq
   \ExpectationCsmcKernel{\nTimeSteps,\nDimensions,\state_{1:\nTimeSteps}}{\nParticles}[ \ind\{\OutputParticleIndex_{\timeIndex} \neq 0\}]}
\end{align}
the \emph{acceptance rate at time~$t$} associated with Algorithm~\ref{alg:iterated_csmc}. This name is justified because $K_t = 0$ in Algorithm~\ref{alg:iterated_csmc} implies $\state_t' = \state_t$, i.e.\ $\state_{t}[l+1] = \state_{t}[l]$.

\subsection{Curse of dimension}
\label{subsec:csmc_scaling}

\subsubsection{High-dimensional regime}
\label{subsec:high-dimensional_regime}

We now prove that the \gls{ICSMC} algorithm suffers from a curse of dimension. This is established for a special case of the Feynman--Kac model from Section~\ref{subsec:feynman--kac_model} which factorises into $\nDimensions$ \gls{IID} `spatial' components. Most other theoretical results in this work will be established under this regime. However, we stress that none of the algorithms discussed in this work are limited to this \gls{IID} setting.

\begin{enumerate}[label=\textbf{A\arabic*}, series=model_assumptions]
  \item \label{as:iid_model} The mutation kernels and potential functions factorise as
\begin{align}
 \Mutation_\timeIndex(\state_{\timeIndex - 1}, \diff \state_\timeIndex) = \prod_{\dimensionIndex = 1}^\nDimensions \MutationSingle_\timeIndex(\stateSingle_{\timeIndex-1,\dimensionIndex}, \diff \stateSingle_{\timeIndex,\dimensionIndex}) \quad \text{and} \quad \Potential_\timeIndex(\state_\timeIndex) = \prod_{\dimensionIndex = 1}^\nDimensions \PotentialSingle_\timeIndex(\stateSingle_{\timeIndex,\dimensionIndex}),
\end{align}
with the convention that any quantity with time index~$0$ is to be ignored and where
\begin{itemize}
 \item $\smash{\state_{\timeIndex} = \stateSingle_{\timeIndex, 1:D} \in \spaceState}$, recalling that $\smash{(\spaceState, \sigFieldState) = (\spaceStateSingle^{D},  \sigFieldStateSingle^{\otimes D})}$;
 \item $\MutationSingle_1 \in \probMeasSet(\spaceStateSingle)$ is a probability measure with density $\mutationSingle_1 \colon \spaceStateSingle \to [0,\infty)$ and, for $\timeIndex \geq 2$, $\MutationSingle_\timeIndex\colon \spaceStateSingle \times \sigFieldStateSingle \to [0,1]$ is a Markov kernel with transition density $\mutationSingle_\timeIndex\colon \spaceStateSingle^2 \to [0,\infty)$; 
 
 \item $\PotentialSingle_\timeIndex \colon \spaceStateSingle \to (0,\infty)$ is a strictly positive and measurable potential function. \tendmark
\end{itemize} 
\end{enumerate}
Thus, under \ref{as:iid_model}, $\Target_{\nTimeSteps,\nDimensions} = \TargetSingle_\nTimeSteps^{\otimes \nDimensions}$, with the following probability measure on $\reals^T$:
\begin{align}
 \TargetSingle_\nTimeSteps(\diff \stateSingle_{1:\nTimeSteps})
 & \propto \MutationSingle_1(\diff \stateSingle_1) \PotentialSingle_1(\stateSingle_1) \prod_{\timeIndex = 2}^\nTimeSteps \MutationSingle_\timeIndex(\stateSingle_{\timeIndex - 1}, \diff \stateSingle_\timeIndex) \PotentialSingle_\timeIndex(\stateSingle_\timeIndex).
\end{align}

\subsubsection{Convergence to a degenerate limit}

We now show that in high (`spatial') dimensions, the law of genealogies under the \gls{ICSMC} algorithm converges to limit that is degenerate in the sense that all particle lineages immediately coalesce with the reference path so that all acceptance probabilities are zero. Once could na\"ively hope that backward sampling circumvents this problem because it draws a new reference path that is not confined to one of the $N+1$ surviving lineages. Unfortunately, our analysis shows that backward sampling, too, returns the old reference path in high dimensions. Typical behaviour of the genealogies is illustrated in Figure~\ref{fig:breakdown_of_csmc}.


\begin{figure}[!htb]
  \begin{subfigure}[b]{1\linewidth}
  \centering
    \begin{tikzpicture}
    
    \begin{scope}[xshift=0cm]
    \scriptsize
    \matrix (a) [matrix of math nodes, row sep=0.2em, column sep=1.5em, nodes={anchor = center, inner sep=1pt}]
    {
      \particle_1^6 & \particle_2^6 & \particle_3^6 & \particle_4^6 & \particle_5^6\\
      \particle_1^5 & \particle_2^5 & \particle_3^5 & \particle_4^5 & \particle_5^5\\
      \particle_1^4 & \particle_2^4 & \particle_3^4 & \particle_4^4 & \particle_5^4\\
      \particle_1^3 & \particle_2^3 & \particle_3^3 & \particle_4^3 & \particle_5^3\\
      \particle_1^2 & \particle_2^2 & \particle_3^2 & \particle_4^2 & \particle_5^2\\
      \particle_1^1 & \particle_2^1 & \particle_3^1 & \particle_4^1 & \particle_5^1\\
      \particle_1^0 & \particle_2^0 & \particle_3^0 & \particle_4^0 & \particle_5^0\\
    };
    
    \draw[solid, color = red, opacity = 0.2, line width = 3pt](a-7-1) to (a-7-2);
    \draw[solid, color = red, opacity = 0.2, line width = 3pt](a-7-2) to (a-7-3);
    \draw[solid, color = red, opacity = 0.2, line width = 3pt](a-7-3) to (a-7-4);
    \draw[solid, color = red, opacity = 0.2, line width = 3pt](a-7-4) to (a-7-5);
    
    \draw[solid, color = blue, opacity = 0.2, line width = 3pt](a-4-1) to (a-3-2);
    \draw[solid, color = blue, opacity = 0.2, line width = 3pt](a-3-2) to (a-3-3);
    \draw[solid, color = blue, opacity = 0.2, line width = 3pt](a-3-3) to (a-2-4);
    \draw[solid, color = blue, opacity = 0.2, line width = 3pt](a-2-4) to (a-3-5);

    \draw[solid, opacity = 1](a-7-1) to (a-7-2);
    \draw[solid, opacity = 1](a-7-2) to (a-7-3);
    \draw[solid, opacity = 1](a-7-3) to (a-7-4);
    \draw[solid, opacity = 1](a-7-4) to (a-7-5);

    \draw[solid, opacity = 1](a-5-1) to (a-6-2);
    \draw[dotted, opacity = 1](a-5-1) to (a-5-2);
    \draw[dotted, opacity = 1](a-4-1) to (a-4-2);
    \draw[solid, opacity = 1](a-4-1) to (a-3-2);
    \draw[dotted, opacity = 1](a-2-1) to (a-2-2);
    \draw[solid, opacity = 1](a-2-1) to (a-1-2);
    
    \draw[dotted, opacity = 1](a-7-2) to (a-6-3);
    \draw[solid, opacity = 1](a-6-2) to (a-5-3);
    \draw[dotted, opacity = 1](a-6-2) to (a-4-3);
    \draw[solid, opacity = 1](a-3-2) to (a-3-3);
    \draw[dotted, opacity = 1](a-1-2) to (a-2-3);
    \draw[solid, opacity = 1](a-1-2) to (a-1-3);
    
    \draw[dotted, opacity = 1](a-7-3) to (a-6-4);
    \draw[solid, opacity = 1](a-5-3) to (a-5-4);
    \draw[dotted, opacity = 1](a-4-3) to (a-4-4);
    \draw[dotted, opacity = 1](a-3-3) to (a-3-4);
    \draw[solid, opacity = 1](a-3-3) to (a-2-4);
    \draw[solid, opacity = 1](a-1-3) to (a-1-4);
    
    \draw[solid, opacity = 1](a-5-4) to (a-6-5);
    \draw[solid, opacity = 1](a-5-4) to (a-5-5);
    \draw[solid, opacity = 1](a-5-4) to (a-4-5);
    \draw[solid, opacity = 1](a-2-4) to (a-3-5);
    \draw[solid, opacity = 1](a-2-4) to (a-2-5);
    \draw[solid, opacity = 1](a-1-4) to (a-1-5);
    \end{scope}


    \begin{scope}[xshift=4cm]
    \scriptsize
    
    \matrix (b) [matrix of math nodes, row sep=0.2em, column sep=1.5em, nodes={anchor = center, inner sep=1pt}]
    {
      \particle_1^6 & \particle_2^6 & \particle_3^6 & \particle_4^6 & \particle_5^6\\
      \particle_1^5 & \particle_2^5 & \particle_3^5 & \particle_4^5 & \particle_5^5\\
      \particle_1^4 & \particle_2^4 & \particle_3^4 & \particle_4^4 & \particle_5^4\\
      \particle_1^3 & \particle_2^3 & \particle_3^3 & \particle_4^3 & \particle_5^3\\
      \particle_1^2 & \particle_2^2 & \particle_3^2 & \particle_4^2 & \particle_5^2\\
      \particle_1^1 & \particle_2^1 & \particle_3^1 & \particle_4^1 & \particle_5^1\\
      \particle_1^0 & \particle_2^0 & \particle_3^0 & \particle_4^0 & \particle_5^0\\
    };
    
    \draw[solid, color = red, opacity = 0.2, line width = 3pt](b-7-1) to (b-7-2);
    \draw[solid, color = red, opacity = 0.2, line width = 3pt](b-7-2) to (b-7-3);
    \draw[solid, color = red, opacity = 0.2, line width = 3pt](b-7-3) to (b-7-4);
    \draw[solid, color = red, opacity = 0.2, line width = 3pt](b-7-4) to (b-7-5);
    
    \draw[solid, color = blue, opacity = 0.2, line width = 3pt](b-7-1) to (b-7-2);
    \draw[solid, color = blue, opacity = 0.2, line width = 3pt](b-7-2) to (b-7-3);
    \draw[solid, color = blue, opacity = 0.2, line width = 3pt](b-7-3) to (b-2-4);
    \draw[solid, color = blue, opacity = 0.2, line width = 3pt](b-2-4) to (b-2-5);

    \draw[solid, opacity = 1](b-7-1) to (b-7-2);
    \draw[solid, opacity = 1](b-7-2) to (b-7-3);
    \draw[solid, opacity = 1](b-7-3) to (b-7-4);
    \draw[solid, opacity = 1](b-7-4) to (b-7-5);

    \draw[dotted, opacity = 1](b-7-1) to (b-6-2);
    \draw[dotted, opacity = 1](b-7-1) to (b-5-2);
    \draw[dotted, opacity = 1](b-7-1) to (b-4-2);
    \draw[dotted, opacity = 1](b-3-1) to (b-3-2);
    \draw[dotted, opacity = 1](b-3-1) to (b-2-2);
    \draw[dotted, opacity = 1](b-3-1) to (b-1-2);
    
    \draw[dotted, opacity = 1](b-7-2) to (b-6-3);
    \draw[dotted, opacity = 1](b-7-2) to (b-5-3);
    \draw[dotted, opacity = 1](b-7-2) to (b-4-3);
    \draw[dotted, opacity = 1](b-7-2) to (b-3-3);
    \draw[dotted, opacity = 1](b-1-2) to (b-2-3);
    \draw[dotted, opacity = 1](b-1-2) to (b-1-3);
    
    \draw[dotted, opacity = 1](b-7-3) to (b-6-4);
    \draw[dotted, opacity = 1](b-7-3) to (b-5-4);
    \draw[dotted, opacity = 1](b-7-3) to (b-4-4);
    \draw[dotted, opacity = 1](b-7-3) to (b-3-4);
    \draw[solid, opacity = 1](b-7-3) to (b-2-4);
    \draw[dotted, opacity = 1](b-3-3) to (b-1-4);
    
    \draw[solid, opacity = 1](b-7-4) to (b-6-5);
    \draw[solid, opacity = 1](b-7-4) to (b-5-5);
    \draw[solid, opacity = 1](b-7-4) to (b-4-5);
    \draw[solid, opacity = 1](b-7-4) to (b-3-5);
    \draw[solid, opacity = 1](b-2-4) to (b-2-5);
    \draw[solid, opacity = 1](b-2-4) to (b-1-5);

    \end{scope}
    \begin{scope}[xshift=8cm]
    \scriptsize
    \matrix (c) [matrix of math nodes, row sep=0.2em, column sep=1.5em, nodes={anchor = center, inner sep=1pt}]
    {
      \particle_1^6 & \particle_2^6 & \particle_3^6 & \particle_4^6 & \particle_5^6\\
      \particle_1^5 & \particle_2^5 & \particle_3^5 & \particle_4^5 & \particle_5^5\\
      \particle_1^4 & \particle_2^4 & \particle_3^4 & \particle_4^4 & \particle_5^4\\
      \particle_1^3 & \particle_2^3 & \particle_3^3 & \particle_4^3 & \particle_5^3\\
      \particle_1^2 & \particle_2^2 & \particle_3^2 & \particle_4^2 & \particle_5^2\\
      \particle_1^1 & \particle_2^1 & \particle_3^1 & \particle_4^1 & \particle_5^1\\
      \particle_1^0 & \particle_2^0 & \particle_3^0 & \particle_4^0 & \particle_5^0\\
    };
    
    \draw[solid, color = red, opacity = 0.2, line width = 3pt](c-7-1) to (c-7-2);
    \draw[solid, color = red, opacity = 0.2, line width = 3pt](c-7-2) to (c-7-3);
    \draw[solid, color = red, opacity = 0.2, line width = 3pt](c-7-3) to (c-7-4);
    \draw[solid, color = red, opacity = 0.2, line width = 3pt](c-7-4) to (c-7-5);
    
    \draw[solid, color = blue, opacity = 0.2, line width = 3pt](c-7-1) to (c-7-2);
    \draw[solid, color = blue, opacity = 0.2, line width = 3pt](c-7-2) to (c-7-3);
    \draw[solid, color = blue, opacity = 0.2, line width = 3pt](c-7-3) to (c-7-4);
    \draw[solid, color = blue, opacity = 0.2, line width = 3pt](c-7-4) to (c-7-5);

    \draw[solid, opacity = 1](c-7-1) to (c-7-2);
    \draw[solid, opacity = 1](c-7-2) to (c-7-3);
    \draw[solid, opacity = 1](c-7-3) to (c-7-4);
    \draw[solid, opacity = 1](c-7-4) to (c-7-5);

    \draw[dotted, opacity = 1](c-7-1) to (c-6-2);
    \draw[dotted, opacity = 1](c-7-1) to (c-5-2);
    \draw[dotted, opacity = 1](c-7-1) to (c-4-2);
    \draw[dotted, opacity = 1](c-7-1) to (c-3-2);
    \draw[dotted, opacity = 1](c-7-1) to (c-2-2);
    \draw[dotted, opacity = 1](c-7-1) to (c-1-2);
    
    \draw[dotted, opacity = 1](c-7-2) to (c-6-3);
    \draw[dotted, opacity = 1](c-7-2) to (c-5-3);
    \draw[dotted, opacity = 1](c-7-2) to (c-4-3);
    \draw[dotted, opacity = 1](c-7-2) to (c-3-3);
    \draw[dotted, opacity = 1](c-7-2) to (c-2-3);
    \draw[dotted, opacity = 1](c-7-2) to (c-1-3);
    
    \draw[dotted, opacity = 1](c-7-3) to (c-6-4);
    \draw[dotted, opacity = 1](c-7-3) to (c-5-4);
    \draw[dotted, opacity = 1](c-7-3) to (c-4-4);
    \draw[dotted, opacity = 1](c-7-3) to (c-3-4);
    \draw[dotted, opacity = 1](c-7-3) to (c-2-4);
    \draw[dotted, opacity = 1](c-7-3) to (c-1-4);
    
    \draw[solid, opacity = 1](c-7-4) to (c-6-5);
    \draw[solid, opacity = 1](c-7-4) to (c-5-5);
    \draw[solid, opacity = 1](c-7-4) to (c-4-5);
    \draw[solid, opacity = 1](c-7-4) to (c-3-5);
    \draw[solid, opacity = 1](c-7-4) to (c-2-5);
    \draw[solid, opacity = 1](c-7-4) to (c-1-5);
    \end{scope}


    
    \node (low) [above of = a, yshift = 2em, font = \footnotesize] {Low dimensions};
    \node (moderate) [above of = b, yshift = 2em, font = \footnotesize] {Moderate dimensions};
    \node (high) [above of = c, yshift = 2em, font = \footnotesize] {High dimensions};

  \end{tikzpicture}
    \caption{Without backward sampling.}
  \end{subfigure}\\
  \begin{subfigure}[b]{1\linewidth}
  \centering
    \begin{tikzpicture}
    \begin{scope}[xshift=0cm]
    \scriptsize
    \matrix (a) [matrix of math nodes, row sep=0.2em, column sep=1.5em, nodes={anchor = center, inner sep=1pt}]
    {
      \particle_1^6 & \particle_2^6 & \particle_3^6 & \particle_4^6 & \particle_5^6\\
      \particle_1^5 & \particle_2^5 & \particle_3^5 & \particle_4^5 & \particle_5^5\\
      \particle_1^4 & \particle_2^4 & \particle_3^4 & \particle_4^4 & \particle_5^4\\
      \particle_1^3 & \particle_2^3 & \particle_3^3 & \particle_4^3 & \particle_5^3\\
      \particle_1^2 & \particle_2^2 & \particle_3^2 & \particle_4^2 & \particle_5^2\\
      \particle_1^1 & \particle_2^1 & \particle_3^1 & \particle_4^1 & \particle_5^1\\
      \particle_1^0 & \particle_2^0 & \particle_3^0 & \particle_4^0 & \particle_5^0\\
    };
    
    \draw[solid, color = red, opacity = 0.2, line width = 3pt](a-7-1) to (a-7-2);
    \draw[solid, color = red, opacity = 0.2, line width = 3pt](a-7-2) to (a-7-3);
    \draw[solid, color = red, opacity = 0.2, line width = 3pt](a-7-3) to (a-7-4);
    \draw[solid, color = red, opacity = 0.2, line width = 3pt](a-7-4) to (a-7-5);

    \draw[solid, color = blue, opacity = 0.2, line width = 3pt](a-2-1) to (a-3-2);
    \draw[solid, color = blue, opacity = 0.2, line width = 3pt](a-3-2) to (a-3-3);
    \draw[solid, color = blue, opacity = 0.2, line width = 3pt](a-3-3) to (a-3-4);
    \draw[solid, color = blue, opacity = 0.2, line width = 3pt](a-3-4) to (a-3-5);
    
    \draw[solid, opacity = 1](a-7-1) to (a-7-2);
    \draw[solid, opacity = 1](a-7-2) to (a-7-3);
    \draw[solid, opacity = 1](a-7-3) to (a-7-4);
    \draw[solid, opacity = 1](a-7-4) to (a-7-5);

    \draw[solid, opacity = 1](a-5-1) to (a-6-2);
    \draw[dotted, opacity = 1](a-5-1) to (a-5-2);
    \draw[dotted, opacity = 1](a-4-1) to (a-4-2);
    \draw[solid, opacity = 1](a-4-1) to (a-3-2);
    \draw[dotted, opacity = 1](a-2-1) to (a-2-2);
    \draw[solid, opacity = 1](a-2-1) to (a-1-2);
    
    \draw[dotted, opacity = 1](a-7-2) to (a-6-3);
    \draw[solid, opacity = 1](a-6-2) to (a-5-3);
    \draw[dotted, opacity = 1](a-6-2) to (a-4-3);
    \draw[solid, opacity = 1](a-3-2) to (a-3-3);
    \draw[dotted, opacity = 1](a-1-2) to (a-2-3);
    \draw[solid, opacity = 1](a-1-2) to (a-1-3);
    
    \draw[dotted, opacity = 1](a-7-3) to (a-6-4);
    \draw[solid, opacity = 1](a-5-3) to (a-5-4);
    \draw[dotted, opacity = 1](a-4-3) to (a-4-4);
    \draw[dotted, opacity = 1](a-3-3) to (a-3-4);
    \draw[solid, opacity = 1](a-3-3) to (a-2-4);
    \draw[solid, opacity = 1](a-1-3) to (a-1-4);
    
    \draw[solid, opacity = 1](a-5-4) to (a-6-5);
    \draw[solid, opacity = 1](a-5-4) to (a-5-5);
    \draw[solid, opacity = 1](a-5-4) to (a-4-5);
    \draw[solid, opacity = 1](a-2-4) to (a-3-5);
    \draw[solid, opacity = 1](a-2-4) to (a-2-5);
    \draw[solid, opacity = 1](a-1-4) to (a-1-5);
    \end{scope}
    \begin{scope}[xshift=4cm]
    \scriptsize
    \matrix (b) [matrix of math nodes, row sep=0.2em, column sep=1.5em, nodes={anchor = center, inner sep=1pt}]
    {
      \particle_1^6 & \particle_2^6 & \particle_3^6 & \particle_4^6 & \particle_5^6\\
      \particle_1^5 & \particle_2^5 & \particle_3^5 & \particle_4^5 & \particle_5^5\\
      \particle_1^4 & \particle_2^4 & \particle_3^4 & \particle_4^4 & \particle_5^4\\
      \particle_1^3 & \particle_2^3 & \particle_3^3 & \particle_4^3 & \particle_5^3\\
      \particle_1^2 & \particle_2^2 & \particle_3^2 & \particle_4^2 & \particle_5^2\\
      \particle_1^1 & \particle_2^1 & \particle_3^1 & \particle_4^1 & \particle_5^1\\
      \particle_1^0 & \particle_2^0 & \particle_3^0 & \particle_4^0 & \particle_5^0\\
    };
    
    \draw[solid, color = red, opacity = 0.2, line width = 3pt](b-7-1) to (b-7-2);
    \draw[solid, color = red, opacity = 0.2, line width = 3pt](b-7-2) to (b-7-3);
    \draw[solid, color = red, opacity = 0.2, line width = 3pt](b-7-3) to (b-7-4);
    \draw[solid, color = red, opacity = 0.2, line width = 3pt](b-7-4) to (b-7-5);
    
    \draw[solid, color = blue, opacity = 0.2, line width = 3pt](b-7-1) to (b-5-2);
    \draw[solid, color = blue, opacity = 0.2, line width = 3pt](b-5-2) to (b-7-3);
    \draw[solid, color = blue, opacity = 0.2, line width = 3pt](b-7-3) to (b-7-4);
    \draw[solid, color = blue, opacity = 0.2, line width = 3pt](b-7-4) to (b-7-5);

    \draw[solid, opacity = 1](b-7-1) to (b-7-2);
    \draw[solid, opacity = 1](b-7-2) to (b-7-3);
    \draw[solid, opacity = 1](b-7-3) to (b-7-4);
    \draw[solid, opacity = 1](b-7-4) to (b-7-5);

    \draw[dotted, opacity = 1](b-7-1) to (b-6-2);
    \draw[dotted, opacity = 1](b-7-1) to (b-5-2);
    \draw[dotted, opacity = 1](b-7-1) to (b-4-2);
    \draw[dotted, opacity = 1](b-3-1) to (b-3-2);
    \draw[dotted, opacity = 1](b-3-1) to (b-2-2);
    \draw[dotted, opacity = 1](b-3-1) to (b-1-2);
    
    \draw[dotted, opacity = 1](b-7-2) to (b-6-3);
    \draw[dotted, opacity = 1](b-7-2) to (b-5-3);
    \draw[dotted, opacity = 1](b-7-2) to (b-4-3);
    \draw[dotted, opacity = 1](b-7-2) to (b-3-3);
    \draw[dotted, opacity = 1](b-1-2) to (b-2-3);
    \draw[dotted, opacity = 1](b-1-2) to (b-1-3);
    
    \draw[dotted, opacity = 1](b-7-3) to (b-6-4);
    \draw[dotted, opacity = 1](b-7-3) to (b-5-4);
    \draw[dotted, opacity = 1](b-7-3) to (b-4-4);
    \draw[dotted, opacity = 1](b-7-3) to (b-3-4);
    \draw[solid, opacity = 1](b-7-3) to (b-2-4);
    \draw[dotted, opacity = 1](b-3-3) to (b-1-4);
    
    \draw[solid, opacity = 1](b-7-4) to (b-6-5);
    \draw[solid, opacity = 1](b-7-4) to (b-5-5);
    \draw[solid, opacity = 1](b-7-4) to (b-4-5);
    \draw[solid, opacity = 1](b-7-4) to (b-3-5);
    \draw[solid, opacity = 1](b-2-4) to (b-2-5);
    \draw[solid, opacity = 1](b-2-4) to (b-1-5);
    \end{scope}
    
    \begin{scope}[xshift=8cm]
    \scriptsize
    \matrix (c) [matrix of math nodes, row sep=0.2em, column sep=1.5em, nodes={anchor = center, inner sep=1pt}]
    {
      \particle_1^6 & \particle_2^6 & \particle_3^6 & \particle_4^6 & \particle_5^6\\
      \particle_1^5 & \particle_2^5 & \particle_3^5 & \particle_4^5 & \particle_5^5\\
      \particle_1^4 & \particle_2^4 & \particle_3^4 & \particle_4^4 & \particle_5^4\\
      \particle_1^3 & \particle_2^3 & \particle_3^3 & \particle_4^3 & \particle_5^3\\
      \particle_1^2 & \particle_2^2 & \particle_3^2 & \particle_4^2 & \particle_5^2\\
      \particle_1^1 & \particle_2^1 & \particle_3^1 & \particle_4^1 & \particle_5^1\\
      \particle_1^0 & \particle_2^0 & \particle_3^0 & \particle_4^0 & \particle_5^0\\
    };
    
    \draw[solid, color = red, opacity = 0.2, line width = 3pt](c-7-1) to (c-7-2);
    \draw[solid, color = red, opacity = 0.2, line width = 3pt](c-7-2) to (c-7-3);
    \draw[solid, color = red, opacity = 0.2, line width = 3pt](c-7-3) to (c-7-4);
    \draw[solid, color = red, opacity = 0.2, line width = 3pt](c-7-4) to (c-7-5);

    \draw[solid, color = blue, opacity = 0.2, line width = 3pt](c-7-1) to (c-7-2);
    \draw[solid, color = blue, opacity = 0.2, line width = 3pt](c-7-2) to (c-7-3);
    \draw[solid, color = blue, opacity = 0.2, line width = 3pt](c-7-3) to (c-7-4);
    \draw[solid, color = blue, opacity = 0.2, line width = 3pt](c-7-4) to (c-7-5);

    \draw[solid, opacity = 1](c-7-1) to (c-7-2);
    \draw[solid, opacity = 1](c-7-2) to (c-7-3);
    \draw[solid, opacity = 1](c-7-3) to (c-7-4);
    \draw[solid, opacity = 1](c-7-4) to (c-7-5);

    \draw[dotted, opacity = 1](c-7-1) to (c-6-2);
    \draw[dotted, opacity = 1](c-7-1) to (c-5-2);
    \draw[dotted, opacity = 1](c-7-1) to (c-4-2);
    \draw[dotted, opacity = 1](c-7-1) to (c-3-2);
    \draw[dotted, opacity = 1](c-7-1) to (c-2-2);
    \draw[dotted, opacity = 1](c-7-1) to (c-1-2);
    
    \draw[dotted, opacity = 1](c-7-2) to (c-6-3);
    \draw[dotted, opacity = 1](c-7-2) to (c-5-3);
    \draw[dotted, opacity = 1](c-7-2) to (c-4-3);
    \draw[dotted, opacity = 1](c-7-2) to (c-3-3);
    \draw[dotted, opacity = 1](c-7-2) to (c-2-3);
    \draw[dotted, opacity = 1](c-7-2) to (c-1-3);
    
    \draw[dotted, opacity = 1](c-7-3) to (c-6-4);
    \draw[dotted, opacity = 1](c-7-3) to (c-5-4);
    \draw[dotted, opacity = 1](c-7-3) to (c-4-4);
    \draw[dotted, opacity = 1](c-7-3) to (c-3-4);
    \draw[dotted, opacity = 1](c-7-3) to (c-2-4);
    \draw[dotted, opacity = 1](c-7-3) to (c-1-4);
    
    \draw[solid, opacity = 1](c-7-4) to (c-6-5);
    \draw[solid, opacity = 1](c-7-4) to (c-5-5);
    \draw[solid, opacity = 1](c-7-4) to (c-4-5);
    \draw[solid, opacity = 1](c-7-4) to (c-3-5);
    \draw[solid, opacity = 1](c-7-4) to (c-2-5);
    \draw[solid, opacity = 1](c-7-4) to (c-1-5);
    \end{scope}
    \node (low) [above of = a, yshift = 2em, font = \footnotesize] {Low dimensions};
    \node (moderate) [above of = b, yshift = 2em, font = \footnotesize] {Moderate dimensions};
    \node (high) [above of = c, yshift = 2em, font = \footnotesize] {High dimensions};
  \end{tikzpicture}
     \caption{With backward sampling.}
  \end{subfigure}%
  \caption{Breakdown of the \gls{ICSMC} algorithm in high dimensions. Black lines represent particle lineage induced by the algorithm, i.e.\ a line connects $\smash{\particle_{t-1}^m}$ and $\smash{\particle_t^n}$ iff $\smash{a_{t-1}^n = m}$. Solid lines (\protect\tikz[baseline=-0.5ex]\protect\draw[opacity = 1, solid] (0,0) -- (0.5,0);) represent the surviving lineages at time $T = 5$. Dotted lines (\,\protect\tikz[baseline=-0.5ex]\protect\draw[opacity = 1, dotted] (0,0) -- (0.5,0);) represent lineages that have died out. The red line (\protect\tikz[baseline=-0.5ex]\protect\draw[opacity = 0.3, line width = 3pt, solid, color = red] (0,0) -- (0.5,0);) represents the old reference path $\smash{\state_{1:T}[l] = (\particle_1^0, \dotsc, \particle_T^0)}$. The blue line (\protect\tikz[baseline=-0.5ex]\protect\draw[opacity = 0.3, line width = 3pt, solid, color = blue] (0,0) -- (0.5,0);) represents the new reference path $\smash{\state_{1:T}[l+1] = (\particle_1^{k_1}, \dotsc, \particle_T^{k_t})}$. 
  }
  \label{fig:breakdown_of_csmc}
\end{figure}
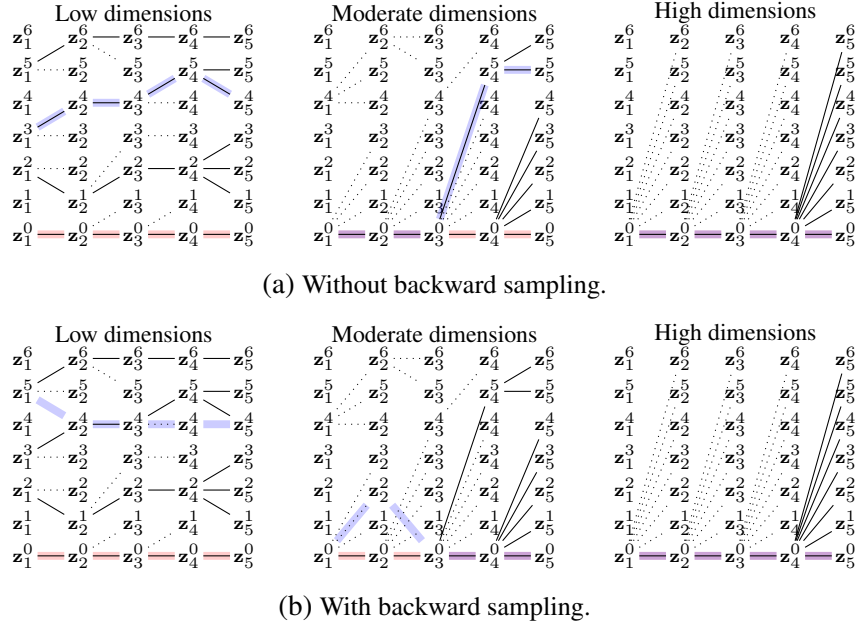

The degenerate limit of the law of the genealogies and the indices of the new reference path under the \gls{ICSMC} algorithm is the law $\smash{\IteratedCsmcKernel{\nTimeSteps}{\nParticles}(\diff a_{1:\nTimeSteps-1} \times \diff\outputParticleIndex_{1:\nTimeSteps})}$
which deterministically sets all ancestor indices and all indices of the new reference path to $0$. More formally,
\begin{align}
  \IteratedCsmcKernel{\nTimeSteps}{\nParticles}(\diff a_{1:\nTimeSteps-1} \times \diff\outputParticleIndex_{1:\nTimeSteps})
  & \coloneqq \dDirac_0^{\otimes ((\nParticles+1)(\nTimeSteps-1) + \nTimeSteps)}(\diff a_{1:\nTimeSteps-1} \times \diff \outputParticleIndex_{1:\nTimeSteps}).
  \label{eq:joint_law_rwcsmc_high}
\end{align}
We let $\smash{\ExpectationCsmcKernel{\nTimeSteps}{\nParticles}}$ denote expectation \WRT this law. The proof that the law of the genealogies and new reference path indices indeed converges to this trivial limit will be given below in Proposition~\ref{prop:limiting_csmc_algorithm} which relies on the following assumptions, where $\E$ denotes expectation w.r.t.\ $X_{1:T} \sim \pi_T$ and where
\begin{gather}
  r_{t|T} \coloneqq \smash{\E[\log G_t(X_t)] - \E[\log M_t(G_t)(X_{t-1})],}\\
  b_{t|T} \coloneqq \smash{\E[\log G_t(X_t) + \log m_{t+1}(X_t, X_{t+1})] - \E[\log M_t(G_t m_{t+1}(\ccdot, X_{t+1}))(X_{t-1})].}
  \end{gather}
\begin{enumerate}[label=\textbf{A\arabic*}, resume=model_assumptions]
  \item \label{as:degeneracy_of_resampling_kernels} $\inf_{t \in [T]} r_{t|T} \eqqcolon \underline{r}_T > 0$.\tendmark
  \item \label{as:degeneracy_of_backward_kernels} $\inf_{t \in [T-1]} b_{t|T} \eqqcolon \underline{b}_T > 0$. \tendmark
\end{enumerate}
 Assumptions~\ref{as:degeneracy_of_resampling_kernels} and \ref{as:degeneracy_of_backward_kernels} are not restrictive: (a) they do not depend on multiplication of $G_t$ by some positive constant; (b) they hold if the model factorises over time (see Assumption~\ref{as:independent_over_time} in Section~\ref{subsec:large_time-horizon_asymptotics}) unless the potential functions $G_t$ are almost-everywhere constant; (c) Appendix~\ref{app:subsec:verification_of_csmc_assumptions_for_breakdown} shows that these assumptions hold even in a simple linear-Gaussian state-space model. 

We now state our first main result, Proposition~\ref{prop:limiting_csmc_algorithm} (proved in Appendix~\ref{app:subsec:proof_of_prop:limiting_csmc_algorithm}), which shows that in high dimensions, all particle lineages coalesce immediately with the reference path unless in number of particles, $N+1$, $N = N(D)$ grows exponentially in the spatial dimension $D$ -- even with the backward-sampling or forced-move extensions. Let $\lVert \ccdot \rVert$ denote the total variation distance.

\begin{proposition}[curse of dimension]\label{prop:limiting_csmc_algorithm} 
 Let $ \nTimeSteps \in \naturals$. Assume \ref{as:iid_model} and \ref{as:degeneracy_of_resampling_kernels} (as well as \ref{as:degeneracy_of_backward_kernels} if backward sampling is used) and write
 \begin{align}
  \!\!\!\!\smash{d_{T,D,\state_{1:T}}^N \coloneqq \bigl\lVert \ExpectationCsmcKernel{\nTimeSteps,\nDimensions,\state_{1:\nTimeSteps}}{\nParticles}[\ind\{(A_{1:\nTimeSteps-1}, K_{1:T}) \in \ccdot\}] - \ExpectationCsmcKernel{\nTimeSteps}{\nParticles}[\ind\{(A_{1:\nTimeSteps-1}, K_{1:T}) \in \ccdot\}]\bigr\rVert.}\!\!\!\!\!
 \end{align}
 Then there exists a family $\ConcentrationSet_{\nTimeSteps, \nDimensions} \in \sigFieldState_{T,D}$ with $\lim_{\nDimensions \to \infty} \Target_{\nTimeSteps, \nDimensions}(\ConcentrationSet_{\nTimeSteps, \nDimensions}) = 1$ and
 \[
  \smash{\log N = \lo(D) \quad \Longrightarrow \quad \lim\nolimits_{\nDimensions \to \infty} \sup\nolimits_{\state_{1:\nTimeSteps} \in \ConcentrationSet_{\nTimeSteps, \nDimensions}} \smash{d_{T,D,\state_{1:T}}^N} = 0.} \mendmark
 \] 
\end{proposition}

An immediate consequence of Proposition~\ref{prop:limiting_csmc_algorithm} is the following corollary which shows that all acceptance rates vanish in high dimensions.
\begin{corollary}\label{cor:vanishing_acceptance_rates_csmc}
 Under the assumptions of Proposition~\ref{prop:limiting_csmc_algorithm} and with the same sets $\ConcentrationSet_{\nTimeSteps, \nDimensions} \in \sigFieldState_{T,D}$, for any $t \in [\nTimeSteps]$:
 \[
   \log N = \lo(D) \quad \Longrightarrow \quad \lim\nolimits_{\nDimensions \to \infty} \sup\nolimits_{\state_{1:\nTimeSteps} \in \ConcentrationSet_{\nTimeSteps, \nDimensions}}  \acceptanceRate{\nTimeSteps,\nDimensions,\state_{1:T}}{\nParticles}(\timeIndex) =  0. \mendmark
 \]
\end{corollary} 

\section{Simplified dimensionally stable methodology: the RW-EHMM algorithm}
\label{sec:rwehmm}
\glsreset{RWEHMM}
\glsreset{ESJD}
\glsreset{EHMM}

\subsection{Description of the algorithm}
\label{subsec:rwcsmc_main_ideas_in_a_simpler_setting}

Our novel \gls{IRWCSMC} algorithm will be presented in the next section. To ease the exposition, we first -- in this section -- introduce another novel algorithm, the \emph{\gls{RWEHMM}} algorithm. This method can be viewed as a simplified version of the \gls{IRWCSMC} algorithm and is likewise stable in high dimensions. However, the implementation of a single \gls{RWEHMM} update requires $\bo(N^2T)$ operations whilst a single \gls{IRWCSMC} update only requires $\bo(NT)$ operations. 

\subsubsection{Basic algorithm}

The \emph{\gls{RWEHMM}} algorithm proposed in this section also induces a $\Target_{\nTimeSteps, \nDimensions}$-invariant Markov kernel. It can viewed as an instance of the \emph{\gls{EHMM}} methods from \citet{neal2003markov, neal2004inferring} which are closely related to iterated \gls{CSMC} methods with backward sampling as explained in \citet{finke2016embedded}. The main difference between iterated \gls{CSMC} and \gls{EHMM} methods is that the former use resampling steps to permit implementation in $\bo(NT)$ operations wheras the latter typically require $\bo(N^2T)$ operations. 


The $l$th update of the novel \gls{RWEHMM} scheme is outlined in Algorithm~\ref{alg:iterated_rwcsmc} where we use the convention that any action described for the $\particleIndex$th particle index is to be performed conditionally independently for all $\particleIndex \in [\nParticles]$ and any action described for the $\dimensionIndex$th `spatial' component is to be performed conditionally independently for all $\dimensionIndex \in [D]$.

\noindent\parbox{\textwidth}{
\begin{flushleft}
 \begin{framedAlgorithm}[\gls{RWEHMM}] \label{alg:iterated_rwehmm} Given $\state_{1:T} \coloneqq \state_{1:\nTimeSteps}[l] \in \spaceState_{T,D}$.
 \begin{enumerate}
  \item \label{alg:iterated_rwehmm:1} For $\timeIndex \in [\nTimeSteps]$:
  set $\Particle_\timeIndex^0 = \particle_\timeIndex^0 \coloneqq \state_\timeIndex[l]$ and sample $\Particle_\timeIndex^{1:\nParticles} = \particle_\timeIndex^{1:\nParticles}$ as follows:
  \begin{enumerate}
   \item sample $U_{\timeIndex,\dimensionIndex}^{1:\nParticles} \sim \dN(0, \varSigma)$,
   \item set $\smash{\particleSingle_{\timeIndex,\dimensionIndex}^\particleIndex \coloneqq \particleSingle_{\timeIndex,\dimensionIndex}^0 + \sqrt{\ell_\timeIndex /\nDimensions} U_{\timeIndex,\dimensionIndex}^{\particleIndex}}$,
   \item set $\particle_{\timeIndex}^\particleIndex \coloneqq \particleSingle_{\timeIndex,1:D}^\particleIndex$.
  \end{enumerate}
  
   \item \label{alg:iterated_rwehmm:2} Sample $K_{1:T} = k_{1:T} \in [\nParticles]_0$ with probability 
   \begin{align}
     \xi_T(\particle_{1:T}, \{k_{1:T}\}) \coloneqq \smash{\dfrac{\Target_{T,D}(\particle_1^{k_1}, \dotsc, \particle_T^{k_T})}{\sum_{l_{1:T} \in [N]_0^T} \Target_{T,D}(\particle_1^{l_1}, \dotsc, \particle_T^{l_T})}}. \label{eq:rwehmm_discrete_target}
   \end{align}


  \item \label{alg:iterated_rwehmm:3} Set $\State_{1:\nTimeSteps}' \coloneqq \state_{1:T}' \coloneqq (\particle_1^{\outputParticleIndex_1}, \dotsc, \particle_\nTimeSteps^{\outputParticleIndex_\nTimeSteps})$. 
  
  \item \label{alg:iterated_rwehmm:4} Return $\state_{1:\nTimeSteps}[l+1] \coloneqq \state_{1:T}'$. 
  
 \end{enumerate}
\end{framedAlgorithm}
\end{flushleft}
}

Step~\ref{alg:iterated_rwehmm:1} of Algorithm~\ref{alg:iterated_rwehmm} scatters particles around the reference particle $\particle_t^0 = \state_t$ by applying Gaussian noise independently in each dimension $d \in [D]$:
\begin{equation}
  \smash{(Z_{t,d}^1, \dotsc, Z_{t,d}^N)^\T \sim \dN\bigl(\unitMat_N z_{t,d}^0, \tfrac{\ell_t}{D}\varSigma\bigr),} \label{eq:applying_gaussian_noise}
 \end{equation}
 where $\unitMat_N$ is a vector of $1$s of length $N$ and where 
 \begin{itemize}
  \item $\smash{\varSigma \coloneqq \tfrac{1}{2}(\unitMat_N \unitMat_N^\T + \iMat_N)}$ is an $(N \times N)$ covariance matrix with $1$ on the diagonal and $\tfrac{1}{2}$ everywhere else which governs the correlation between (univariate spatial components of) different particles,
  \item $\ell_t > 0$ is some scale factor that governs how far (on average) particles are scattered around the reference path.
 \end{itemize}
This proposal was introduced by \citet{tjelmeland2004using} who noted that
\begin{itemize}
 \item the marginal distributions of individual particles are simply Gaussian random-walk moves with variance $\ell_t/D$, i.e.\ $\smash{\Particle_t^n \sim \dN(\particle_t^0, \tfrac{\ell_t}{D}\iMat_D)}$, for $n \in [N]$;
 
 \item \eqref{eq:applying_gaussian_noise} can be viewed as first sampling a new `centre' $\smash{Z_{t,d}' = z_{t,d}' \sim \dN(z_{t,d}^0, \ell_t/2)}$ and then sampling $\smash{Z_{t,d}^1, \dotsc, Z_{t,d}^N \iidSim \dN(z_{t,d}', \tfrac{\ell_t}{2D})}$:
 \[
  \int_{-\infty}^\infty \dN(z'; z_{t,d}^0, \tfrac{\ell_t}{2D}) \biggl[\prod_{n=1}^N \dN(z_{t,d}^n, z', \tfrac{\ell_t}{2D}) \biggr]\intDiff z' = \dN\bigl(z_{t,d}^{1:N}; \unitMat_N z_{t,d}^0, \tfrac{\ell_t}{D}\varSigma\bigr). 
 \]
\end{itemize}
Other types of `local' proposals (i.e.\ not necessarily based on Gaussian random walks) could be used. However, we limit our analysis to this particular structure because it is symmetric in the sense that its density cancels out in the selection functions. More formally, for any $n,m \in [N]_0$, where $\smash{z_{t,d}^{-n} \coloneqq (z_{t,d}^0, \dotsc, z_{t,d}^{n-1}, z_{t,d}^{n+1}, \dotsc, z_{t,d}^N)}$:
\begin{align}
  \dN\bigl(z_{t,d}^{-m}; \unitMat_N z_{t,d}^m, \tfrac{\ell_t}{D}\varSigma\bigr) = \dN\bigl(z_{t,d}^{-n}; \unitMat_N z_{t,d}^n, \tfrac{\ell_t}{D}\varSigma\bigr).
  \label{eq:rwcsmc_symmetric_proposal:0}
 \end{align}
  



\subsubsection{Implementation in \texorpdfstring{$\bo(N^2T)$}{} operations}

Even though Step~\ref{alg:iterated_rwehmm:2} of Algorithm~\ref{alg:iterated_rwehmm} requires sampling from the distribution $\smash{\xi_T(\particle_{1:T}, \ccdot)}$ whose support is $(N+1)^T$-dimensional, \citet{neal2003markov} recognised that -- by exploiting the Markov property of the model -- this can be achieved in $\bo(N^2T)$ operations using standard forward filtering--backward sampling recursions for finite-state \glsdescplural{HMM} as follows.
\begin{enumerate}
 \item \emph{Forward filtering.} For $t = 1,\dotsc, T$, compute (with convention $w_0^n \coloneqq 1$):
  \begin{align}
   w_t^n \coloneqq \sum_{m \in [N]_0} \frac{w_{t-1}^m}{\sum_{l \in [N]_0} w_{t-1}^l} \mutation_t(\particle_{t-1}^m, \particle_t^n) \Potential_t(\particle_t^n).
  \end{align}
 
 \item \emph{Backward sampling.} For $t = T, \dotsc, 1$ (with convention $\mutation_{T+1} \equiv 1$), sample $K_t = k_t \in [N]_0$ with probability
 \begin{align}
    \dfrac{w_t^{k_t} \mutation_{t+1}(\particle_t^{k_t}, \particle_{t+1}^{k_{t+1}})}{\sum_{n \in [N]_0} w_t^{n} \mutation_{t+1}(\particle_t^{n}, \particle_{t+1}^{k_{t+1}})}.
 \end{align}
\end{enumerate}

\subsubsection{Induced \texorpdfstring{$\Target_{\nTimeSteps,\nDimensions}$-}{}invariant Markov kernel}


Given $\smash{\State_{1:T} = \state_{1:\nTimeSteps} = \state_{1:\nTimeSteps}[l]}$, let
\begin{align}
  \smash{\IteratedEhmmKernelRandomWalk{\nTimeSteps,\nDimensions,\state_{1:T}}{\nParticles}(\diff \particle_{1:\nTimeSteps} \times \diff k_{1:T} \times \diff \state_{1:T}')},
\end{align}
be the law of all the random variables $\smash{(\Particle_{1:T}, K_{1:T}, \State_{1:T}')}$ generated in Steps~\ref{alg:iterated_rwehmm:1}--\ref{alg:iterated_rwehmm:3} of Algorithm~\ref{alg:iterated_rwehmm}.
Appendix~\ref{app:subsec:joint_law_rwehmm} gives a formal definition of this law. 

Let $\smash{\ExpectationEhmmKernelRandomWalk{\nTimeSteps,\nDimensions,\state_{1:T}}{\nParticles}}$ denote expectation \WRT $\smash{\IteratedEhmmKernelRandomWalk{\nTimeSteps,\nDimensions,\state_{1:T}}{\nParticles}}$. Algorithm~\ref{alg:iterated_rwehmm} induces a Markov kernel 
\begin{align}
 \InducedIteratedEhmmKernelRandomWalk{\nTimeSteps,\nDimensions}{\nParticles}(\state_{1:\nTimeSteps}, \diff \state_{1:T}') 
 & \coloneqq \ExpectationEhmmKernelRandomWalk{\nTimeSteps,\nDimensions,\state_{1:\nTimeSteps}}{\nParticles}[\ind\{\State_{1:T}' \in \diff \state_{1:T}'\}], \label{eq:iterated_rwehmm_kernel}
\end{align}
for $(\state_{1:\nTimeSteps},\diff \state_{1:T}') \in \spaceState_{T,D} \times \sigFieldState_{T,D}$. The following proposition shows that this Markov kernel leaves $\Target_{\nTimeSteps, \nDimensions}$ invariant. It can be proved by arguments similar to those in \citet{neal2003markov, neal2004inferring}. For completeness, we give a simple proof in Appendix~\ref{app:subsec:rwehmm_invariance}.

\begin{proposition}\label{prop:rwehmm_invariance}
  For any $N, T, D \in \naturals$, $\Target_{T,D}\InducedIteratedEhmmKernelRandomWalk{\nTimeSteps,\nDimensions}{\nParticles} = \Target_{T,D}$. \tendmark
\end{proposition}

We stress that this proposition does not require the high-dimensional regime from Assumption~\ref{as:iid_model}. That is, Algorithm~\ref{alg:iterated_rwehmm} induces a valid (i.e. $\Target_{\nTimeSteps, \nDimensions}$-invariant) Markov kernel even if the model does not factorise into $\nDimensions$ \gls{IID} components.

\subsection{Stability in high dimensions}
\label{subsec:rwehmm_scaling}

We now show that the \gls{RWEHMM} algorithm is stable in high dimensions. For the analysis, we assume the regime from Assumption~\ref{as:iid_model}.

\subsubsection{Non-degenerate limiting law of the particle indices}

In the following, we show that as as $\nDimensions \to \infty$, the law of the particle indices of the new reference path, $K_{1:T}$, under the \gls{RWEHMM} algorithm converges to a limit which is non-degenerate in the sense that the acceptance probabilities are strictly positive at each time step. Using the convention that $\smash{\partial_t^i}$ denotes the $i$th derivative w.r.t.\ $\stateSingle_t$ and with $\smash{\partial_t \coloneqq \partial_t^1}$, as well as with $\pi_T(\varphi) \coloneqq \int_{\reals^T} \varphi(x_{1:T}) \pi_T(x_{1:T}) \intDiff x_{1:T}$, for any $\pi_T$-integrable function $\varphi: \reals^T \to \reals$, we make the following moment assumption.
\begin{enumerate}[label=\textbf{B\arabic*}, series=rwehmm]
  \item \label{as:moments_bounded_ehmm} The density $\pi_T$ is twice continuously differentiable and for any $s,t \in [T]$, 
  \begin{itemize}
   \item $\partial_s \partial_t \log \pi_T$ is Lipschitz-continuous and bounded,
   \item $\pi_T(\lvert \partial_t \log \pi_T \rvert^4) < \infty$. \tendmark
  \end{itemize}
\end{enumerate}

Hereafter, we assume \ref{as:moments_bounded_ehmm}. The limiting law of $K_{1:T}$ (proved below) is then given by 
\begin{align}
  \IteratedEhmmKernelRandomWalk{\nTimeSteps}{\nParticles}(\diff v_{1:\nTimeSteps} \times \diff\outputParticleIndex_{1:\nTimeSteps})
  & \coloneqq \biggl[ \prod_{\timeIndex = 1}^\nTimeSteps \dN(\diff v_t; \tilde{\mu}_{t|T}, \widetilde{\varSigma}_{t|T}) \biggr] \prod_{t=1}^T \selectionFunctionBoltzmann{k_t}(\{v_t^m\}_{m=1}^N).
\end{align}
Expectations w.r.t.\ $\smash{\IteratedEhmmKernelRandomWalk{\nTimeSteps}{\nParticles}}$ are denoted by $\smash{\ExpectationEhmmKernelRandomWalk{\nTimeSteps}{\nParticles}}$. This law is defined through $N$-dimensional Gaussian random vectors $V_t \coloneqq V_t^{1:N}$ which are such that $V_s$ and $V_t$ are independent whenever $s \neq t$ and with mean vector and covariance matrix
\begin{align}
 \E[V_t] =  - \tfrac{1}{2}\ell_t \calI_{t|T} \unitMat_N \eqqcolon \tilde{\mu}_{t|T}, \quad \text{and} \quad \var[V_t] = \ell_t \calI_{t|T} \varSigma \eqqcolon \widetilde{\varSigma}_{t|T},
\end{align}
where by Lemma~\ref{lem:integration_by_parts_identity} in Appendix~\ref{app:subsec:prop:stability_of_acceptance_rates_rwcsmc},
\begin{align}
 \calI_{t|T} \coloneqq \pi_T([\partial_t \log \pi_T]^2) = - \pi_T(\partial_t^2 \log \pi_T).
\end{align}


\subsubsection{Convergence to the non-degenerate limit}

The following proposition (whose proof is a simpler version of the proof of Proposition~\ref{prop:limiting_rwcsmc_algorithm} in Section~\ref{sec:rwcsmc} and is therefore omitted) shows that in high dimensions, the law of the indices $K_{1:T}$ under the \gls{RWEHMM} update specified in Algorithm~\ref{alg:iterated_rwehmm} converges to a limit that is non-trivial in the sense that the events $\{K_t \neq 0\}$ have positive probability. Again, $\lVert \ccdot \rVert$ is the total variation distance.

\begin{proposition}[convergence of the law of the particle indices]\label{prop:limiting_rwehmm_algorithm}
 Let $T, N \in \naturals$, assume \ref{as:iid_model} as well as \ref{as:moments_bounded_ehmm}, and write
  \begin{align}
  \!\!\!\!\! \smash{\tilde{d}_{T,D,\state_{1:T}}^N \coloneqq \bigl\lVert \ExpectationEhmmKernelRandomWalk{T,D,\state_{1:T}}{\nParticles}[\ind\{ \OutputParticleIndex_{1:\nTimeSteps} \in \ccdot\}] - \ExpectationEhmmKernelRandomWalk{\nTimeSteps}{\nParticles}[\ind\{\OutputParticleIndex_{1:\nTimeSteps} \in \ccdot\}]\bigr\rVert.}\!\!\!\!\!
 \end{align}
 Then there exists a family $\ConcentrationSet_{\nTimeSteps, \nDimensions} \in \sigFieldState_{T,D}$ with $\lim_{\nDimensions \to \infty} \Target_{\nTimeSteps, \nDimensions}(\ConcentrationSet_{\nTimeSteps, \nDimensions}) = 1$ and
 \[
  \smash{\lim\nolimits_{\nDimensions \to \infty} \sup\nolimits_{\state_{1:\nTimeSteps} \in \ConcentrationSet_{\nTimeSteps, \nDimensions}} \smash{\tilde{d}_{T,D,\state_{1:T}}^N} = 0.} \mendmark
 \]
\end{proposition}

The following proposition (proved in Appendix~\ref{app:subsec:prop:lower_bound_of_acceptance_rates_ehmm}) shows that the \emph{acceptance rate} at any time $t$ associated with Algorithm~\ref{alg:iterated_rwehmm},
\begin{align}
 \smash{\acceptanceRateRwEhmm{T,D,\state_{1:T}}{\nParticles}(\timeIndex) 
  \coloneqq
   \ExpectationEhmmKernelRandomWalk{\nTimeSteps,\nDimensions,\state_{1:T}}{\nParticles}[ \ind\{\OutputParticleIndex_t \neq 0 \}],}
\end{align}
converges to a strictly positive limit
\begin{align}
 \smash{\acceptanceRateRwEhmm{\nTimeSteps}{\nParticles}(\timeIndex) 
 \coloneqq
   \ExpectationEhmmKernelRandomWalk{\nTimeSteps}{\nParticles}[ \ind\{\OutputParticleIndex_t \neq 0\}].}
\end{align}
Note that the acceptance rates and their limits depend on $\ell_{1:T}$ even though we do not make this explicit in our notation. 

\begin{proposition}[dimensional stability of the acceptance rates] \label{prop:lower_bound_of_acceptance_rates_ehmm} 
 Assume \ref{as:iid_model} and \ref{as:moments_bounded_ehmm}. Then for $T, N \in \naturals$, $t \in [T]$ and $\ConcentrationSet_{\nTimeSteps, \nDimensions} \in \sigFieldState_{T,D}$ as in Proposition~\ref{prop:limiting_rwehmm_algorithm},
 \begin{align}
   \smash{\lim\nolimits_{\nDimensions \to \infty} \sup\nolimits_{\state_{1:\nTimeSteps} \in \ConcentrationSet_{\nTimeSteps, \nDimensions}} \lvert \acceptanceRateRwEhmm{\nTimeSteps,\nDimensions,\state_{1:T}}{\nParticles}(\timeIndex) - \acceptanceRateRwEhmm{\nTimeSteps}{\nParticles}(\timeIndex)\rvert = 0,} \label{eq:prop:lower_bound_of_acceptance_rates_ehmm:1}
  \end{align}
 where
  \[
   \acceptanceRateRwEhmm{\nTimeSteps}{\nParticles}(\timeIndex) \geq \biggl(1 + \frac{\exp(\ell_t \calI_t )}{\nParticles}\biggr)^{\smash{\mathrlap{-1}}} > 0. \mendmark
  \]
  \end{proposition}

\glsreset{ESJD}

Of course, stabilising the acceptance rates in high dimensions is not sufficient for avoiding a breakdown. A widely used criterion for assessing the performance of \gls{MCMC} algorithms is the \emph{\gls{ESJD}} \citep{sherlock2009optimal}, which (for the time-$t$ component in Algorithm~\ref{alg:iterated_rwehmm}) is given by
\begin{align}
 \smash{\esjdRandomWalkEhmm_{T,D}^N(t) \coloneqq \E[\lVert \State_{t}[l+1] - \State_{t}[l] \rVert_2^2],}
\end{align}
where $\lVert \ccdot \rVert_2$ denotes the Euclidean norm and where $\State_{1:T}[l]$ is the $l$th state of the Markov chain with transition kernel $\smash{\InducedIteratedEhmmKernelRandomWalk{T,D}{N}}$ at stationarity. The following proposition (proved in Appendix~\ref{app:subsec:prop:lower_bound_of_acceptance_rates_ehmm}) shows that the \gls{ESJD} is also stable in high dimensions.

\begin{proposition}[dimensional stability of the \gls{ESJD}]\label{prop:esjd_rwehmm}
 Assume \ref{as:iid_model} and \ref{as:moments_bounded_ehmm} and let $\nTimeSteps, \nParticles \in \naturals$. Then, 
 for any $t \in [T]$,
 \[
  \smash{\lim\nolimits_{\nDimensions \to \infty} \bigl\lvert 
  \esjdRandomWalkEhmm_{T,D}^{N}(t)
  -  \ell_t \acceptanceRateRwEhmm{\nTimeSteps}{\nParticles}(t)\bigr\rvert = 0.} \mendmark
 \]
\end{proposition}

\subsubsection{Stability \texorpdfstring{as $T \to \infty$}{in the time horizon}}
\label{subsec:large_time-horizon_asymptotics_ehmm}

An immediate consequence of Proposition~\ref{prop:lower_bound_of_acceptance_rates_ehmm} is the following corollary which shows that the limiting acceptance rates $\smash{\acceptanceRateRwEhmm{\nTimeSteps}{\nParticles}(t)}$ are guaranteed to be bounded away from zero as $\nTimeSteps \to \infty$. 

\begin{corollary}[time-horizon stability of the acceptance rates] \label{cor:stability_of_acceptance_rates_ehmm} Assume \ref{as:iid_model} and that  $\calI(\ell) \coloneqq \inf_{T \in \naturals}\inf_{t \in [T]} \ell_t \calI_{t|T} < \infty$. Then, for any $N \in \naturals$, 
  \[
    \inf_{\nTimeSteps\in \naturals} \inf_{t \in [T]} \acceptanceRateRwEhmm{\nTimeSteps}{\nParticles}(\timeIndex) \geq \biggl(1 + \frac{\exp(\calI(\ell))}{N}\biggr)^{\smash{\mathrlap{-1}}} > 0. \mendmark
  \]
\end{corollary}

\begin{remark}
 Proposition~\ref{prop:esjd_rwehmm} makes it clear that the \gls{ESJD} is also stable under the additional assumption that $\inf_{t \geq 1} \ell_t > 0$. \tendmark
\end{remark}

\section{Proposed dimensionally stable methodology: the i-RW-CSMC algorithm}
\label{sec:rwcsmc}

\glsreset{ESJD}
\glsreset{RWCSMC}
\glsreset{IRWCSMC}

\subsection{Description of the algorithm}
\label{subsec:rwcsmc_description_of_the_algorithm}

\subsubsection{Basic algorithm}

In this section, we introduce our novel iterated \emph{\gls{IRWCSMC}} algorithm which induces an alternative  $\Target_{\nTimeSteps, \nDimensions}$-invariant Markov kernel. We also prove that the algorithm overcomes the curse of dimension suffered by the existing \gls{ICSMC} approach. The proposed algorithm scatters particles locally around the reference path in the same way as the \gls{RWEHMM} method introduced in the previous section. However, recall that this algorithm required $\bo(N^2T)$ operations per iteration for fixed dimensions $D$. In contrast, the algorithm proposed in this section uses resampling steps to ensure that a single update can be implemented in $\bo(NT)$ operations as in the standard \gls{ICSMC} approach.

For any $\timeIndex \in [\nTimeSteps]$ and any $\state_{1:\nTimeSteps} \in \spaceState_{T,D}$, we also define 
\begin{align}
 \smash{\logWeightRandomWalk_\timeIndex(\state_{\timeIndex-1:\timeIndex}) \coloneqq \log \mutation_\timeIndex(\state_{\timeIndex-1}, \state_\timeIndex) + \log \Potential_\timeIndex(\state_\timeIndex),}
\end{align}
with the convention that any quantity with time index $\timeIndex < 1$ is to be ignored. The $l$th update of the novel \gls{IRWCSMC} scheme is then outlined in Algorithm~\ref{alg:iterated_rwcsmc} where we use the convention that any action described for the $\particleIndex$th particle index is to be performed conditionally independently for all $\particleIndex \in [\nParticles]$ and any action described for the $\dimensionIndex$th `spatial' component is to be performed conditionally independently for all $\dimensionIndex \in [D]$. 

\noindent\parbox{\textwidth}{
\begin{flushleft}
 \begin{framedAlgorithm}[\gls{IRWCSMC}] \label{alg:iterated_rwcsmc} Given $\state_{1:T} \coloneqq \state_{1:\nTimeSteps}[l] \in \spaceState_{T,D}$.
 \begin{enumerate}
  \item \label{alg:iterated_rwcsmc:1} For $\timeIndex \in [\nTimeSteps]$:
 \begin{enumerate}
  \item \label{alg:iterated_rwcsmc:1:a} if $\timeIndex > 1$,
  \begin{enumerate}
   \item set $\smash{A_{\timeIndex-1}^0 = a_{\timeIndex-1}^0 \coloneqq 0}$,
   \item sample $\smash{A_{\timeIndex-1}^\particleIndex = a_{\timeIndex-1}^\particleIndex = l \in [\nParticles]_0}$ with probability 
   \begin{align}
    \!\!\!\!\!\!\!\!\!\!\!\!\!\!\!\!\!\!\!\!\!\!\!\!\!\!\!\!\!\!\!\!\!\selectionFunctionBoltzmann{l}(\{\logWeightRandomWalk_{\timeIndex-1}(\particle_{\timeIndex-2}^{a_{\timeIndex-2}^\particleIndexAlt}, \particle_{\timeIndex-1}^\particleIndexAlt) - \logWeightRandomWalk_\timeIndex(\particle_{\timeIndex-2}^0, \particle_{\timeIndex-1}^0)\}_{\particleIndexAlt = 1}^\nParticles)
    = 
    \dfrac{\mutation_{\timeIndex-1}(\particle_{\timeIndex - 2}^{\smash{a_{\timeIndex - 2}^l}}, \particle_{\timeIndex-1}^{l})\Potential_{\timeIndex-1}(\particle_{\timeIndex-1}^{l})}{\sum_{\particleIndexAlt=0}^\nParticles \mutation_{\timeIndex-1}(\particle_{\timeIndex - 2}^{a_{\timeIndex - 2}^\particleIndexAlt}, \particle_{\timeIndex-1}^\particleIndexAlt) \Potential_{\timeIndex-1}(\particle_{\timeIndex-1}^\particleIndexAlt)},\!\!\!\!\!\!\!
    \end{align}
  \end{enumerate}
  
  \item \label{alg:iterated_rwcsmc:1:b} set $\smash{\Particle_\timeIndex^0 = \particle_\timeIndex^0 \coloneqq \state_\timeIndex}$ and sample $\smash{\Particle_\timeIndex^{1:\nParticles} = \particle_\timeIndex^{1:\nParticles}}$ as follows:
  \begin{enumerate}
   \item sample $\smash{U_{\timeIndex,\dimensionIndex}^{1:\nParticles} \sim \dN(0, \varSigma)}$,
   \item set $\smash{\particleSingle_{\timeIndex,\dimensionIndex}^\particleIndex \coloneqq \particleSingle_{\timeIndex,\dimensionIndex}^0 + \sqrt{\ell_\timeIndex/\nDimensions} U_{\timeIndex,\dimensionIndex}^{\particleIndex}}$,
   \item set $\particle_{\timeIndex}^\particleIndex \coloneqq \particleSingle_{\timeIndex,1:D}^\particleIndex$.
  \end{enumerate}
  \end{enumerate}
   \item \label{alg:iterated_rwcsmc:2a} Sample $\smash{\OutputParticleIndex_\nTimeSteps =\outputParticleIndex_\nTimeSteps \in [\nParticles]_0}$ with probability 
      \begin{align}
   \!\!\!\!\!\!\!\!\!\!\!\!\!\!\selectionFunctionBoltzmann{\outputParticleIndex_\nTimeSteps}(\{\logWeightRandomWalk_{\nTimeSteps}(\particle_{\nTimeSteps-1}^{\mathrlap{a_{\nTimeSteps-1}^\particleIndexAlt}}, \particle_{\nTimeSteps}^\particleIndexAlt) - \logWeightRandomWalk_\timeIndex(\particle_{\nTimeSteps-1}^0, \particle_{\nTimeSteps}^0)\}_{\particleIndexAlt = 1}^\nParticles)
    = 
    \smash{\dfrac{\mutation_{\nTimeSteps}(\particle_{\nTimeSteps-1}^{a_{\nTimeSteps-1}^{\outputParticleIndex_\nTimeSteps}}, \particle_{\nTimeSteps}^{\outputParticleIndex_\nTimeSteps})\Potential_{\nTimeSteps}(\particle_{\nTimeSteps}^{\outputParticleIndex_\nTimeSteps})}{\sum_{\particleIndexAlt=0}^\nParticles \mutation_{\nTimeSteps}(\particle_{\nTimeSteps-1}^{a_{\nTimeSteps - 1}^\particleIndexAlt}, \particle_{\nTimeSteps}^\particleIndexAlt) \Potential_{\nTimeSteps}(\particle_{\nTimeSteps}^\particleIndexAlt)}.}
  \end{align}
   
  \item \label{alg:iterated_rwcsmc:2b} Set $\smash{\OutputParticleIndex_\timeIndex =\outputParticleIndex_\timeIndex \coloneqq a_\timeIndex^{\outputParticleIndex_{\timeIndex+1}}}$, for $t = T-1,\dotsc,1$.
  
    \item \label{alg:iterated_rwcsmc:3} Set $\State_{1:\nTimeSteps}'\coloneqq \state_{1:T} \coloneqq (\particle_1^{\outputParticleIndex_1}, \dotsc, \particle_\nTimeSteps^{\outputParticleIndex_\nTimeSteps})$. 
    
  \item \label{alg:iterated_rwcsmc:4} Return $\state_{1:\nTimeSteps}[l+1] \coloneqq \state_{1:T}'$. 
 \end{enumerate}
\end{framedAlgorithm}
\end{flushleft}
}

\glsreset{RWCSMC}
Step~\ref{alg:iterated_rwcsmc:1} of Algorithm~\ref{alg:iterated_rwcsmc} which (a) performs (conditional) multinomial resampling by drawing the ancestor indices $A_t^n$ and (b) generates the particles $\smash{\Particle_t^n}$ by scattering them around the reference particle $\smash{\particle_t^0 = \state_t}$ in the same way as Algorithm~\ref{alg:iterated_rwehmm}, will be referred to as the \emph{\gls{RWCSMC}} algorithm.

 Step~\ref{alg:iterated_rwcsmc:2a} samples a final-time particle index $K_T = k_T$ with probability proportional to the $k_T$th particle weight at time $T$ and Step~\ref{alg:iterated_rwcsmc:2b} traces back the associated ancestral lineage.


The following proposition (proved in Appendix~\ref{app:subsec:rwcsmc_invariance}) shows that the \gls{IRWCSMC} algorithm can be viewed as a `perturbed' version of the \gls{RWEHMM} algorithm. To our knowledge, this insight is novel.

\begin{proposition}\label{prop:discrete_markov_kernel_rwcsmc_without_backward_sampling}
 The combination of Steps~\ref{alg:iterated_rwcsmc:1:a}, \ref{alg:iterated_rwcsmc:2a} and \ref{alg:iterated_rwcsmc:2b} of Algorithm~\ref{alg:iterated_rwcsmc} induces a $\smash{\xi_T(\particle_{1:T}, \ccdot)}$-invariant Markov kernel. \tendmark
\end{proposition}

We continue the running example from Section~\ref{sec:csmc} which shows that the algorithms analysed in this work reduce to versions of well known classical \gls{MCMC} kernels if $\nParticles = \nTimeSteps  = 1$.

\begin{example}[classical \gls{MCMC} kernels, continued]\label{ex:barker_as_special_case_of_rwcsmc}
 Algorithm~\ref{alg:iterated_rwcsmc} propoes $\Particle_1^1 = \particle_1^1 \sim \dN(\particle_1^0, \tfrac{\ell_1}{\nDimensions} \iMat_{\nDimensions})$, where $\particle_1^0 = \state_1[l]$, which is then accepted as the new state of the Markov chain with probability
 \begin{align}
   \selectionFunctionBoltzmann{1}(\logWeightRandomWalk_{1}(\particle_{1}^1) - \logWeightRandomWalk_1(\particle_{1}^0)) = 
  \frac{\mutation_1(\particle_1^1) \Potential_1(\particle_1^1)}{\mutation_1(\particle_1^0)\Potential_1(\particle_1^0) + \mutation_1(\particle_1^1) \Potential_1(\particle_1^1)}.
 \end{align}
 This can be recognised as Barker's kernel \citep{barker1965monte} with a Gaussian random-walk proposal. Note that the symmetry property from \eqref{eq:rwcsmc_symmetric_proposal:0} ensures that the proposal density cancels out in the acceptance probability. \tendmark
\end{example}

\begin{example}[multi-proposal \gls{MCMC} kernels, continued]
 For $N > 1$ and $T = 1$, Algorithm~\ref{alg:iterated_rwcsmc} is again a special case of a class of \gls{MCMC} algorithms with multiple proposals \citep{tjelmeland2004using}. Related algorithms were analysed in \citet{bedard2012scaling, bedard2013empirical} who also proved scaling limits in high dimensions. \tendmark
\end{example}

\subsubsection{Extensions}

The extensions discussed for the standard \gls{CSMC} algorithm in Section~\ref{subsec:csmc_extensions} can be used for the \gls{IRWCSMC} algorithms with only minor modifications. 

\paragraph*{Forced move} To use the forced-move approach, we simply replace the Boltzmann selection function by the Rosenbluth--Teller selection function in Step~\ref{alg:iterated_rwcsmc:2a} of Algorithm~\ref{alg:iterated_rwcsmc}.

\begin{example}[classical \gls{MCMC} kernels, continued]\label{ex:mh_as_special_case_of_rwcsmc}
 With the forced-move extension, Algorithm~\ref{alg:iterated_rwcsmc} proposes $\Particle_1^1 = \particle_1^1 \sim \dN(\particle_1^0, \tfrac{\ell_1}{\nDimensions} \iMat_{\nDimensions})$, where $\particle_1^0 = \state_1[l]$, which is then accepted as the new state of the Markov chain with probability
 \begin{align}
 \selectionFunctionRosenbluth{1}(\logWeightRandomWalk_{1}(\particle_{1}^1) - \logWeightRandomWalk_1(\particle_{1}^0)) = 
  1 \wedge \frac{\mutation_1(\particle_1^1) \Potential_1(\particle_1^1)}{\mutation_1(\particle_1^0)\Potential_1(\particle_1^0)}.
 \end{align}
 This can be recognised as a \gls{MH} kernel \citep{metropolis1953equation, hastings1970monte} with a Gaussian random-walk proposal. Again, the symmetry property from \eqref{eq:rwcsmc_symmetric_proposal:0} ensures that the proposal density cancels out in the acceptance ratio. \tendmark
\end{example}

\paragraph*{Backward sampling} To employ backward sampling, we sample each particle index $\OutputParticleIndex_\timeIndex =\outputParticleIndex_\timeIndex \in [\nParticles]_0$ in Step~\ref{alg:iterated_rwcsmc:2b} of Algorithm~\ref{alg:iterated_rwcsmc} with probability
\begin{align} 
  \MoveEqLeft \selectionFunctionBoltzmann{\outputParticleIndex_\timeIndex}(\{
  \logBackwardWeightRandomWalk_\timeIndex(\particle_{\timeIndex-1}^{\mathrlap{a_{\timeIndex-1}^{\smash{\particleIndexAlt}}}}, \particle_{\timeIndex}^\particleIndexAlt, \particle_{\timeIndex+1}^{\outputParticleIndex_{\timeIndex+1}}) - \logBackwardWeightRandomWalk_\timeIndex(\particle_{\timeIndex-1}^0, \particle_{\timeIndex}^0, \particle_{\timeIndex+1}^{\outputParticleIndex_{\timeIndex+1}})
  \}_{\particleIndexAlt = 1}^\nParticles)\\
  & =
  \frac{\mutation_\timeIndex(\particle_{\timeIndex - 1}^{a_{\timeIndex - 1}^{\smash{\outputParticleIndex_\timeIndex}}}, \particle_\timeIndex^{\outputParticleIndex_\timeIndex}) \Potential_\timeIndex(\particle_\timeIndex^{\outputParticleIndex_\timeIndex}) \mutation_{\timeIndex+1}(\particle_\timeIndex^{\outputParticleIndex_\timeIndex}, \particle_{\timeIndex+1}^{\outputParticleIndex_{\timeIndex+1}})}{\sum_{\particleIndexAlt=0}^\nParticles \mutation_\timeIndex(\particle_{\timeIndex - 1}^{\mathrlap{a_{\timeIndex - 1}^\particleIndexAlt}}, \particle_\timeIndex^\particleIndexAlt) \Potential_\timeIndex(\particle_\timeIndex^\particleIndexAlt) \mutation_{\timeIndex+1}(\particle_\timeIndex^\particleIndexAlt, \particle_{\timeIndex+1}^{\outputParticleIndex_{\timeIndex+1}})},
\end{align}
where
\begin{align}
 \smash{\logBackwardWeightRandomWalk_\timeIndex(\state_{\timeIndex-1:\timeIndex+1}) \coloneqq \logWeightRandomWalk_\timeIndex(\state_{\timeIndex-1:\timeIndex}) + \log \mutation_{\timeIndex+1}(\state_\timeIndex, \state_{\timeIndex+1}).}
\end{align}

The following proposition (proved in Appendix~\ref{app:subsec:rwcsmc_invariance}) shows that when backward sampling is used, the \gls{IRWCSMC} algorithm can again be viewed as a `perturbed' version of the \gls{RWEHMM} algorithm. To our knowledge, this insight is novel.

\begin{proposition}\label{prop:discrete_markov_kernel_rwcsmc_with_backward_sampling}
 Proposition~\ref{prop:discrete_markov_kernel_rwcsmc_without_backward_sampling} remains valid if backward sampling is used, i.e.\ the combination of Steps~\ref{alg:iterated_rwcsmc:1:a}, \ref{alg:iterated_rwcsmc:2a} and \ref{alg:iterated_rwcsmc:2b} of Algorithm~\ref{alg:iterated_rwcsmc} again induces a $\smash{\xi_T(\particle_{1:T}, \ccdot)}$-invariant Markov kernel. \tendmark
\end{proposition}

\subsubsection{Induced \texorpdfstring{$\Target_{\nTimeSteps,\nDimensions}$-}{}invariant Markov kernel}


Given $\smash{\State_{1:T} = \state_{1:\nTimeSteps} = \state_{1:\nTimeSteps}[l]}$, let
\begin{align}
 \smash{\IteratedCsmcKernelRandomWalk{\nTimeSteps,\nDimensions,\state_{1:T}}{\nParticles}(\diff \particle_{1:\nTimeSteps} \times \diff a_{1:\nTimeSteps-1} \times \diff\outputParticleIndex_{1:\nTimeSteps} \times \diff \state_{1:T}')}
\end{align}
be the law of all the random variables $\smash{(\Particle_{1:T}, A_{1:T-1}, K_{1:T}, \State_{1:T}')}$ generated in Steps~\ref{alg:iterated_rwcsmc:1}--\ref{alg:iterated_rwcsmc:3} of Algorithm~\ref{alg:iterated_rwcsmc} (with or without the forced-move extension and with or without backward sampling). Appendix~\ref{app:subsec:joint_law_rwcsmc} gives a more formal definition of this law.

Let $\smash{\ExpectationCsmcKernelRandomWalk{\nTimeSteps,\nDimensions,\state_{1:\nTimeSteps}}{\nParticles}}$ denote expectation \WRT $\smash{\IteratedCsmcKernelRandomWalk{\nTimeSteps,\nDimensions,\state_{1:T}}{\nParticles}}$. Algorithm~\ref{alg:iterated_rwcsmc} induces a Markov kernel 
\begin{align}
 \smash{\InducedIteratedCsmcKernelRandomWalk{\nTimeSteps,\nDimensions}{\nParticles}(\state_{1:\nTimeSteps}, \diff \state_{1:T}') 
 \coloneqq \ExpectationCsmcKernelRandomWalk{\nTimeSteps,\nDimensions,\state_{1:\nTimeSteps}}{\nParticles}[\ind\{\State_{1:T}' \in \diff \state_{1:T}'\}],}
 \label{eq:iterated_rwcsmc_kernel}
\end{align}
for $\smash{(\state_{1:\nTimeSteps},\diff \state_{1:T}') \in \spaceState_{T,D} \times \sigFieldState_{T,D}}$. The following proposition shows that this Markov kernel leaves $\Target_{\nTimeSteps, \nDimensions}$ invariant. It can be proved by interpreting the \gls{IRWCSMC} kernel as a special case of the generic iterated \gls{CSMC} approach described in \citet{finke2016embedded}. For completeness, we provide a simple proof in Appendix~\ref{app:subsec:rwcsmc_invariance}.  

\begin{proposition}\label{prop:rwcsmc_invariance}
 For any $\smash{N, T, D \in \naturals}$, $\smash{\Target_{T,D}\InducedIteratedCsmcKernelRandomWalk{\nTimeSteps,\nDimensions}{\nParticles} = \Target_{T,D}}$. \tendmark
\end{proposition}

We stress that this proposition does not require the high-dimensional regime from Assumption~\ref{as:iid_model}. That is, Algorithm~\ref{alg:iterated_rwcsmc} induces a valid (i.e. $\Target_{\nTimeSteps, \nDimensions}$-invariant) Markov kernel even if the model does not factorise into $\nDimensions$ \gls{IID} components.

\begin{remark}\label{rem:relationship_with_unconditional_smc}
 As explained in \citet{andrieu2010particle}, the (standard) \gls{CSMC} algorithm is closely linked to the justification of a corresponding `unconditional' \gls{SMC} algorithm. However, for the \gls{RWCSMC} algorithm, no such `unconditional' \gls{SMC} counterpart exists. We expand on this in Appendix~\ref{app:subsec:relationship_with_unconditional_smc}. \tendmark
\end{remark}

%
%
%

\subsection{Stability in high dimensions}
\label{subsec:rwcsmc_scaling}

In this section, prove that the \gls{IRWCSMC} algorithm is dimensionally stable. Throughout this section, we assume that the model follows the high-dimensional regime from Assumption~\ref{as:iid_model}. Such assumptions are common in the literature on optimal scaling \citep{roberts1997weak}. However, we stress that the algorithm is agnostic to this structure, i.e.\ it does not exploit the fact that the target distribution factorises. We thus expect such results to hold more generally.

\subsubsection{Non-degenerate limiting law of the genealogies}
\label{subsec:non-degnerate_limiting_genealogies}

In the following, we show that as as $\nDimensions \to \infty$, the law of the genealogies (and the indices of the new reference path) induced by the \gls{IRWCSMC} algorithm converges to a limit that is non-degenerate in the sense that the particle lineages do not necessarily coalesce with the reference path (and, likewise, the indices of the new reference path do not necessarily coincide with those of the old reference path). Define
\begin{align}
 \smash{\logWeightRandomWalkSingle_t = \log G_t + \log m_t,}
\end{align}
with the convention $\logWeightRandomWalkSingle_{T+1} \equiv 0$. We again use the convention that $\partial_t^i$ denotes the $i$th derivative w.r.t. $\stateSingle_t$ and with $\partial_t \coloneqq \partial_t^1$, and we write $\pi_T(\varphi) \coloneqq \smash{\int_{\reals^T} \varphi(x_{1:T}) \pi_T(x_{1:T}) \intDiff x_{1:T}}$, for any $\pi_T$-integrable function $\smash{\varphi: \reals^T \to \reals}$. With these conventions, we make the following moment assumptions which are similar to the assumptions in \citet{bedard2012scaling}, and which are assumed to hold for any $t \in [T]$.
\begin{enumerate}[label=\textbf{C\arabic*}, series=rwcsmc]
  \item \label{as:moments_bounded} $\logWeightRandomWalkSingle_t$ is twice continuously differentiable and
  \begin{itemize}
   \item $\partial_t^2 \logWeightRandomWalkSingle_t$, $\partial_t^2 \logWeightRandomWalkSingle_{t+1}$, $\partial_t \partial_{t+1} \logWeightRandomWalkSingle_{t+1}$ are Lipschitz-continuous and bounded,
   \item $\pi_T(\lvert \partial_t \logWeightRandomWalkSingle_t \rvert^4), \pi_T(\lvert \partial_t \logWeightRandomWalkSingle_{t+1}\rvert^4) < \infty$. \tendmark
  \end{itemize}
\end{enumerate}
Before stating the result, we define a law $\smash{\IteratedCsmcKernelRandomWalk{\nTimeSteps}{\nParticles}(\diff v_{1:\nTimeSteps} \times \diff w_{1:\nTimeSteps} \times \diff a_{1:\nTimeSteps-1} \times \diff k_{1:T})}$. This is a joint distribution of random variables $(V_{1:T}, W_{1:T}, A_{1:T-1}, K_{1:T})$, where $A_{1:T-1}$ and $K_{1:T}$ are collections of ancestor and particle indices as in the fixed-dimensional case and where $V_t \coloneqq V_t^{1:N}$ and $W_t \coloneqq W_t^{1:N}$ are each $N$-dimensional Gaussian vectors which are such that $(V_s, W_s)$ and $(V_t, W_t)$ are independent whenever $s \neq t$ and such that
\begin{align}
  \E
  \begin{bmatrix}
   V_t\\
   W_t
  \end{bmatrix}
  = \tfrac{1}{2}\ell_t
  \begin{bmatrix}
   \pi_T(\partial_t^2 \logWeightRandomWalkSingle_t) \unitMat_N\\\
   \pi_T(\partial_t^2 \logWeightRandomWalkSingle_{t+1}) \unitMat_N
  \end{bmatrix}
   \eqqcolon \bar{\mu}_{t|T},
\end{align}
and, recalling that $\varSigma = \tfrac{1}{2}(\iMat_N + \unitMat_N \unitMat_N^\T)$,
\begin{align}
 \var
 \begin{bmatrix}
  V_t\\
  W_t
 \end{bmatrix}
  = \ell_t
  \begin{bmatrix}
   \pi_T([\partial_{t} \logWeightRandomWalkSingle_t]^2)  \varSigma & \pi_T([\partial_{t} \logWeightRandomWalkSingle_t][\partial_{t} \logWeightRandomWalkSingle_{t+1}]) \varSigma\\
   \pi_T([\partial_{t} \logWeightRandomWalkSingle_t][\partial_{t} \logWeightRandomWalkSingle_{t+1}]) \varSigma & \pi_T([\partial_{t} \logWeightRandomWalkSingle_{t+1}]^2) \varSigma
  \end{bmatrix}
  \eqqcolon \widebar{\varSigma}_{t|T}.
\end{align}
We will also use the convention that $V_t^0 = W_t^0 \equiv 0$ for any $t \in [T]$ and $W_{0}^n \equiv 0$ for any $n \in [N]$.

Note that under Assumption~\ref{as:moments_bounded}, by Lemma~\ref{lem:integration_by_parts_identity} in Appendix~\ref{app:subsec:prop:stability_of_acceptance_rates_rwcsmc},
  \begin{align}
   \calI_{t|T} & \coloneqq \smash{\pi_T([\partial_t \log \pi_t]^2)= - \pi_T(\partial_t^2 \log \pi_T)}\\
   & = \smash{\var[V_t^n + W_t^n]/\ell_t = - 2 \E[V_t^n + W_t^n]/\ell_t < \infty.}
  \end{align}

We are now ready to state the limiting law of the genealogies. Throughout, we assume \ref{as:moments_bounded}. Then we can define the law $\smash{\IteratedCsmcKernelRandomWalk{\nTimeSteps}{\nParticles}(\diff v_{1:\nTimeSteps} \times \diff w_{1:\nTimeSteps} \times \diff a_{1:\nTimeSteps-1} \times \diff k_{1:T})}$ by the following sampling procedure (a formal definition is given in Appendix~\ref{app:subsec:formal_definition_of_the_limiting_law}).
\begin{enumerate}
 \item For $t \in [T]$, sample $\smash{(V_t, W_t) = (v_t, w_t) \sim \dN(\bar{\mu}_{t|T}, \widebar{\varSigma}_{t|T})}$.
 \item For $t \in [T-1]$,
  \begin{enumerate}
   \item set $\smash{A_t^0 = a_t^0 \coloneqq 0}$,
   \item sample $\smash{A_t^n = a_t^n = l \in [N]_0}$ with probability $\smash{\selectionFunctionBoltzmann{l}(\{v_t^m + w_{t-1}^{a_{t-1}^m}\}_{\particleIndexAlt = 1}^\nParticles)}$, independently for $n \in [N]$.
  \end{enumerate}
  \item \label{enum:rwcsmc_limiting_law:3} Sample $\smash{K_T = k_T\in [N]_0}$ with probability $\smash{\selectionFunctionBoltzmann{k_T}(\{v_T^m + w_{T-1}^{a_{T-1}^m}\}_{\particleIndexAlt = 1}^\nParticles)}$.
  \item \label{enum:rwcsmc_limiting_law:4} For $t = T-1, \dotsc, 1$, set $\smash{K_t = k_t \coloneqq a_t^{k_{t+1}}}$.
\end{enumerate}
As usual, if we use the forced-move extension, we must replace $\selectionFunctionBoltzmann{\particleIndex}$ by $\selectionFunctionRosenbluth{\particleIndex}$ in Step~\ref{enum:rwcsmc_limiting_law:3}. Likewise, if we use backward sampling, we must instead sample $\smash{K_t = k_t \in[N]_0}$ with probability $\smash{\selectionFunctionBoltzmann{k_t}(\{v_t^m + w_t^m + w_{t-1}^{a_{t-1}^m}\}_{\particleIndexAlt = 1}^\nParticles)}$ in Step~\ref{enum:rwcsmc_limiting_law:4}. Hereafter, $\ExpectationCsmcKernelRandomWalk{\nTimeSteps}{\nParticles}$ denotes expectation \WRT $\IteratedCsmcKernelRandomWalk{\nTimeSteps}{\nParticles}$.

\subsubsection{Convergence to the non-degenerate limit}

Proposition~\ref{prop:limiting_rwcsmc_algorithm}, proved in Appendix~\ref{app:subsec:prop:limiting_rwcsmc_algorithm}, shows that in high dimensions, the law of the genealogies and the indices of the new reference path under the \gls{IRWCSMC} update specified in Algorithm~\ref{alg:iterated_rwcsmc} converges to the limiting law specified above. Here again, $\lVert \ccdot \rVert$ denotes the total variation distance.

\begin{proposition}[convergence of the law of the genealogies]\label{prop:limiting_rwcsmc_algorithm}
 Let $T, N \in \naturals$, assume \ref{as:iid_model} as well as \ref{as:moments_bounded}, and write
  \begin{align}
  \!\!\!\!\! \smash{\bar{d}_{T,D,\state_{1:T}}^N \coloneqq \bigl\lVert \ExpectationCsmcKernelRandomWalk{T,D,\state_{1:T}}{\nParticles}[\ind\{(A_{1:\nTimeSteps-1}, \OutputParticleIndex_{1:\nTimeSteps}) \in \ccdot\}] - \ExpectationCsmcKernelRandomWalk{\nTimeSteps}{\nParticles}[\ind\{(A_{1:\nTimeSteps-1}, \OutputParticleIndex_{1:\nTimeSteps}) \in \ccdot\}]\bigr\rVert.}\!\!\!\!\!
 \end{align}
 Then there exists a family $\smash{\ConcentrationSet_{\nTimeSteps, \nDimensions} \in \sigFieldState_{T,D}}$ with $\smash{\lim_{\nDimensions \to \infty} \Target_{\nTimeSteps, \nDimensions}(\ConcentrationSet_{\nTimeSteps, \nDimensions}) = 1}$ and
 \[
  \smash{\lim\nolimits_{\nDimensions \to \infty} \sup\nolimits_{\state_{1:\nTimeSteps} \in \ConcentrationSet_{\nTimeSteps, \nDimensions}} \smash{\bar{d}_{T,D,\state_{1:T}}^N} = 0.} \mendmark
 \]
%
%
\end{proposition}

The following corollary shows that the \emph{acceptance rate} at any time~$t$ associated with Algorithm~\ref{alg:iterated_rwcsmc},
\begin{align}
 \smash{\acceptanceRateRwCsmc{T,D,\state_{1:T}}{\nParticles}(t) 
  \coloneqq
   \ExpectationCsmcKernelRandomWalk{\nTimeSteps,\nDimensions,\state_{1:T}}{\nParticles}[ \ind\{\OutputParticleIndex_t \neq 0\}],}
\end{align}
converges to a strictly positive limit
\begin{align}
 \smash{\acceptanceRateRwCsmc{\nTimeSteps}{\nParticles}(t) 
 \coloneqq
   \ExpectationCsmcKernelRandomWalk{\nTimeSteps}{\nParticles}[ \ind\{\OutputParticleIndex_t \neq 0\}].}
\end{align}
Note that the acceptance rates and their limits depend on $\ell_{1:T}$ even though we do not make this explicit in our notation. 

\begin{corollary}[dimensional stability of the acceptance rates] \label{cor:dimensional_stability_of_acceptance_rates_rwcsmc} 
 Assume \ref{as:iid_model} and \ref{as:moments_bounded}, $T, N \in \naturals$, and let $t \in [T]$. Then $\smash{ \acceptanceRateRwCsmc{\nTimeSteps}{\nParticles}(t) > 0}$. Furthermore, with $\ConcentrationSet_{\nTimeSteps, \nDimensions} \in \sigFieldState_{T,D}$ as in Proposition~\ref{prop:limiting_rwcsmc_algorithm}:
 \[
   \smash{\lim\nolimits_{\nDimensions \to \infty} \sup\nolimits_{\state_{1:\nTimeSteps} \in \ConcentrationSet_{\nTimeSteps, \nDimensions}} \lvert \acceptanceRateRwCsmc{\nTimeSteps,\nDimensions,\state_{1:T}}{\nParticles}(t) - \acceptanceRateRwCsmc{\nTimeSteps}{\nParticles}(t)\rvert = 0.} \mendmark
  \]
 \end{corollary}
\begin{proof}
 The convergence follows immediately from Proposition~\ref{prop:limiting_rwcsmc_algorithm}. The strict positivity of the limit is due to the finite-moments assumption \ref{as:moments_bounded}. 
\end{proof}


\begin{example}[classical MCMC kernels, continued]
 As mentioned above, Algorithm~\ref{alg:iterated_rwcsmc} reduces to a classical \gls{MCMC} kernel with a suitably scaled Gaussian random-walk proposal if $T = N = 1$. In this case, we asymptotic acceptance rates for Barker's kernel and for the \gls{MH} kernel derived in \citet{roberts1997weak, bedard2012scaling, agraval2021optimal}:
 \begin{align}
  \acceptanceRateRwCsmc{1}{1}(1)
   = 
  \begin{dcases}
   \E[\selectionFunctionBoltzmann{1}(V_1^1)] = 
   \E\bigl[\tfrac{\exp\{V_1^1\}}{1 +  \exp\{V_1^1\}}\bigr], 
    & \text{without forced-move,\!\!\!\!\!\!\!\!\!\!\!\!\!}\\[-1ex]
   \E[\selectionFunctionRosenbluth{1}(V_1^1)]
   = \E[1 \wedge \exp\{V_1^1\}]
  = 2 \smash{\standardNormalCdf\bigl(- \tfrac{\sqrt{\ell_1 \calI_{1|1}}}{2}\bigr)}, & \text{with forced-move,\!\!\!\!\!\!\!}
  \end{dcases}
 \label{eq:asymptotic_acceptance_rate_of_rwmh}
 \end{align}
 where $\standardNormalCdf$ is standard-normal cumulative distribution function. \tendmark
\end{example}

\glsreset{ESJD}

Of course, stabilising the acceptance rates in high dimensions is not sufficient for avoiding a breakdown. A widely used criterion for assessing the performance of \gls{MCMC} algorithms is the \emph{\gls{ESJD}} \citep{sherlock2009optimal}, which (for the time-$t$ component in Algorithm~\ref{alg:iterated_rwcsmc}) is given by
\begin{align}
 \smash{\esjdRandomWalk_{T,D}^N(t) \coloneqq \E[\lVert \State_{t}[l+1] - \State_{t}[l] \rVert_2^2],}
\end{align}
where $\lVert \ccdot \rVert_2$ denotes the Euclidean norm and where $\State_{1:T}[l]$ is the $l$th state of the Markov chain with transition kernel $\smash{\InducedIteratedCsmcKernelRandomWalk{T,D}{N}}$ at stationarity. The following proposition (whose proof is the same as that of Proposition~\ref{prop:esjd_rwehmm} in Appendix~\ref{app:subsec:prop:esjd_rwehmm} and is therefore omitted) shows that the \gls{ESJD} is stable in high dimensions.

\begin{proposition}[dimensional stability of the \gls{ESJD}]\label{prop:esjd_rwcsmc}
 Assume \ref{as:iid_model} as well as \ref{as:moments_bounded}, and let $\nTimeSteps, \nParticles \in \naturals$. Then, 
 for any $t \in [T]$,
 \[
  \smash{\lim\nolimits_{\nDimensions \to \infty} \bigl\lvert 
  \esjdRandomWalk_{T,D}^N(t)
  -  \ell_t \acceptanceRateRwCsmc{\nTimeSteps}{\nParticles}(t)\bigr\rvert = 0.} \mendmark
 \]
\end{proposition}

\subsubsection{Stability \texorpdfstring{as $T \to \infty$}{in the time horizon}}
\label{subsec:large_time-horizon_asymptotics}

In this section, we discuss scaling of the number of particles, $N+1$, in the time horizon, $T$, in high (spatial) dimensions. Throughout, we let $\ell \coloneqq (\ell_\timeIndex)_{\timeIndex \geq 1}$ be a sequence of positive scaling factors.

Under stability assumptions on the Feynman--Kac model and for some fixed spatial dimension $D$, it is well known that the acceptance rates of the \gls{ICSMC} algorithm,  $\acceptanceRate{\nTimeSteps,D,\state_{1:T}}{\nParticles}(t)$, can be bounded away from zero as $\nTimeSteps  \to \infty$ by scaling the number of particles appropriately. To be more specific:
 Without backward sampling, $\acceptanceRate{\nTimeSteps,\nDimensions,\state_{1:T}}{\nParticles}(t)$ can be controlled by growing $\nParticles$ linearly in $\nTimeSteps$ \citep{andrieu2018uniform, lindsten2015uniform, delmoral2016particle}.
 With backward sampling, $\acceptanceRate{\nTimeSteps,\nDimensions,\state_{1:T}}{\nParticles}(t)$ can be controlled without scaling $\nParticles$ with $\nTimeSteps$ \citep{lee2020coupled}.
Proposition~\ref{prop:stability_of_acceptance_rates} (proved in Appendix~\ref{app:subsec:prop:stability_of_acceptance_rates_rwcsmc}) verifies that -- under the following `factorisation-over-time' assumption -- the \gls{IRWCSMC} algorithm admits the same scaling of $N$ with $T$ as the \gls{ICSMC} algorithm.

\begin{enumerate}[label=\textbf{A\arabic*}, resume=model_assumptions]
 \item \label{as:independent_over_time} For any $t \in [T]$ and any $(x_{t-1}, x_{t-1}') \in \reals^2$, $M_t(x_{t-1}, \ccdot) = M_t(x_{t-1}', \ccdot)$. \tendmark
\end{enumerate}

\begin{proposition}[time-horizon stability of the acceptance rates] \label{prop:stability_of_acceptance_rates} Assume \ref{as:iid_model}, \ref{as:independent_over_time} as well as \ref{as:moments_bounded} and that $\calI(\ell) \coloneqq \inf_{T \in \naturals}\inf_{t \in [T]} \ell_t \calI_{t|T} < \infty$.
 \begin{enumerate}
  \item \label{prop:stability_of_acceptance_rates:without_backward_sampling} Without backward sampling, if there exists $C > 0$ such that $\nParticles \geq C \nTimeSteps$:
  \begin{align}
    \inf_{\nTimeSteps\in \naturals} \inf_{t \in [T]}  \acceptanceRateRwCsmc{\nTimeSteps}{\nParticles}(t) \geq \exp\biggl(-\frac{\exp(\calI(\ell))}{C}\biggr) > 0.
  \end{align}
  
  \item \label{prop:stability_of_acceptance_rates:with_backward_sampling} With backward sampling, for any $N \in \naturals$:
  \[
      \inf_{\nTimeSteps\in \naturals} \inf_{t \in [T]}  \acceptanceRateRwCsmc{\nTimeSteps}{\nParticles}(t) \geq \biggl(1 + \frac{\exp(\calI(\ell))}{N}\biggr)^{\mathrlap{-1}} > 0. \mendmark
  \]
 \end{enumerate}
\end{proposition}

\begin{remark}
 Proposition~\ref{prop:esjd_rwcsmc} shows that $\smash{\inf_{T \in \naturals} \inf_{t \in [T]} \esjdRandomWalk_{T,D}^N(t) > 0}$ (i.e.\ the \gls{ESJD} is also stable) under the additional assumption $\inf_{t \geq 1} \ell_t > 0$. \tendmark
\end{remark}

Note that the lower bound for the case with backward sampling in Proposition~\ref{prop:stability_of_acceptance_rates} is the same as the one obtained for the \gls{RWEHMM} algorithm in Corollary~\ref{cor:stability_of_acceptance_rates_ehmm}. This is not a coincidence: under  Assumption~\ref{as:independent_over_time}, Algorithm~\ref{alg:iterated_rwcsmc} with backward sampling induces the same Markov kernel as Algorithm~\ref{alg:iterated_rwehmm}, $\smash{\InducedIteratedCsmcKernelRandomWalk{\nTimeSteps,\nDimensions}{\nParticles}}=\smash{\InducedIteratedEhmmKernelRandomWalk{\nTimeSteps,\nDimensions}{\nParticles}}$. However, recall that in Corollary~\ref{cor:stability_of_acceptance_rates_ehmm}, we were able to prove stability of the acceptance rates in $T$ without relying on Assumption~\ref{as:independent_over_time}. This, along with Propositions~\ref{prop:discrete_markov_kernel_rwcsmc_without_backward_sampling} and \ref{prop:discrete_markov_kernel_rwcsmc_with_backward_sampling} (which show that the \gls{IRWCSMC} algorithm can be viewed as a `perturbed' version of the \gls{RWEHMM} algorithm) motivates the following conjecture.

\begin{conjecture}\label{conj:independent_over_time_not_necessary}
 Assumption~\ref{as:independent_over_time} is not necessary to guarantee stability of the acceptance rates in $T$ with the scaling of $N = N(T)$ from Proposition~\ref{prop:stability_of_acceptance_rates}. \tendmark
\end{conjecture}

\section{Numerical illustration}
\label{sec:numerical_illustration}

In this section, we illustrate the results on a simple state-space model specified as follows. Let $\standardNormalPdf$ denote a Lebesgue density of the standard normal distribution. Let $\mathbf{y}_{\timeIndex} = (y_{\timeIndex,\dimensionIndex})_{\dimensionIndex \in [D]} \in \reals^{D}$ be some $D$-dimensional vector of observations collected at time $\timeIndex \in [\nTimeSteps]$. Then for any $\dimensionIndex \in [D]$, 
\begin{align}
 \PotentialSingle_\timeIndex(\stateSingle_{\timeIndex,\dimensionIndex}) 
  & \coloneqq \standardNormalPdf(y_{\timeIndex,\dimensionIndex} - \stateSingle_{\timeIndex,\dimensionIndex}),\\
 \mutationSingle_\timeIndex(\stateSingle_{\timeIndex-1,\dimensionIndex}, \stateSingle_{\timeIndex,\dimensionIndex})
  & \coloneqq \standardNormalPdf(\stateSingle_{\timeIndex,\dimensionIndex} - \stateSingle_{\timeIndex-1,\dimensionIndex}).
\end{align}
The results shown below are based on 100 independent runs of each algorithm for each value of $D$; each run uses $L = 25,000$ iterations initialised from stationarity, and each uses a different observation sequence of length $T = 25$ sampled from the model. Throughout, we use $N+1 = 32$ particles. In the \gls{IRWCSMC} algorithm, $\ell_1 = \dotsc = \ell_T = 1$.

\textbf{Figure~\ref{fig:acceptance_rates}} displays the $\Target_{T,D}$-averaged acceptance rates as a function of the time index $t$. More precisely, for $\smash{\State_{1:T} \sim \Target_{T,D}}$, it shows:
\begin{enumerate}
 \item first column: $\smash{\E[\acceptanceRate{\nTimeSteps,\nDimensions,\State_{1:T}}{\nParticles}(\timeIndex)]}$;
 \item second column: $\smash{\E[\acceptanceRateRwCsmc{T,D,\State_{1:T}}{\nParticles}(\timeIndex)]}$.
\end{enumerate}
The upper-left panel shows that for the \gls{ICSMC} algorithm, the acceptance rates vanish in high dimensions. In contrast, the upper-right panel shows that the acceptance rates converge to a non-trivial limit for the \gls{IRWCSMC} algorithm. The first row also shows that, in both algorithms, the acceptance rates are an increasing function of the time index $\timeIndex$ due to the coalescence of the particle paths with the reference path. The second row shows that the backward-sampling extension removes this dependence on the time index and leads to acceptance rates which are stable over time. However, the lower-left panel illustrates that backward sampling does not save the \gls{ICSMC} algorithm in high dimensions -- its acceptance rates still vanish. Additionally, in Appendix~\ref{app:sec:additional_simulation_results}, we illustrate that the \gls{ESS} \citep{kong1994sequential} of the resampling and backward-sampling weights converges to a non-trivial limit $> 1$ under the \gls{IRWCSMC} algorithm, whereas it collapses to $1$ for the \gls{ICSMC} algorithm.

\begin{figure}
 \vspace{-0cm}
 \noindent{}
 \centering  
 \includegraphics{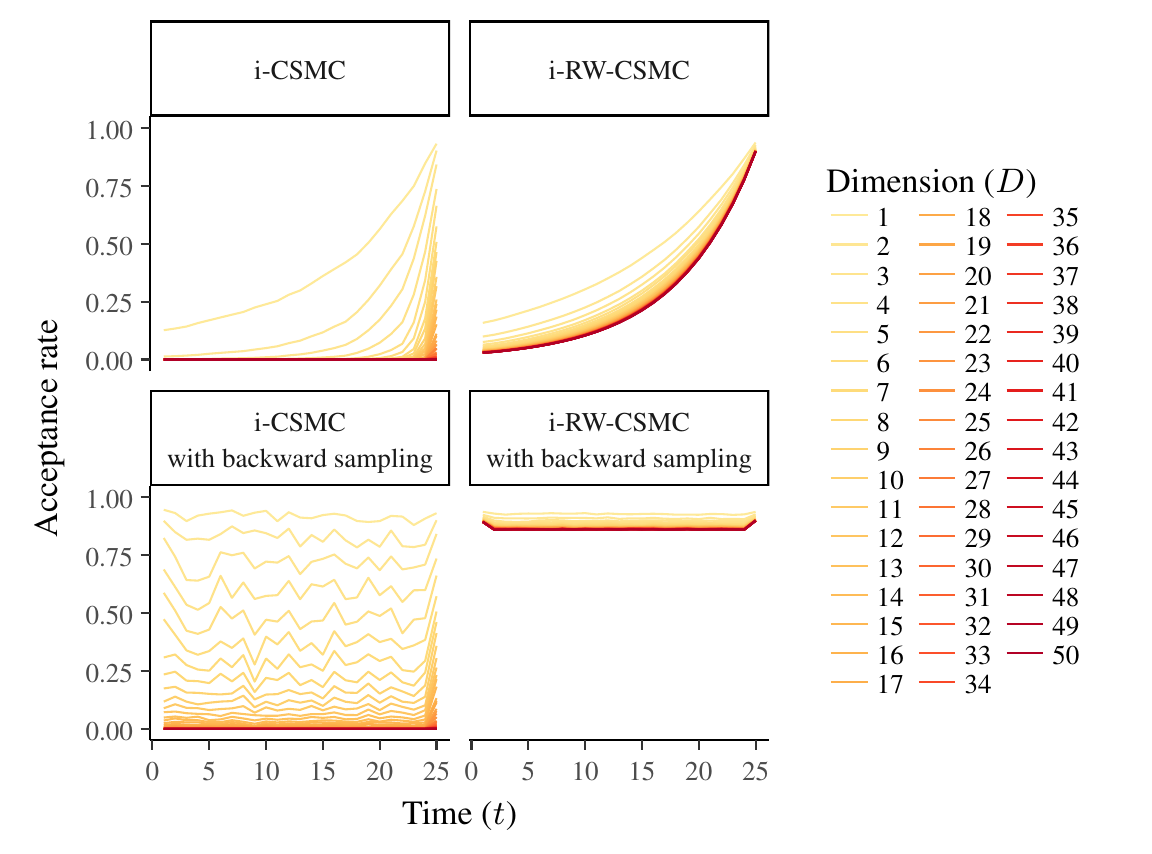}
 \caption{The $\Target_{T,D}$-averaged acceptance rates as a function of $t$.}
 \label{fig:acceptance_rates}
\end{figure}

\glsreset{ESJD}

\textbf{Figure~\ref{fig:esjd}} displays the \gls{ESJD} as a function of $t$. More specifically, it shows:
\begin{enumerate}
 \item first column: $\smash{\esjd_{T,D}^N(t) \coloneqq \E[\lVert \State_{t}[l+1] - \State_{t}[l] \rVert_2^2]}$;
 \item second column: $\smash{\esjdRandomWalk_{T,D}^N(t) \coloneqq \E[\lVert \widebar{\State}_{t}[l+1] - \widebar{\State}_{t}[l] \rVert_2^2]}$,
\end{enumerate}
where $\State_{1:T}[l]$ is the $l$th state of a stationary Markov chain with transition kernels $\InducedIteratedCsmcKernel{T,D}{N}$ and $\smash{\widebar{\State}_{1:T}[l]}$ is the $l$th state of a stationary Markov chain with transition kernels $\smash{\InducedIteratedCsmcKernelRandomWalk{T,D}{N}}$. The first column illustrates that in high dimensions, the \gls{ESJD} of the \gls{ICSMC} algorithm vanishes in high dimensions. In contrast, the \gls{ESJD} of the \gls{IRWCSMC} algorithm converges to $\ell_t \acceptanceRateRwCsmc{T}{N}(t) > 0$ (in accordance with Proposition~\ref{prop:esjd_rwcsmc}).

\begin{figure}
\vspace{-0cm}
 \noindent{}
 \centering
 \includegraphics{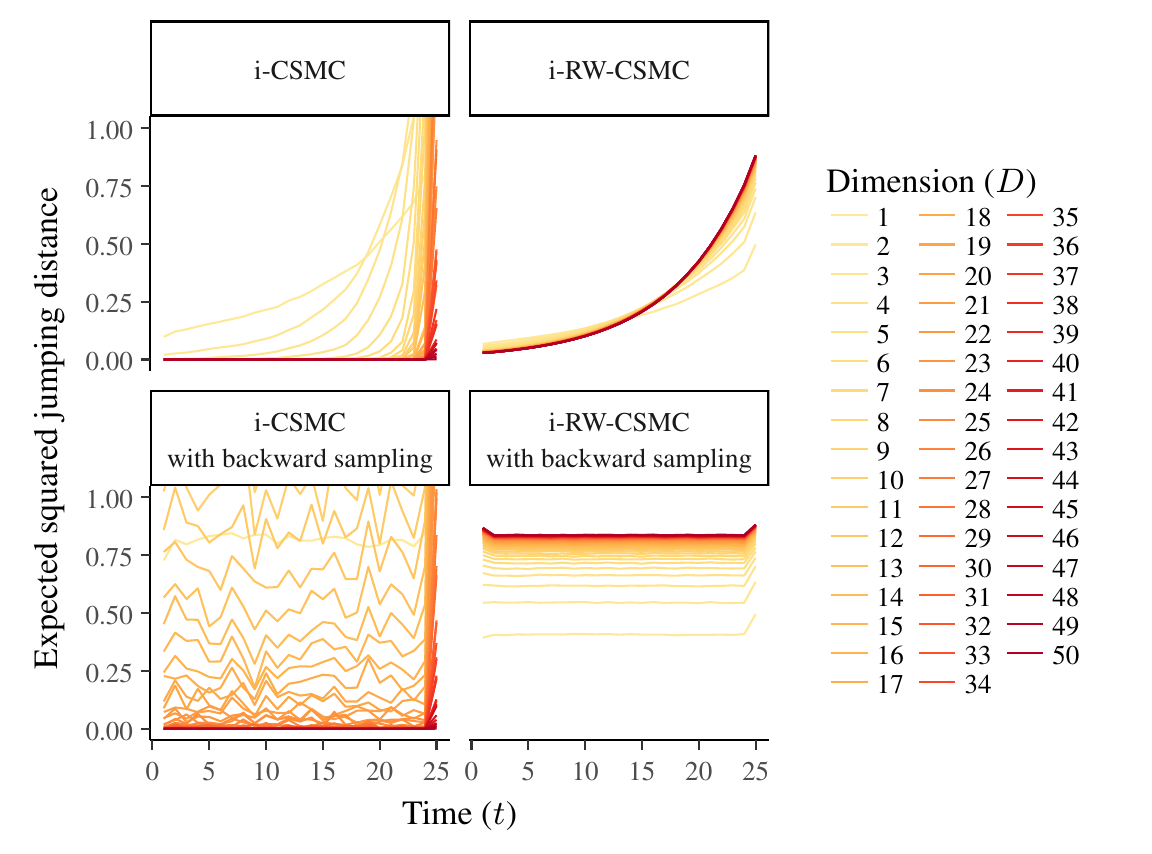}
 \caption{\Glsdesc{ESJD} as a function of $t$.}
 \label{fig:esjd}
\end{figure}

\textbf{Figure~\ref{fig:lag-d_autocorrelations}} displays the lag-$D$ autocorrelation of the sample for the first `spatial' component at each time $t$. More precisely, writing $\State_t[l] = X_{t,1:D}[l]$ and $\smash{\widebar{\State}_t[l] = \widebar{X}_{t,1:D}[l]}$, it shows:
\begin{enumerate}
 \item first column: $\smash{\corr(X_{t,1}[l+D], X_{t,1}[l])}$;
 \item second column: $\smash{\corr(\widebar{X}_{t,1}[l+D], \widebar{X}_{t,1}[l])}$.
\end{enumerate}
The fact that this leads to non-trivial limit in the second column illustrates that the \gls{IRWCSMC} algorithm is stable in high dimensions as long as the number of iterations grows linearly with $D$. However, the first column shows that increasing the number of iterations in this manner does not save the \gls{ICSMC} algorithm from breaking down in high dimensions.

\begin{figure}
 \vspace{-0cm}
 \noindent{}
 \centering
 \includegraphics{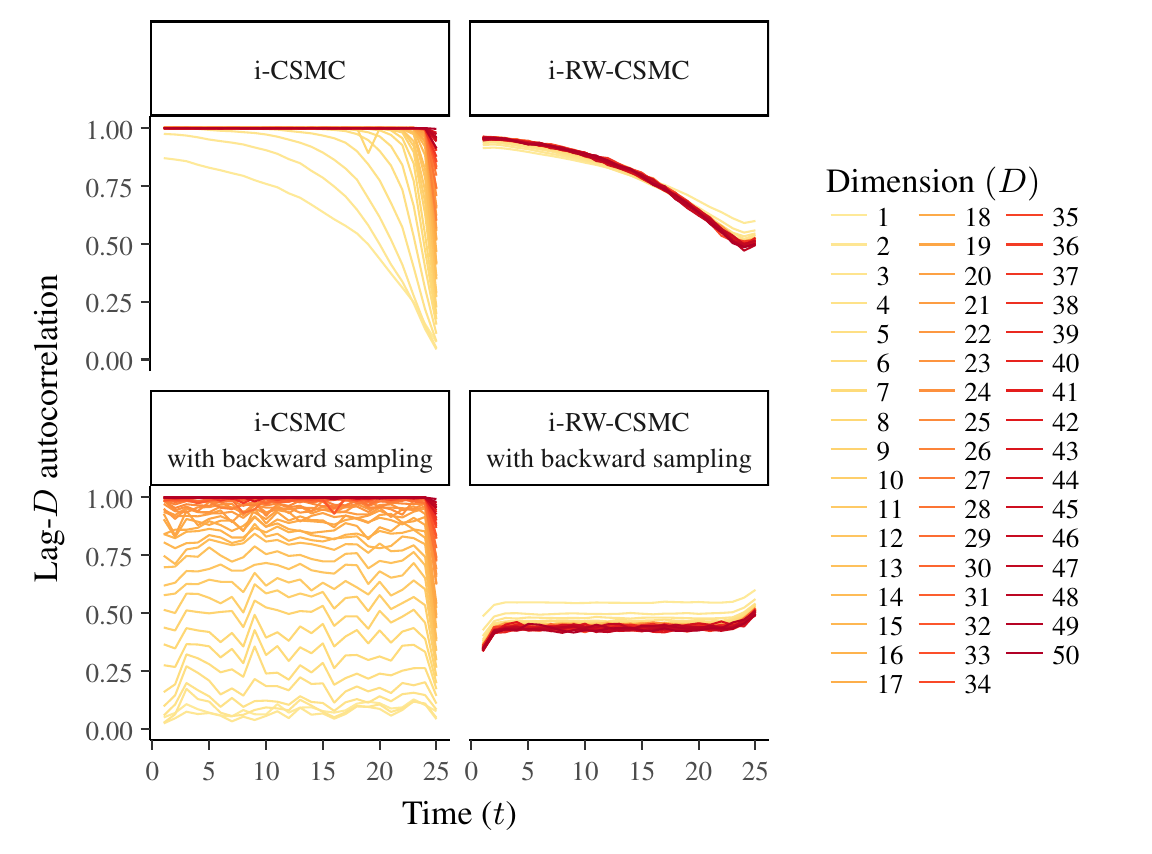}
 \caption{Lag-$\nDimensions$ autocorrelations of the first marginal of the $D$-dimensional time-$t$ state vector as a function of $t$.}
 \label{fig:lag-d_autocorrelations}
\end{figure}


\section{Conclusion}
\glsreset{ICSMC}
\glsreset{MH}
\glsreset{IRWCSMC}

\paragraph*{Comparison with classical MCMC algorithms} The \gls{ICSMC} algorithm \citep{andrieu2010particle} is a powerful tool for inference about the joint smoothing distribution in state-space models and more generally. This is because the algorithm automatically exploits the decorrelation in the `time' direction exhibited by the model. Let $T$ denotes the time horizon and let $D$ denote the `spatial' dimension of each latent state.

For the moment assume that $D$ is fixed and reasonably small.
\begin{itemize}
 \item The \gls{ICSMC} algorithm has $\bo(T)$ complexity\footnote{Recall that `complexity' is measured as the number of full likelihood evaluations needed to control the approximation error of a fixed-dimensional marginal of the joint smoothing distribution.} which is in contrast to a na\"ive independent \gls{MH} algorithm which does not exploit the structure of the model and therefore suffers $\bo(\eul^T)$ complexity, i.e.\ a curse of dimension in $T$.
 
 \item The \gls{ICSMC} algorithm can be combined with backward sampling to reduce the complexity to $\bo(1)$ \citep{lee2020coupled}. This is similar to combining the independent \gls{MH} updates with a `blocking' (in the time direction) strategy. Indeed, \citet{singh2017blocking} reduced the complexity of the standard \gls{ICSMC} algorithm by using time-direction blocking instead of backward sampling.
\end{itemize}

Now, consider the case that the `spatial' dimension $D$ is large. 
\begin{itemize}
 \item Unfortunately, the \gls{ICSMC} algorithm \citep{andrieu2010particle} (with or without backward sampling) then proposed propose `global' moves in the `space' direction and consequently suffer a curse of dimension in $D$ (i.e.\ complexity $\bo(\eul^D)$) in the same way as an independent \gls{MH} algorithm (with or without blocking in the time direction) both. Indeed, if $T = N = 1$ then the \gls{ICSMC} algorithm reduces to a standard independent \gls{MH} algorithm for which this drawback is well known.
 
 \item It is also well known that proposing `local' moves can overcome this curse of dimension \citep{roberts1997weak}. That is, in the classical \gls{MCMC} setting, we must use, e.g., (suitably scaled) random-walk rather than independence proposals. In this work, we have therefore proposed a novel iterated `local' \gls{CSMC} algorithm, termed \gls{IRWCSMC} algorithm, which utilises Gaussian random-walk proposals whose variance is of order $D^{-1}$. The algorithm provably avoids the curse of dimension in $D$ suffered by the original \gls{ICSMC} algorithm. It can be combined with backward sampling -- which is again akin to blocking in the time direction in classical \gls{MCMC} algorithms -- to potentially again remove the need for scaling the number of particles with $T$, leading to a $\bo(D)$ complexity. Our algorithm reduces to a standard Gaussian random-walk \gls{MH} algorithm if $T=N=1$. 
 \end{itemize}

\paragraph*{Limitations} The \gls{IRWCSMC} algorithm shares the well-known limitations of the random-walk \gls{MH} algorithm. That is, it requires the state space to be continuous in order to allow for suitably scaled local proposals; and such `local' proposals may not easily move between well-separated modes. 

\paragraph*{Extensions} Further reductions in the complexity from $\bo(D)$ to $\bo(1)$ may be feasible by employing blocking strategies in the `space' direction \citep{rebeschini2013local, finke2017approximate}. Indeed, \citet{murphy2015blocked} proposed a spatially-blocked \gls{ICSMC} algorithm. However, such strategies require a specific `spatial' (de)correlation structure which can only be found in particular models.

The analysis of the \gls{IRWCSMC} algorithm in high dimensions could be extended in a number of ways. First, we could easily consider models in which the potential functions $G_t$ depend not only on $x_t$ but also on $x_{t-1}$ or allow scale factors $\ell_t$ differ across particles. Second, we could consider the case that ancestor sampling instead of backward sampling is used. Third, in the same way as this has been done in the literature on optimal scaling for classical \gls{MCMC} algorithms, the assumptions on the high-dimensional regime could be relaxed, e.g.\ by allowing for some of the $D$ components of the target distribution to be differently scaled or by allowing for dependence between some of the $D$ components \citep[see, e.g.,][]{sherlock2009optimal, yang2020optimal}. Likewise, we could allow for non-Gaussian proposals \citep{neal2011optimal}. Finally, we could investigate the optimal choice of the scaling factors $\ell_t^\star$ and the associated optimal acceptance rates as well as proving a suitable $T$-dimensional diffusion limit of a time-scaled $T$-dimensional spatial marginal of the Markov chain induced by the algorithm.

We stress that the \gls{IRWCSMC} algorithm introduced in this work is by no means the only way of achieving `local' \gls{CSMC} updates. We have only focussed on this particular algorithm to keep the presentation simple. Alternative local \gls{CSMC} algorithms are possible. For instance, the first such local algorithm proposed in the literature is the method from \citet{shestopaloff2018sampling} which uses \gls{MCMC} kernels for proposing local moves around the reference path. Their algorithm reduces to a delayed-acceptance \gls{MH} algorithm \citep{christen2005markov} if $T=N=1$. The algorithm showed promising performance in some high-dimensional settings in \citep{finke2016embedded}. A general framework which admits this approach as well as the algorithms analysed in this work as special cases is provided in \citet{finke2016embedded}. We are currently exploring the idea of scattering particles locally around the reference path using (partially) deterministic mappings, e.g.\ as in Hamiltonian Monte Carlo methods, or via piecewise-deterministic \gls{MCMC} methods.


\begin{acks}[Acknowledgments]
The first author would like to thank Arnaud Doucet for insightful discussions which led to this research. 
\end{acks}

\begin{funding}
The authors acknowledge support from the Singapore Ministry of Education Tier~2 (MOE2016-T2-2-135) and a Young Investigator Award Grant (NUSYIA FY16 P16; R-155-000-180-133).
\end{funding}



\bibliographystyle{imsart-nameyear} 
\bibliography{csmc_in_high_dim.bbl}       

\begin{appendix}

\section{Details for Section~\ref{sec:csmc}}
\label{app:sec:csmc}

\subsection{Joint law induced by Algorithm~\ref{alg:iterated_csmc}}
\label{app:subsec:joint_law_csmc}

We now formally define the joint law of all random variables generated in Algorithm~\ref{alg:iterated_csmc}. This will be used in some of the proofs below. 

To simplify the presentation, we note that we fixed the reference path in Algorithm~\ref{alg:iterated_csmc} to always have particle index $0$, i.e.\ we always set $\smash{\Particle_t^0 \coloneqq \particle_t^0 \coloneqq \state_t}$ as well as $\smash{A_{t-1}^0 = a_{t-1}^0 \coloneqq 0}$. However, in some of the proofs below, it is more convenient to work with a slightly more general version of the algorithm which, at the beginning each time step, draws a particle index $J_t = j_t$ from a uniform distribution on $[N]_0$ and then sets $\smash{A_{t-1}^{j_t} = a_{t-1}^{j_t} \coloneqq j_{t-1}}$ as well as $\smash{\Particle_t^{j_t} \coloneqq \particle_t^{j_t} \coloneqq \state_t}$. 

Conditional on $\State_{1:T} = \state_{1:\nTimeSteps} = \state_{1:\nTimeSteps}[l]$, the joint law of all random variables $(J_{1:T}, \Particle_{1:T}, A_{1:T-1}, K_{1:T}, \State_{1:T}')$ generated by this slightly generalised version of Algorithm~\ref{alg:iterated_csmc} may be written as
\begin{align}
  \MoveEqLeft\IteratedCsmcKernelRandomWalk{T,D,\state_{1:\nTimeSteps}}{N,\star}(\diff j_{1:T} \times \diff \particle_{1:\nTimeSteps} \times \diff a_{1:\nTimeSteps-1} \times \diff k_T)\\
  & \coloneqq \dUnif_{[N]_0^T}(\diff j_{1:T}) \dDirac_{\state_{1:T}}(\diff \particle_1^{j_1} \times \dotsb \times \diff \particle_T^{j_T}) \Biggl[ \prodSubstackAligned{n}{=}{0}{n}{\neq}{j_1}^N \Mutation_1(\diff \particle_1^\particleIndex)\Biggr]\\[-2ex]
  & \quad \times \Biggl[ \prod_{\smash{\timeIndex = 2}}^{\smash{\nTimeSteps}} \dDirac_{j_{t-1}}(\diff a_{\timeIndex - 1}^{j_t}) \prodSubstackAligned{n}{=}{0}{n}{\neq}{j_t}^N
  \ResamplingKernel{\timeIndex-1,\nDimensions}{\nParticles}(\particle_{\timeIndex - 1}, \diff a_{\timeIndex - 1}^\particleIndex)
  \Mutation_\timeIndex(\particle_{\timeIndex - 1}^{a_{\timeIndex - 1}^\particleIndex}, \diff \particle_\timeIndex^\particleIndex)\Biggr]\\[-2ex]
  & \quad \times \ResamplingKernel{\nTimeSteps, \nDimensions}{\nParticles}(\particle_\nTimeSteps, \diff\outputParticleIndex_\nTimeSteps)\\
  & \quad \times 
  \begin{cases}
  \biggl[\,
   \smashoperator{\prod_{\timeIndex = 1}^{\smash{\nTimeSteps-1}}} \dDirac_{a_t^{k_{t+1}}}(\diff k_t)\biggr] & \text{[without backward sampling]}\\
   \biggl[\,\smashoperator{\prod_{\timeIndex = 1}^{\nTimeSteps-1}} \BackwardKernel{\timeIndex, \nDimensions}{\nParticles}((\particle_{\timeIndex},\particle_{\timeIndex+1}^{\outputParticleIndex_{\timeIndex + 1}}), \diff\outputParticleIndex_\timeIndex)\biggr] & \text{[with backward sampling]}
  \end{cases}\\
  & \quad \times \dDirac_{(\particle_1^{k_1}, \dotsc, \particle_T^{k_T})}(\diff \state_{1:T}'). \label{eq:csmc_generalised_joint_law}
\end{align}
Here, we have defined the following quantities.
\begin{itemize}
 \item \textbf{Resampling kernels.} For any $\particleIndex \in [\nParticles]_0$ and any $\timeIndex \in [\nTimeSteps]$,
\begin{align}
  \ResamplingKernel{\timeIndex,\nDimensions}{\nParticles}(\particle_{\timeIndex}, \{\particleIndex\})
  & \coloneqq 
  \selectionFunctionBoltzmann{\particleIndex}(\{\logWeight_{\timeIndex}(\particle_{\timeIndex}^\particleIndexAlt) - \logWeight_\timeIndex(\particle_{\timeIndex}^0)\}_{\particleIndexAlt = 1}^\nParticles)
  = 
  \dfrac{\Potential_t(\particle_t^n)}{\sum_{\particleIndexAlt = 0}^\nParticles \Potential_t(\particle_t^\particleIndexAlt)}.
\end{align}
When using the forced-move extension, replace $\smash{\ResamplingKernel{T,\nDimensions}{\nParticles}(\particle_T, \{k_T\})}$ at time $t=T$ in \eqref{eq:csmc_generalised_joint_law} by
\begin{align}
 \begin{dcases}
   \dfrac{\Potential_T(\particle_T^{k_T})}{\sum_{\particleIndexAlt = 0}^\nParticles \Potential_T(\particle_T^\particleIndexAlt) - \Potential_T(\particle_T^{k_T}) \wedge \Potential_T(\particle_T^{j_T})}, & \text{if $k_T \neq j_T$,}\\
   1 - \sumSubstackAligned{l}{=}{0}{l}{\neq}{j_T}^N \dfrac{\Potential_T(\particle_T^{l})}{\sum_{\particleIndexAlt = 0}^\nParticles \Potential_T(\particle_T^\particleIndexAlt) - \Potential_T(\particle_T^{l}) \wedge \Potential_T(\particle_T^{j_T})}, & \text{if $k_T = j_T$.}
 \end{dcases}
\end{align}

\item \textbf{Backward kernels.} For $t \in [T-1]$,
 \begin{align}
  \BackwardKernel{\timeIndex,\nDimensions}{\nParticles}((\particle_{\timeIndex},  \particle_{\timeIndex+1}^{\outputParticleIndex_{\timeIndex+1}}), \{\particleIndex\})
  & \coloneqq 
     \selectionFunctionBoltzmann{\particleIndex}(\{
    \logBackwardWeight_\timeIndex(\particle_{\timeIndex}^\particleIndexAlt, \particle_{\timeIndex+1}^{\outputParticleIndex_{\timeIndex+1}}) - \logBackwardWeight_\timeIndex(\particle_{\timeIndex}^0, \particle_{\timeIndex+1}^{\outputParticleIndex_{\timeIndex+1}})
  \}_{\particleIndexAlt = 1}^\nParticles)\\
 & =
  \frac{\Potential_\timeIndex(\particle_\timeIndex^{n}) \mutation_{\timeIndex+1}(\particle_\timeIndex^{n}, \particle_{\timeIndex+1}^{\outputParticleIndex_{\timeIndex+1}})}{\sum_{\particleIndexAlt=0}^\nParticles \Potential_\timeIndex(\particle_\timeIndex^\particleIndexAlt) \mutation_{\timeIndex+1}(\particle_\timeIndex^\particleIndexAlt, \particle_{\timeIndex+1}^{\outputParticleIndex_{\timeIndex+1}})}.
\end{align}
\end{itemize}

From this definition, we can recover the joint law of all random variables $(\Particle_{1:T}, A_{1:T-1}, K_{1:T}, \State_{1:T}')$ generated in Steps~\ref{alg:iterated_csmc:1}--\ref{alg:iterated_csmc:3} of Algorithm~\ref{alg:iterated_csmc} by conditioning on the event $\{J_1 = 0, \dotsc, J_T = 0\}$, i.e.\
\begin{align}
 \IteratedCsmcKernel{T,D,\state_{1:\nTimeSteps}}{N} \coloneqq \IteratedCsmcKernel{T,D,\state_{1:\nTimeSteps}}{N,\star}(\ccdot|J_1 = 0, \dotsc, J_T = 0).
\end{align}
Let $\ExpectationCsmcKernel{T,D,\state_{1:\nTimeSteps}}{N,\star}$ denote expectation w.r.t.\ $\IteratedCsmcKernel{T,D,\state_{1:\nTimeSteps}}{N,\star}$. In the remainder of this section, we will sometimes work with $\smash{\IteratedCsmcKernel{T,D,\state_{1:\nTimeSteps}}{N,\star}}$ rather than $\smash{\IteratedCsmcKernel{T,D,\state_{1:\nTimeSteps}}{N,\star}}$. This is justified because both versions of the \gls{ICSMC} algorithm induce the same Markov kernel:
\begin{align}
 \ExpectationCsmcKernel{\nTimeSteps,\nDimensions,\state_{1:\nTimeSteps}}{\nParticles,\star}[\ind\{\State_{1:T}' \in \diff \state_{1:T}'\}] 
 & = \ExpectationCsmcKernel{\nTimeSteps,\nDimensions,\state_{1:\nTimeSteps}}{\nParticles}[\ind\{\State_{1:T}' \in \diff \state_{1:T}'\}]\\
 & = \InducedIteratedCsmcKernel{\nTimeSteps,\nDimensions}{\nParticles}(\state_{1:\nTimeSteps}, \diff \state_{1:T}'). \label{eq:same_induced_markov_kernel_csmc}
\end{align} 

\subsection{Proof of Proposition~\ref{prop:csmc_invariance}}
\label{app:subsec:csmc_invariance}

\begin{namedproof}[of Proposition~\ref{prop:csmc_invariance}]
 For the plain algorithm (with neither the backward sampling nor the forced-move extension) we can readily check that
 \begin{align}
  \MoveEqLeft \Target_{T,D}(\diff \state_{1:T}) \IteratedCsmcKernel{\nTimeSteps,\nDimensions,\state_{1:T}}{\nParticles, \star}(\diff j_{1:T} \times \diff \particle_{1:\nTimeSteps} \times \diff a_{1:\nTimeSteps-1} \times \diff\outputParticleIndex_{1:\nTimeSteps} \times \diff \state_{1:T}')\\
  & = \Target_{T,D}(\diff \state_{1:T}') \IteratedCsmcKernel{\nTimeSteps,\nDimensions,\state_{1:\nTimeSteps}'}{\nParticles,\star}(\diff k_{1:T} \times \diff \particle_{1:\nTimeSteps} \times \diff a_{1:\nTimeSteps-1} \times \diff j_{1:\nTimeSteps} \times \diff \state_{1:T}), \label{eq:prop:csmc_invariance:1}
 \end{align}
 i.e.\ \eqref{eq:prop:csmc_invariance:1} admits $\Target_{T,D}(\diff \state_{1:T}')$ as a marginal. 
 
 For the backward-sampling extension, let $\smash{\mathring{\mathbb{P}}_{T,D,\state_{1:T}}^{N,\star}}$ be the same as $\smash{\IteratedCsmcKernel{\nTimeSteps,\nDimensions,\state_{1:T}}{\nParticles,\star}}$ (without backward sampling) except that the terms $\smash{ \dDirac_{j_{t-1}}(\diff a_{\timeIndex - 1}^{j_t})}$ in \eqref{eq:csmc_generalised_joint_law} are replaced by $\smash{\BackwardKernel{t-1,\nDimensions}{\nParticles}((\particle_{t-1},  \particle_{t}^{j_t}), \{a_{t-1}^{j_t}\})}$. Then 
  \begin{align}
  \MoveEqLeft \Target_{T,D}(\diff \state_{1:T}) \IteratedCsmcKernel{\nTimeSteps,\nDimensions,\state_{1:T}}{\nParticles,\star}(\diff j_{1:T} \times \diff \particle_{1:\nTimeSteps} \times \diff a_{1:\nTimeSteps-1} \times \diff\outputParticleIndex_{1:\nTimeSteps} \times \diff \state_{1:T}')\\
  & = \Target_{T,D}(\diff \state_{1:T}') \mathring{\mathbb{P}}_{T,D,\state_{1:\nTimeSteps}'}^{N,\star}(\diff k_{1:T} \times \diff \particle_{1:\nTimeSteps} \times \diff a_{1:\nTimeSteps-1} \times \diff j_{1:\nTimeSteps} \times \diff \state_{1:T}), \label{eq:prop:csmc_invariance:2}
 \end{align}
  That is, \eqref{eq:prop:csmc_invariance:2} again admits $\Target_{T,D}(\diff \state_{1:T}')$ as a marginal. Incidentally, $\smash{\mathring{\mathbb{P}}_{T,D,\state_{1:T}}^{N,\star}}$ can be recognised as the law of all the random variables generated by the \gls{ICSMC} algorithm with ancestor sampling from \citet{lindsten2012ancestor}. 
  
  Finally, the algorithm with the forced-move extension can be justified as a partially collapsed Gibbs sampler because this extension leaves the marginal distribution of $K_T$ (under \eqref{eq:prop:csmc_invariance:1} without backward sampling or under \eqref{eq:prop:csmc_invariance:2} with backward sampling) conditional on $(\State_{1:T}, J_{1:T}, \Particle_{1:T}, A_{1:T})$ invariant. \qedwhite
\end{namedproof}

\subsection{Verification of Assumptions~\ref{as:degeneracy_of_resampling_kernels} and \ref{as:degeneracy_of_backward_kernels} in a linear-Gaussian state-space model}
\label{app:subsec:verification_of_csmc_assumptions_for_breakdown}

 Consider a state-space model with $D$-dimensional observations $\mathbf{y}_t = y_{t,1:D} \in \reals^D$ and 
 \begin{align}
   \mathbf{H}_t(\state_t, \diff \mathbf{y}_t) & \textstyle\coloneqq \prod_{d=1}^D \dN(\diff y_{t,d}; x_{t,d}, 1),\\
   \Mutation_t(\state_{t-1}, \diff \state_t) &\textstyle \coloneqq \prod_{d=1}^D \dN(\diff x_{t,d}; x_{t-1,d}, 1).
 \end{align}
 Let $\standardNormalPdf$ be a density function of a univariate standard normal distribution. Then this model satisfies \ref{as:iid_model} with
 \begin{align}
  G_t(x_{t,d}) & = \standardNormalPdf(y_{t,d} - x_{t,d}),\\
  m_{t}(x_{t-1,d}, x_{t,d}) & = \standardNormalPdf(x_{t,d} - x_{t-1,d}).
 \end{align}
 To simplify the notation and calculations, we hereafter drop the subscript $d$, take $y_1 = \dotsc, y_T = 0$ and assume as initial density $\smash{m_1(x_{1})  =\smash{\standardNormalPdf(x_{1} / \sqrt{\sigma^2 + 1}) / \sqrt{\sigma^2 + 1}}}$, where $\smash{\sigma^2 \coloneqq (\sqrt{5}-1)/2}$. Standard Kalman-filtering and Kalman-smoothing recursions then give $\smash{\Prob(X_t \in \diff x_t| Y_{1:t} = y_{1:t}) = \dN(\diff x_t; 0, \sigma^2)}$ as well as $\smash{\pi_T(\diff x_{1:T})} = \smash{\Prob(X_{1:T} \in \diff x_{1:T}|Y_{1:T} = y_{1:T})} = 
 \smash{\dN(\diff x_{1:T}; \zeroMat_T, C)}$ with $\smash{[C]_{s,t} = u^{\lvert t - s\rvert} \sigma_{s \vee t}^2}$, where $u \coloneqq \sigma^2 / (\sigma^2 + 1)$ and
 \begin{align}
  \sigma_t^2 \coloneqq 
   u \frac{1 - (u^2)^{T-t}}{1 - u^2} \ind\{t < T\} + (u^2)^{T-t} \sigma^2.
  \end{align}
  Tedious but simple algebra then shows that \ref{as:degeneracy_of_resampling_kernels} and \ref{as:degeneracy_of_backward_kernels} hold because
  \[
  \underline{b}_T > \underline{r}_T = r_{T|T}  =
  \begin{rcases}
  \begin{dcases}
   \tfrac{1}{2}(\log(\sigma^2 + 2) - \sigma^2),
   & \text{if $T = 1$,}\\
   \tfrac{1}{2}(\log(2) + [\sigma^2(u^2-2) + u]/2),
   & \text{if $T > 1$,}
  \end{dcases}
  \end{rcases}
  > 0.15. 
  \]

\subsection{Proof of Proposition~\ref{prop:limiting_csmc_algorithm}}
\label{app:subsec:proof_of_prop:limiting_csmc_algorithm}

\begin{namedproof}[of Proposition~\ref{prop:limiting_csmc_algorithm}]
A telescoping-sum argument gives the following decomposition which will form the basis of the proof both in the case with and without resampling. Here, we let $\zeroMat$ denote a vector of zeros of appropriate length.
 \begin{align}
  d_{T,D,\state_{1:T}}^N & \leq \sum_{t=1}^{T-1} \sum_{n = 1}^N \bigl\lVert \ExpectationCsmcKernel{\nTimeSteps,\nDimensions,\state_{1:\nTimeSteps}}{\nParticles}[\ind\{A_t^n \in \ccdot\}| A_{1:t-1} = \zeroMat] - \dDirac_0(\ccdot)\bigr\rVert\\
  & \quad + \bigl\lVert \ExpectationCsmcKernel{\nTimeSteps,\nDimensions,\state_{1:\nTimeSteps}}{\nParticles}[\ind\{K_T \in \ccdot\}|A_{1:T-1} = \zeroMat] - \dDirac_0(\ccdot)\bigr\rVert\\
  & \quad + \sum_{t=1}^{T-1} \bigl\lVert \ExpectationCsmcKernel{\nTimeSteps,\nDimensions,\state_{1:\nTimeSteps}}{\nParticles}[\ind\{K_t \in \ccdot\}|(A_{1:T-1}, K_{t+1:T}) = \zeroMat] - \dDirac_0(\ccdot)\bigr\rVert. \label{eq:telescoping-sum_decomposition_csmc}
\end{align}
 
First, we consider the case \emph{without} backward sampling. We write
\begin{align}
 \calR_{t|T,D}(\state_{t-1:t}) 
 & \coloneqq  \textstyle [\frac{1}{D} \sum_{d=1}^D \log M_t(G_t)(x_{t-1,d}) -  \log G_t(x_{t,d})] + r_{t|T}\\
 & = \textstyle \frac{1}{D} \sum_{d=1}^D \log M_t(G_t)(x_{t-1,d}) - \E[\log M_t(G_t)(X_{t-1})]\\
 & \quad +  \textstyle \frac{1}{D} \sum_{d=1}^D \E[\log G_t(X_t)] -  \log G_t(x_{t,d}).
\end{align}
For some $\eta \in (0,1/2)$, we set
\begin{align}
 \ConcentrationSet_{\nTimeSteps, \nDimensions} 
 \coloneqq \{\state_{1:T} \in \spaceState_{T,D} \mid \forall \, t \in [T]: \lvert \calR_{t|T,D}(\state_{t-1:t}) \rvert < D^{-\eta}\}.
\end{align}
 Since $\eta < 1/2$ implies that $D^\eta \sqrt{2 D \log \log D} = \bo(D)$, the law of the iterated logarithm then implies that $\lim_{\nDimensions \to \infty} \Target_{\nTimeSteps, \nDimensions}(\ConcentrationSet_{\nTimeSteps, \nDimensions}) = 1$.
 
  We can now turn to the terms in the decomposition \eqref{eq:telescoping-sum_decomposition_csmc}. Let $\smash{(\state_{1:T,D})_{D \geq 1}}$ be some sequence in $\smash{(\spaceState_{T,D})_{D \geq 1}}$, i.e.\ $\smash{\state_{t,D} = x_{t,1:D,D} \in \reals^{D}}$, for any $D \geq 1$. For any $n \in [N]$ and $t \in [T]$, let $\smash{\Particle_{\timeIndex,D}^n \sim \Mutation_t(\state_{t-1,D}, \ccdot)}$ and $\smash{\Particle_{\timeIndex,D}^0 \coloneqq \state_{t,D}}$. For any $t \in [T]$ (with $A_t^n$ replaced by $K_T$ if $t = T$), we have that 
\begin{align}
 \!\!\!\!\!\!\MoveEqLeft \lvert \ExpectationCsmcKernel{\nTimeSteps,\nDimensions,\state_{1:T,D}}{\nParticles}[\ind\{A_t^n \in [N]\}| A_{1:t-1} = \zeroMat] - \dDirac_{0}([N])\rvert \\
 & = \ExpectationCsmcKernel{\nTimeSteps,\nDimensions,\state_{1:T,D}}{\nParticles}\biggl[ \frac{\sum_{n=1}^N \exp\{\logWeight_{\timeIndex}(\Particle_{\timeIndex,D}^n) - \logWeight_\timeIndex(\Particle_{\timeIndex,D}^0)\}}{1 + \sum_{m=1}^N \exp\{\logWeight_{\timeIndex}(\Particle_{\timeIndex,D}^m) - \logWeight_\timeIndex(\Particle_{\timeIndex,D}^0)\}}  \,\bigg|\, A_{1:t-1} = \zeroMat\biggr].\!\!\!\!\!\! \label{eq:resampling_kernel_expectatation_csmc}
\end{align}
For any $n \in [N]$ and $t \in [T]$, let $\smash{\Particle_{\timeIndex,D}^n \sim \Mutation_t(\state_{t-1,D}, \ccdot)}$ and $\smash{\Particle_{\timeIndex,D}^0 \coloneqq \state_{t,D}}$.  To complete the proof, it then suffices to show that 
\begin{align}
 \textstyle \sum_{n=1}^N \exp\{\logWeight_{\timeIndex}(\Particle_{\timeIndex,D}^n) - \logWeight_\timeIndex(\Particle_{\timeIndex,D}^0)\} \convergesInProbability 0,
\end{align}
as $D \to \infty$, whenever $\log N = \lo(D)$. By Markov's inequality, for any $\varepsilon > 0$,
\begin{align}
  \MoveEqLeft \textstyle \mathbb{P}(\{\sum_{n=1}^N \exp\{\logWeight_{\timeIndex}(\Particle_{\timeIndex,D}^n) - \logWeight_\timeIndex(\Particle_{\timeIndex,D}^0)\} > \varepsilon \})\\
  & \leq \textstyle N \varepsilon^{-1} \exp\{- D r_{t|T} + D\lvert \calR_{t|T,D}(\state_{t-1:t,D}) \rvert\}\\
  & \leq \textstyle N \varepsilon^{-1} \exp\{-D r_{t|T} + D^{1-\eta}\},
\end{align}
and $\smash{r_{t|T} \geq \underline{r}_T > 0}$ by \ref{as:degeneracy_of_resampling_kernels}. This completes the proof in the case without backward sampling.

We now consider the case \emph{with} backward sampling. We write
\begin{align}
 \MoveEqLeft \calB_{t|T,D}(\state_{t-1:t+1})\\
 & \coloneqq  \textstyle [\frac{1}{D} \sum_{d=1}^D \log M_t(G_t m_{t+1}(\ccdot, x_{t+1,d}))(x_{t-1,d}) - \logBackwardWeightSingle_t(x_{t:t+1,d})] + b_{t|T}\\
  & = \textstyle \frac{1}{D} \sum_{d=1}^D \log M_t(G_t m_{t+1}(\ccdot,x_{t+1,d}))(x_{t-1,d}) - \E[\log M_t(G_t m_{t+1}(\ccdot, X_{t+1}))(X_{t-1})]\!\!\!\!\!\!\!\!\!\!\!\\
 & \quad +  \textstyle \frac{1}{D} \sum_{d=1}^D \E[\log \{ G_t(X_t)m_{t+1}(X_t, X_{t+1})\}] -  \log\{G_t(x_{t,d} m_{t+1}(x_{t,d}, x_{t+1,d})\}.
\end{align}
with convention $\calB_{T|T,D} = \calR_{T|T,D}$. For some $\eta \in (0,1/2)$, we set
\begin{align}
 \!\!\!\!\ConcentrationSet_{\nTimeSteps, \nDimensions} 
 & \coloneqq \{\state_{1:T} \in \spaceState_{T,D} \mid \forall \, t \in [T]: \lvert \calR_{t|T,D}(\state_{t-1:t})\rvert  \vee \lvert \calB_{t|T,D}(\state_{t-1:t+1}) \rvert < D^{-\eta} \}.\!\!\!\!
\end{align}
 Since $\eta < 1/2$ implies that $\smash{D^\eta \sqrt{2 D \log \log D} = \bo(D)}$, the law of the iterated logarithm then again implies that $\lim_{\nDimensions \to \infty} \Target_{\nTimeSteps, \nDimensions}(\ConcentrationSet_{\nTimeSteps, \nDimensions}) = 1$. 
 
 Again we consider the decomposition \eqref{eq:telescoping-sum_decomposition_csmc}. We then have that 
\begin{align}
\!\!\!\!\!\!\!\!\! \MoveEqLeft \lvert \ExpectationCsmcKernel{\nTimeSteps,\nDimensions,\state_{1:T,D}}{\nParticles}[\ind\{K_t \in [N]\}| (A_{1:T-1}, K_{t+1}) = \zeroMat] - \dDirac_0([N])\rvert\\
 & \!\!\!\!\!\!\!\!\!\!\!\! = \ExpectationCsmcKernel{\nTimeSteps,\nDimensions,\state_{1:T,D}}{\nParticles}\biggl[ \frac{\sum_{n=1}^N \exp\{\logBackwardWeight_\timeIndex(\Particle_{\timeIndex,D}^n, \Particle_{\timeIndex+1,D}^0) - \logBackwardWeight_\timeIndex(\Particle_{\timeIndex,D}^0, \Particle_{\timeIndex+1,D}^0)\}}{1 + \sum_{m=1}^N \exp\{\logBackwardWeight_\timeIndex(\Particle_{\timeIndex,D}^\particleIndexAlt, \Particle_{\timeIndex+1,D}^0) - \logBackwardWeight_\timeIndex(\Particle_{\timeIndex,D}^0, \Particle_{\timeIndex+1,D}^0)\}}  \,\bigg|\, A_{1:t-1} = \zeroMat\biggr].\!\!\!\!\!\!\!\!\! \label{eq:backward_kernel_expectatation_csmc}
\end{align}
To complete the proof, it suffices to show that
\begin{align}
 \textstyle \sum_{n=1}^N \logBackwardWeight_\timeIndex(\Particle_{\timeIndex,D}^n, \Particle_{\timeIndex+1,D}^0) - \logBackwardWeight_\timeIndex(\Particle_{\timeIndex,D}^0, \Particle_{\timeIndex+1,D}^0) \convergesInProbability 0,
\end{align}
as $D \to \infty$,  whenever $\log N = \lo(D)$. By Markov's inequality, for any $\varepsilon > 0$, 
\begin{align}
  \MoveEqLeft \textstyle \mathbb{P}(\{\sum_{n=1}^N \exp\{\logBackwardWeight_t(\Particle_{\timeIndex,D}^n, \Particle_{\timeIndex+1,D}^0) - \logBackwardWeight_t(\Particle_{\timeIndex,D}^0, \Particle_{\timeIndex+1,D}^0) > \varepsilon \})\\
  & \leq N \varepsilon^{-1} \exp\{- D b_{t|T} + D\lvert \calB_{t|T,D}(\state_{t-1:t+1,D}) \rvert\}\\
  & \leq N \varepsilon^{-1} \exp\{-D b_{t|T} + D^{1-\eta}\},
\end{align}
 and $b_{t|T} \geq \underline{b}_T > 0$ by \ref{as:degeneracy_of_backward_kernels}.  This completes the proof in the case with backward sampling. \qedwhite
\end{namedproof}

\section{Details for Section~\ref{sec:rwehmm}}
\label{app:sec:rwehmm}


\subsection{Joint law induced by Algorithm~\ref{alg:iterated_rwehmm}}
\label{app:subsec:joint_law_rwehmm}

We now formally define the joint law of all random variables generated in Algorithm~\ref{alg:iterated_rwehmm}. This will be used in some of the proofs below. 

To simplify the presentation, we note that we again fixed the reference path in Algorithm~\ref{alg:iterated_rwehmm} to always have particle index $0$, i.e.\ we always set $\smash{\Particle_t^0 \coloneqq \particle_t^0 \coloneqq \state_t}$. However, in some of the proofs below, it is more convenient to work with a slightly more general version of the algorithm which, at each time step, draws a particle index $J_t = j_t$ from a uniform distribution on $[N]_0$ and then sets $\smash{\Particle_t^{j_t} \coloneqq \particle_t^{j_t} \coloneqq \state_t}$. 

Conditional on $\State_{1:T} = \state_{1:\nTimeSteps} = \state_{1:\nTimeSteps}[l]$, the joint law of all random variables $(J_{1:T}, \Particle_{1:T}, K_{1:T}, \State_{1:T}')$ generated by this slightly generalised version of Algorithm~\ref{alg:iterated_rwehmm} may be written as
\begin{align}
  \MoveEqLeft\IteratedEhmmKernelRandomWalk{T,D,\state_{1:\nTimeSteps}}{N,\star}(\diff j_{1:T} \times \diff \particle_{1:\nTimeSteps} \times \diff k_T)\\[-1ex]
  & \coloneqq \dUnif_{[N]_0^T}(\diff j_{1:T}) \dDirac_{\state_{1:T}}(\diff \particle_1^{j_1} \times \dotsb \times \diff \particle_T^{j_T}) \biggl[ \prod_{t=1}^T  \SamplingKernelRandomWalk{\timeIndex,\nDimensions}{\nParticles}(\particle_\timeIndex^{j_t}, \diff \particle_\timeIndex^{-j_t}) \biggr]\\[-1.5ex]
  & \quad \times \xi_T(\particle_{1:T}, \diff k_{1:T}) \dDirac_{(\particle_1^{k_1}, \dotsc, \particle_T^{k_T})}(\diff \state_{1:T}'). \label{eq:rwehmm_generalised_joint_law}
\end{align}
Here, we have defined $\smash{\particle_t^{-n} \coloneqq (\particle_t^0, \dotsc, \particle_t^{n-1}, \particle_t^{n+1}, \dotsc, \particle_t^N)}$ as well as the following quantities.
\begin{itemize}
 \item \textbf{Proposal kernels.} For any $t \in [T]$ and any $n \in [N]_0$,
\begin{align}
 \SamplingKernelRandomWalk{\timeIndex,\nDimensions}{\nParticles}(\particle_\timeIndex^n, \diff \particle_\timeIndex^{-n}) \coloneqq \prod_{\dimensionIndex = 1}^\nDimensions \dN\bigl(\diff \particleSingle_{\timeIndex,\dimensionIndex}^{-n};\particleSingle_{\timeIndex, \dimensionIndex}^n \unitMat_{\nParticles}, \tfrac{\ell_\timeIndex}{\nDimensions}\varSigma\bigr),
\end{align}
 where $z_{t,d}^{-n} \coloneqq (z_{t,d}^0, \dotsc, z_{t,d}^{n-1}, z_{t,d}^{n+1}, \dotsc, z_{t,d}^N)$.

 \item \textbf{Selection probability.} 
\begin{align}
  \xi_T(\particle_{1:T}, \{k_{1:T}\}) \coloneqq \dfrac{\Target_{T,D}(\particle_1^{k_1}, \dotsc, \particle_T^{k_T})}{\sum_{l_{1:T} \in [N]_0^T} \Target_{T,D}(\particle_1^{l_1}, \dotsc, \particle_T^{l_T})}.
\end{align}
\end{itemize}

From this definition, we can recover the joint law of all random variables $(\Particle_{1:T}, K_{1:T}, \State_{1:T}')$ generated in Steps~\ref{alg:iterated_rwehmm:1}--\ref{alg:iterated_rwehmm:3} of Algorithm~\ref{alg:iterated_rwehmm} by conditioning on the event $\{J_1 = 0, \dotsc, J_T = 0\}$, i.e.\
\begin{align}
 \IteratedEhmmKernelRandomWalk{T,D,\state_{1:\nTimeSteps}}{N} \coloneqq \IteratedEhmmKernelRandomWalk{T,D,\state_{1:\nTimeSteps}}{N,\star}(\ccdot|J_1 = 0, \dotsc, J_T = 0).
\end{align}
Let $\ExpectationEhmmKernelRandomWalk{T,D,\state_{1:\nTimeSteps}}{N,\star}$ denote expectation w.r.t.\ $\IteratedEhmmKernelRandomWalk{T,D,\state_{1:\nTimeSteps}}{N,\star}$. In the remainder of this section, we will sometimes work with $\smash{\IteratedEhmmKernelRandomWalk{T,D,\state_{1:\nTimeSteps}}{N,\star}}$ rather than $\smash{\IteratedEhmmKernelRandomWalk{T,D,\state_{1:\nTimeSteps}}{N,\star}}$. This is justified because both versions of the \gls{RWEHMM} algorithm induce the same Markov kernel:
\begin{align}
 \ExpectationEhmmKernelRandomWalk{\nTimeSteps,\nDimensions,\state_{1:\nTimeSteps}}{\nParticles,\star}[\ind\{\State_{1:T}' \in \diff \state_{1:T}'\}] 
 & = \ExpectationEhmmKernelRandomWalk{\nTimeSteps,\nDimensions,\state_{1:\nTimeSteps}}{\nParticles}[\ind\{\State_{1:T}' \in \diff \state_{1:T}'\}]\\
 & = \InducedIteratedEhmmKernelRandomWalk{\nTimeSteps,\nDimensions}{\nParticles}(\state_{1:\nTimeSteps}, \diff \state_{1:T}'). \label{eq:same_induced_markov_kernel_rwehmm}
\end{align}

\subsection{Proof of Proposition~\ref{prop:rwehmm_invariance}}
\label{app:subsec:rwehmm_invariance}

\begin{namedproof}[of Proposition~\ref{prop:rwehmm_invariance}]

 Recall that by \eqref{eq:rwcsmc_symmetric_proposal:0}, the random-walk type proposal used to scatter the particles around the reference path is symmetric in the sense that
 \begin{align}
  \lambda(\diff \particle_t^j) \SamplingKernelRandomWalk{\timeIndex,\nDimensions}{\nParticles}(\particle_\timeIndex^j, \diff \particle_\timeIndex^{-j}) = \lambda(\diff \particle_t^k) \SamplingKernelRandomWalk{\timeIndex,\nDimensions}{\nParticles}(\particle_\timeIndex^k, \diff \particle_\timeIndex^{-k}),
 \end{align}
 for any $j, k \in [N]_0$, where $\smash{\particle_t^{-n} \coloneqq (\particle_t^0, \dotsc, \particle_t^{n-1}, \particle_t^{n+1}, \dotsc, \particle_t^N)}$ and where $\lambda$ denotes a suitable version of the Lebesgue measure.

 We can then readily check that
 \begin{align}
  \MoveEqLeft \Target_{T,D}(\diff \state_{1:T}) \IteratedEhmmKernelRandomWalk{\nTimeSteps,\nDimensions,\state_{1:T}}{\nParticles, \star}(\diff j_{1:T} \times \diff \particle_{1:\nTimeSteps} \times \diff\outputParticleIndex_{1:\nTimeSteps} \times \diff \state_{1:T}')\\
  & = \Target_{T,D}(\diff \state_{1:T}') \IteratedEhmmKernelRandomWalk{\nTimeSteps,\nDimensions,\state_{1:\nTimeSteps}'}{\nParticles,\star}(\diff k_{1:T} \times \diff \particle_{1:\nTimeSteps} \times \diff j_{1:\nTimeSteps} \times \diff \state_{1:T}), \label{eq:prop:rwehmm_invariance:1}
 \end{align}
 i.e.\ \eqref{eq:prop:rwehmm_invariance:1} admits $\Target_{T,D}(\diff \state_{1:T}')$ as a marginal. \qedwhite
\end{namedproof}

\subsection{Proof of Proposition~\ref{prop:lower_bound_of_acceptance_rates_ehmm}}
\label{app:subsec:prop:lower_bound_of_acceptance_rates_ehmm}
  \begin{namedproof}[of Proposition~\ref{prop:lower_bound_of_acceptance_rates_ehmm}]
   The first part (the convergence statement) follows immediately from Proposition~\ref{prop:limiting_rwehmm_algorithm}. For the second part (the lower bound), we note that
   \begin{align}
    \acceptanceRateRwCsmc{\nTimeSteps}{\nParticles}(\timeIndex)
    = \sum_{n \in [N]} \ExpectationEhmmKernelRandomWalk{\nTimeSteps}{\nParticles}[\selectionFunctionBoltzmann{n}(\{V_t^l\}_{l=1}^N)]
    \geq \biggl(1 + \frac{\exp(\ell_t \calI_{t|T})}{N}\biggr)^{\mathrlap{-1}},
   \end{align}
   where the penultimate inequality follows by Lemma~\ref{lem:lower_bound_on_expectation_of_selection_function} in Appendix~\ref{app:subsec:prop:stability_of_acceptance_rates_rwcsmc}. \qedwhite
  \end{namedproof}

\subsection{Proof of Proposition~\ref{prop:esjd_rwehmm}}
\label{app:subsec:prop:esjd_rwehmm}

\begin{namedproof}[of Proposition~\ref{prop:esjd_rwehmm}]
 Let $\ConcentrationSet_{\nTimeSteps, \nDimensions} \in \sigFieldState_{T,D}$ be as in Proposition~\ref{prop:limiting_rwehmm_algorithm}. We then have
 \begin{align}
   \MoveEqLeft \bigl\lvert \esjdRandomWalkEhmm_{T,D}^N(t)  - \ell_t \acceptanceRateRwEhmm{\nTimeSteps}{\nParticles}(t)\bigr\rvert\\
   & = \textstyle \bigl\lvert \int_{\spaceState_{T,D}}  \Target_{T,D}(\diff \state_{1:T}) \ExpectationEhmmKernelRandomWalk{T,D,\state_{1:T}}{N}\bigl[\lVert \State_t' - \state_t \rVert_2^2\bigr] -  \ell_t \acceptanceRateRwEhmm{\nTimeSteps}{\nParticles}(t)\bigr\rvert\\
  & \quad \leq \textstyle \ell_t\sup_{\state_{1:T} \in \ConcentrationSet_{T,D}} \bigl\lvert  \acceptanceRateRwEhmm{\nTimeSteps,\nDimensions,\state_{1:T}}{\nParticles}(t) - \acceptanceRateRwEhmm{\nTimeSteps}{\nParticles}(t)\bigr\rvert\\
  & \qquad + \textstyle \sup_{\state_{1:T} \in \spaceState_{T,D}} \bigl\lvert \ExpectationEhmmKernelRandomWalk{T,D,\state_{1:T}}{N}\bigl[\lVert \State_t' - \state_t \rVert_2^2\bigr] -  \ell_t  \acceptanceRateRwEhmm{\nTimeSteps,\nDimensions,\state_{1:T}}{\nParticles}(t)\bigr\rvert\\
  & \qquad + \textstyle \Target_{T,D}(\spaceState_{T,D} \setminus \ConcentrationSet_{T,D})\sup_{\state_{1:T} \in \spaceState_{T,D}} \bigl\lvert \ExpectationEhmmKernelRandomWalk{T,D,\state_{1:T}}{N}\bigl[\lVert \State_t' - \state_t \rVert_2^2\bigr]\bigr\rvert. \label{eq:esjd_proof:1}
 \end{align}
 We now consider the limit of each of the terms in the last line of \eqref{eq:esjd_proof:1} as $D \to \infty$. The first term converges to zero by Proposition~\ref{prop:lower_bound_of_acceptance_rates_ehmm}. Take $\smash{U_{t,d}^n = D^{1/2} \ell_t^{-1/2} (Z_{t,d}^n - x_{t,d})}$ as in Algorithm~\ref{alg:iterated_rwehmm}, then $\smash{\widebar{U}_{t,D}^n \coloneqq D^{-1} \sum_{d=1}^D (U_{t,d}^n)^2 \convergesInProbability 1}$, as $D \to \infty$, and hence $\smash{\lvert \widebar{U}_{t,D}^n - 1\rvert \convergesInProbability 0}$ and $\smash{\lvert \widebar{U}_{t,D}^n \rvert \convergesInProbability 1}$ by the continuous mapping theorem. Furthermore, $\smash{(\lvert \widebar{U}_{t,D}^n - 1\rvert)_{D \geq 1}}$ and $\smash{(\lvert \widebar{U}_{t,D}^n\rvert)_{D \geq 1}}$ are uniformly integrable. Thus, for the second term, for any $\state_{1:T} \in \spaceState_{T,D}$, 
 \begin{align}
  \MoveEqLeft \bigl\lvert \ExpectationEhmmKernelRandomWalk{T,D,\state_{1:T}}{N}\bigl[\lVert \State_t' - \state_t \rVert_2^2\bigr] - \ell_t  \acceptanceRateRwEhmm{\nTimeSteps,\nDimensions,\state_{1:T}}{\nParticles}(t) \bigr\rvert\\
  & \leq \textstyle \bigl\lvert \ell_t \sum_{n=1}^N \ExpectationEhmmKernelRandomWalk{T,D,\state_{1:T}}{N}\bigl[\bigl(\widebar{U}_{t,D}^n - 1\bigr)\ind\{K_t = n\}\bigr] \bigr\rvert\\
  & \leq \ell_t N \E\bigl[\bigl\lvert \widebar{U}_{t,D}^1 - 1\bigr\rvert \bigr] \to 0.
 \end{align}
 Similarly, for any $\state_{1:T} \in \spaceState_{T,D}$, 
  \begin{align}
   \bigl\lvert \ExpectationEhmmKernelRandomWalk{T,D,\state_{1:T}}{N}\bigl[\lVert \State_t' - \state_t \rVert_2^2\bigr]\bigr\rvert
  & \leq \textstyle \bigl\lvert \ell_t \sum_{n=1}^N \ExpectationEhmmKernelRandomWalk{T,D,\state_{1:T}}{N}\bigl[\widebar{U}_{t,D}^n \ind\{K_t = n\}\bigr] \bigr\rvert\\
  & \leq \ell_t N \E\bigl[\bigl\lvert \widebar{U}_{t,D}^1\bigr\rvert \bigr] \to 1,
 \end{align}
 and the third term therefore converges to zero since $\pi(\spaceState_{t,D} \setminus \ConcentrationSet_{T,D}) \to 0$. \qedwhite
\end{namedproof}

\section{Details for Section~\ref{sec:rwcsmc}}
\label{app:sec:rwcsmc}

\subsection{Joint law induced by Algorithm~\ref{alg:iterated_rwcsmc}}
\label{app:subsec:joint_law_rwcsmc}

We now formally define the joint law of all random variables generated in Algorithm~\ref{alg:iterated_rwcsmc}. This will be used in some of the proofs below. 

To simplify the presentation, we note that we again fixed the reference path in Algorithm~\ref{alg:iterated_rwcsmc} to always have particle index $0$, i.e.\ we always set $\smash{\Particle_t^0 \coloneqq \particle_t^0 \coloneqq \state_t}$ as well as $\smash{A_{t-1}^0 = a_{t-1}^0 \coloneqq 0}$. However, in some of the proofs below, it is more convenient to work with a slightly more general version of the algorithm which, at the beginning each time step, draws a particle index $J_t = j_t$ from a uniform distribution on $[N]_0$ and then sets $\smash{A_{t-1}^{j_t} = a_{t-1}^{j_t} \coloneqq j_{t-1}}$ as well as $\smash{\Particle_t^{j_t} \coloneqq \particle_t^{j_t} \coloneqq \state_t}$. 

Conditional on $\State_{1:T} = \state_{1:\nTimeSteps} = \state_{1:\nTimeSteps}[l]$, the joint law of all random variables $(J_{1:T}, \Particle_{1:T}, A_{1:T-1}, K_{1:T}, \State_{1:T}')$ generated by this slightly generalised version of Algorithm~\ref{alg:iterated_rwcsmc} may be written as
\begin{align}
  \MoveEqLeft\IteratedCsmcKernelRandomWalk{T,D,\state_{1:\nTimeSteps}}{N,\star}(\diff j_{1:T} \times \diff \particle_{1:\nTimeSteps} \times \diff a_{1:\nTimeSteps-1} \times \diff k_T)\\
  & \coloneqq \dUnif_{[N]_0^T}(\diff j_{1:T}) \dDirac_{\state_{1:T}}(\diff \particle_1^{j_1} \times \dotsb \times \diff \particle_T^{j_T}) \biggl[ \prod_{t=1}^T  \SamplingKernelRandomWalk{\timeIndex,\nDimensions}{\nParticles}(\particle_\timeIndex^{j_t}, \diff \particle_\timeIndex^{-j_t}) \biggr]\\[-2ex]
  & \quad \times \Biggl[ \prod_{\smash{\timeIndex = 2}}^{\smash{\nTimeSteps}} \dDirac_{j_{t-1}}(\diff a_{\timeIndex - 1}^{j_t}) \prodSubstackAligned{n}{=}{0}{n}{\neq}{j_t}^N
  \ResamplingKernelRandomWalk{\timeIndex,\nDimensions}{\nParticles}((\particle_{\timeIndex-1:\timeIndex}, a_{\timeIndex-1}), \{a_{t-1}^n\}) \Biggr]\\[-2ex]
  & \quad \times \ResamplingKernel{\nTimeSteps, \nDimensions}{\nParticles}(\particle_\nTimeSteps, \diff\outputParticleIndex_\nTimeSteps)\\
  & \quad \times 
  \begin{cases}
  \biggl[\,
   \smashoperator{\prod_{\timeIndex = 1}^{\smash{\nTimeSteps-1}}} \dDirac_{a_t^{k_{t+1}}}(\diff k_t)\biggr] & \text{[without backward sampling]\!\!\!\!\!\!\!\!\!\!}\\
   \biggl[\,\smashoperator{\prod_{\timeIndex = 1}^{\nTimeSteps-1}} \BackwardKernelRandomWalk{\timeIndex,\nDimensions}{\nParticles}((\particle_{\timeIndex-1:\timeIndex}, a_{\timeIndex-1}, \particle_{\timeIndex+1}^{\outputParticleIndex_{\timeIndex+1}}), \{k_t\})\biggr] & \text{[with backward sampling]}
  \end{cases}\\
  & \quad \times \dDirac_{(\particle_1^{k_1}, \dotsc, \particle_T^{k_T})}(\diff \state_{1:T}'). \label{eq:rwcsmc_generalised_joint_law}
\end{align}
Here, we have defined $\smash{\particle_t^{-n} \coloneqq (\particle_t^0, \dotsc, \particle_t^{n-1}, \particle_t^{n+1}, \dotsc, \particle_t^N)}$ as well as the following quantities.
\begin{itemize}
 \item \textbf{Proposal kernels.} For any $t \in [T]$ and any $n \in [N]_0$, as in the \gls{RWEHMM} algorithm,
\begin{align}
 \SamplingKernelRandomWalk{\timeIndex,\nDimensions}{\nParticles}(\particle_\timeIndex^n, \diff \particle_\timeIndex^{-n}) \coloneqq \prod_{\dimensionIndex = 1}^\nDimensions \dN\bigl(\diff \particleSingle_{\timeIndex,\dimensionIndex}^{-n};\particleSingle_{\timeIndex, \dimensionIndex}^n \unitMat_{\nParticles}, \tfrac{\ell_\timeIndex}{\nDimensions}\varSigma\bigr),
\end{align}
 where $z_{t,d}^{-n} \coloneqq (z_{t,d}^0, \dotsc, z_{t,d}^{n-1}, z_{t,d}^{n+1}, \dotsc, z_{t,d}^N)$.

 \item \textbf{Resampling kernels.} For any $\particleIndex \in [\nParticles]_0$ and any $\timeIndex \in [\nTimeSteps]$,
\begin{align}
  \ResamplingKernelRandomWalk{\timeIndex,\nDimensions}{\nParticles}((\particle_{\timeIndex-1:\timeIndex}, a_{\timeIndex-1}), \{\particleIndex\})
  & \coloneqq 
   \selectionFunctionBoltzmann{\particleIndex}(\{\logWeightRandomWalk_{\timeIndex}(\particle_{\timeIndex-1}^{\mathrlap{a_{\timeIndex-1}^\particleIndexAlt}}, \particle_{\timeIndex}^\particleIndexAlt) - \logWeightRandomWalk_\timeIndex(\particle_{\timeIndex-1}^{\mathrlap{a_{\timeIndex-1}^0}}, \particle_{\timeIndex}^0)\}_{\particleIndexAlt = 1}^\nParticles)\\
   & =
      \dfrac{\mutation_t(\particle_{t-1}^{a_{t - 1}^n}, \particle_t^n)\Potential_t(\particle_t^n)}{\sum_{\particleIndexAlt=0}^\nParticles \mutation_t(\particle_{t-1}^{a_{t-1}^\particleIndexAlt}, \particle_t^\particleIndexAlt) \Potential_t(\particle_t^\particleIndexAlt)},
\end{align}
When using the forced-move extension, we replace $\smash{\ResamplingKernelRandomWalk{T,\nDimensions}{\nParticles}((\particle_{T-1:T}, a_{T-1}), \{k_T\})}$ at time $t=T$ in \eqref{eq:csmc_generalised_joint_law} by
\begin{align}
 \begin{dcases}
   \dfrac{\mathbf{h}^{k_T}}{\sum_{\particleIndexAlt = 0}^\nParticles \mathbf{h}^m - \mathbf{h}^{k_T} \wedge \mathbf{h}^{j_T}}, & \text{if $k_T \neq j_T$,}\\
   1 - \sumSubstackAligned{l}{=}{0}{l}{\neq}{j_T}^N \dfrac{\mathbf{h}^{l}}{\sum_{\particleIndexAlt = 0}^\nParticles \mathbf{h}^m - \mathbf{h}^{l} \wedge \mathbf{h}^{j_T}}, & \text{if $k_T = j_T$,}
 \end{dcases}
\end{align}
where we have defined the shorthand $\smash{\mathbf{h}^n \coloneqq \mutation_T(\particle_{T-1}^{a_{T - 1}^{n}}, \particle_T^{n}) \Potential_T(\particle_T^{n})}$.

\item \textbf{Backward kernels.} For $t \in [T-1]$,
 \begin{align}
  \MoveEqLeft \BackwardKernelRandomWalk{\timeIndex,\nDimensions}{\nParticles}((\particle_{\timeIndex-1:\timeIndex}, a_{\timeIndex-1}, \particle_{\timeIndex+1}^{\outputParticleIndex_{\timeIndex+1}}), \{\particleIndex\})\\
  & \coloneqq 
   \selectionFunctionBoltzmann{\particleIndex}(\{
    \logBackwardWeightRandomWalk_\timeIndex(\particle_{\timeIndex-1}^{\mathrlap{a_{\timeIndex-1}^\particleIndexAlt}}, \particle_{\timeIndex}^\particleIndexAlt, \particle_{\timeIndex+1}^{\outputParticleIndex_{\timeIndex+1}}) - \logBackwardWeightRandomWalk_\timeIndex(\particle_{\timeIndex-1}^{\mathrlap{a_{\timeIndex-1}^0}}, \particle_{\timeIndex}^0, \particle_{\timeIndex+1}^{\outputParticleIndex_{\timeIndex+1}})
  \}_{\particleIndexAlt = 1}^\nParticles)\\
  & =
    \frac{\mutation_\timeIndex(\particle_{\timeIndex - 1}^{a_{\timeIndex - 1}^{n}}, \particle_\timeIndex^{n}) \Potential_\timeIndex(\particle_\timeIndex^{n}) \mutation_{\timeIndex+1}(\particle_\timeIndex^{n}, \particle_{\timeIndex+1}^{\outputParticleIndex_{\timeIndex+1}})}{\sum_{\particleIndexAlt=0}^\nParticles \mutation_\timeIndex(\particle_{\timeIndex - 1}^{\mathrlap{a_{\timeIndex - 1}^\particleIndexAlt}}, \particle_\timeIndex^\particleIndexAlt) \Potential_\timeIndex(\particle_\timeIndex^\particleIndexAlt) \mutation_{\timeIndex+1}(\particle_\timeIndex^\particleIndexAlt, \particle_{\timeIndex+1}^{\outputParticleIndex_{\timeIndex+1}})}.
\end{align}
\end{itemize}

From this definition, we can recover the joint law of all random variables $(\Particle_{1:T}, A_{1:T-1}, K_{1:T}, \State_{1:T}')$ generated in Steps~\ref{alg:iterated_rwcsmc:1}--\ref{alg:iterated_rwcsmc:3} of Algorithm~\ref{alg:iterated_rwcsmc} by conditioning on the event $\{J_1 = 0, \dotsc, J_T = 0\}$, i.e.\
\begin{align}
 \IteratedCsmcKernelRandomWalk{T,D,\state_{1:\nTimeSteps}}{N} \coloneqq \IteratedCsmcKernelRandomWalk{T,D,\state_{1:\nTimeSteps}}{N,\star}(\ccdot|J_1 = 0, \dotsc, J_T = 0).
\end{align}
Let $\ExpectationCsmcKernelRandomWalk{T,D,\state_{1:\nTimeSteps}}{N,\star}$ denote expectation w.r.t.\ $\IteratedCsmcKernelRandomWalk{T,D,\state_{1:\nTimeSteps}}{N,\star}$. In the remainder of this section, we will sometimes work with $\smash{\IteratedCsmcKernelRandomWalk{T,D,\state_{1:\nTimeSteps}}{N,\star}}$ rather than $\smash{\IteratedCsmcKernelRandomWalk{T,D,\state_{1:\nTimeSteps}}{N,\star}}$. This is justified because both versions of the \gls{IRWCSMC} algorithm induce the same Markov kernel:
\begin{align}
 \ExpectationCsmcKernelRandomWalk{\nTimeSteps,\nDimensions,\state_{1:\nTimeSteps}}{\nParticles,\star}[\ind\{\State_{1:T}' \in \diff \state_{1:T}'\}] 
 & = \ExpectationCsmcKernelRandomWalk{\nTimeSteps,\nDimensions,\state_{1:\nTimeSteps}}{\nParticles}[\ind\{\State_{1:T}' \in \diff \state_{1:T}'\}]\\
 & = \InducedIteratedCsmcKernelRandomWalk{\nTimeSteps,\nDimensions}{\nParticles}(\state_{1:\nTimeSteps}, \diff \state_{1:T}'). \label{eq:same_induced_markov_kernel_rwcsmc}
\end{align}

\subsection{Proof of Propositions~\ref{prop:discrete_markov_kernel_rwcsmc_without_backward_sampling}. \ref{prop:discrete_markov_kernel_rwcsmc_with_backward_sampling} and \ref{prop:rwcsmc_invariance}}
\label{app:subsec:rwcsmc_invariance}

In this section, we prove Propositions~\ref{prop:discrete_markov_kernel_rwcsmc_without_backward_sampling}, \ref{prop:discrete_markov_kernel_rwcsmc_with_backward_sampling} and \ref{prop:rwcsmc_invariance} using the slightly generalised extended state-space construction defined above.

The proof of Proposition~\ref{prop:rwcsmc_invariance} proceeds along the same lines as the proof of Proposition~\ref{prop:csmc_invariance}.

\begin{namedproof}[of Proposition~\ref{prop:rwcsmc_invariance}]

 Recall that by \eqref{eq:rwcsmc_symmetric_proposal:0}, the random-walk type proposal used to scatter the particles around the reference path is symmetric in the sense that
 \begin{align}
  \lambda(\diff \particle_t^j) \SamplingKernelRandomWalk{\timeIndex,\nDimensions}{\nParticles}(\particle_\timeIndex^j, \diff \particle_\timeIndex^{-j}) = \lambda(\diff \particle_t^k) \SamplingKernelRandomWalk{\timeIndex,\nDimensions}{\nParticles}(\particle_\timeIndex^k, \diff \particle_\timeIndex^{-k}), \label{eq:rwcsmc_symmetric_proposal}
 \end{align}
 for any $j, k \in [N]_0$, where $\smash{\particle_t^{-n} \coloneqq (\particle_t^0, \dotsc, \particle_t^{n-1}, \particle_t^{n+1}, \dotsc, \particle_t^N)}$ and where $\lambda$ denotes a suitable version of the Lebesgue measure.

 For the plain algorithm (with neither the backward sampling nor the forced-move extension), we can then readily check that
 \begin{align}
  \MoveEqLeft \Target_{T,D}(\diff \state_{1:T}) \IteratedCsmcKernelRandomWalk{\nTimeSteps,\nDimensions,\state_{1:T}}{\nParticles,\star}(\diff j_{1:T} \times \diff \particle_{1:\nTimeSteps} \times \diff a_{1:\nTimeSteps-1} \times \diff\outputParticleIndex_{1:\nTimeSteps} \times \diff \state_{1:T}')\\
  & \!\!\!\!\!\!  =  \Target_{T,D}(\diff \state_{1:T}') \IteratedCsmcKernelRandomWalk{\nTimeSteps,\nDimensions,\state_{1:\nTimeSteps}'}{\nParticles,\star}(\diff k_{1:T} \times \diff \particle_{1:\nTimeSteps} \times \diff a_{1:\nTimeSteps-1} \times \diff j_{1:\nTimeSteps} \times \diff \state_{1:T}), \label{eq:prop:rwcsmc_invariance:1}
 \end{align}
 i.e.\ \eqref{eq:prop:rwcsmc_invariance:1} admits $\Target_{T,D}(\diff \state_{1:T}')$ as a marginal.
 
 For the backward-sampling extension, let $\smash{\widehat{\mathbb{P}}_{T,D,\state_{1:T}}^{N,\star}}$ be the same as $\smash{\IteratedCsmcKernelRandomWalk{\nTimeSteps,\nDimensions,\state_{1:T}}{\nParticles,\star}}$ (without backward sampling) except that the terms $\smash{ \dDirac_{j_{t-1}}(\diff a_{\timeIndex - 1}^{j_t})}$ in \eqref{eq:rwcsmc_generalised_joint_law} are replaced by $\smash{\BackwardKernelRandomWalk{t-1,\nDimensions}{\nParticles}((\particle_{\timeIndex-2:\timeIndex-1}, a_{\timeIndex-2}, \particle_{\timeIndex}^{\outputParticleIndex_{\timeIndex}}), \{a_{t-1}^{j_t}\})}$. Then 
  \begin{align}
  \MoveEqLeft \Target_{T,D}(\diff \state_{1:T}) \IteratedCsmcKernelRandomWalk{\nTimeSteps,\nDimensions,\state_{1:T}}{\nParticles,\star}(\diff j_{1:T} \times \diff \particle_{1:\nTimeSteps} \times \diff a_{1:\nTimeSteps-1} \times \diff\outputParticleIndex_{1:\nTimeSteps} \times \diff \state_{1:T}')\\
  & \!\!\!\!\!\! = \Target_{T,D}(\diff \state_{1:T}') \widehat{\mathbb{P}}_{T,D,\state_{1:\nTimeSteps}'}^{N,\star}(\diff k_{1:T} \times \diff \particle_{1:\nTimeSteps} \times \diff a_{1:\nTimeSteps-1} \times \diff j_{1:\nTimeSteps} \times \diff \state_{1:T}). \label{eq:prop:rwcsmc_invariance:2}
 \end{align}
  That is, \eqref{eq:prop:rwcsmc_invariance:2} again admits $\Target_{T,D}(\diff \state_{1:T}')$ as a marginal. Incidentally, $\smash{\widehat{\mathbb{P}}_{T,D,\state_{1:\nTimeSteps}}^{N,\star}}$ can be recognised as the law of all the random variables generated by an \gls{IRWCSMC} algorithm with ancestor sampling. This shows that ancestor sampling is a valid alternative to backward sampling in this algorithm. 
  
  Finally, use of the forced-move extension can again be justified as a partially collapsed Gibbs sampler because applying this extension in Step~\ref{alg:iterated_rwcsmc:2a} of Algorithm~\ref{alg:iterated_rwcsmc} leaves the marginal distribution of $K_T$ conditional on $(\State_{1:T}, J_{1:T}, \Particle_{1:T}, A_{1:T})$ invariant.  \qedwhite
\end{namedproof}

\begin{namedproof}[of Proposition~\ref{prop:discrete_markov_kernel_rwcsmc_without_backward_sampling}]
 Let
 \begin{align}
  \MoveEqLeft \varXi_T((\particle_{1:T}, j_{1:T}), \diff a_{1:T-1} \times \diff k_{1:T})\\
  & \coloneqq \biggl[ \prod_{\smash{\timeIndex = 2}}^{\smash{T}} 
  \dDirac_{0}(\diff a_{\timeIndex - 1}^{j_t}) 
 \prodSubstackAligned{n}{=}{0}{n}{\neq}{j_t}^{\smash{N}}
  \ResamplingKernelRandomWalk{\timeIndex-1,D}{N}((\particle_{\timeIndex-2:\timeIndex-1}, a_{\timeIndex-2}), \diff a_{\timeIndex - 1}^\particleIndex)\biggr]\\[-1.5ex]
  & \quad \times \ResamplingKernelRandomWalk{T, D}{N}((\particle_{T-1:T}, a_{T-1}), \diff\outputParticleIndex_T) \prod_{t=1}^{\smash{T-1}} \dDirac_{a_t^{k_{t+1}}}(\diff k_t).
 \end{align}
 Simple algebra then verifies that
 \begin{align}
  \MoveEqLeft \xi_T(\particle_{1:T}, \diff j_{1:T})\varXi_T((\particle_{1:T}, j_{1:T}), \diff a_{1:T-1} \times \diff k_{1:T})\\
  & = \xi_T(\particle_{1:T}, \diff k_{1:T})\varXi_T((\particle_{1:T}, k_{1:T}), \diff a_{1:T-1} \times \diff j_{1:T}).
 \end{align}
 This completes the proof. \qedwhite
\end{namedproof}

\begin{namedproof}[of Proposition~\ref{prop:discrete_markov_kernel_rwcsmc_with_backward_sampling}]
  Let
 \begin{align}
  \MoveEqLeft \varXi_T((\particle_{1:T}, j_{1:T}), \diff a_{1:T-1} \times \diff k_{1:T})\\
  & \coloneqq \biggl[ \prod_{\smash{\timeIndex = 2}}^{\smash{T}} 
  \dDirac_{0}(\diff a_{\timeIndex - 1}^{j_t}) 
 \prodSubstackAligned{n}{=}{0}{n}{\neq}{j_t}^{\smash{N}}
  \ResamplingKernelRandomWalk{\timeIndex-1,D}{N}((\particle_{\timeIndex-2:\timeIndex-1}, a_{\timeIndex-2}), \diff a_{\timeIndex - 1}^\particleIndex)\biggr]\\[-2ex]
  & \quad \times \ResamplingKernelRandomWalk{T, D}{N}((\particle_{T-1:T}, a_{T-1}), \diff\outputParticleIndex_T) \prod_{t=1}^{\smash{T-1}} \BackwardKernelRandomWalk{t,\nDimensions}{\nParticles}((\particle_{t-1:t}, a_{t-1}, \particle_{t+1}^{j_{t+1}}), \diff j_t\}),
 \end{align}
 and
 \begin{align}
  \MoveEqLeft \widehat{\varXi}_T((\particle_{1:T}, j_{1:T}), \diff a_{1:T-1} \times \diff k_{1:T})\\
  & \coloneqq \biggl[ \prod_{\smash{\timeIndex = 2}}^{\smash{T}} 
 \BackwardKernelRandomWalk{t-1,\nDimensions}{\nParticles}((\particle_{t-2:t-1}, a_{t-1}, \particle_{t}^{j_t}), \{a_{t-1}^{j_t}\})
 \prodSubstackAligned{n}{=}{0}{n}{\neq}{j_t}^{\smash{N}}
  \ResamplingKernelRandomWalk{\timeIndex-1,D}{N}((\particle_{\timeIndex-2:\timeIndex-1}, a_{\timeIndex-2}), \diff a_{\timeIndex - 1}^\particleIndex)\biggr]\\[-3ex]
  & \quad \times \ResamplingKernelRandomWalk{T, D}{N}((\particle_{T-1:T}, a_{T-1}), \diff\outputParticleIndex_T) \prod_{t=1}^{\smash{T-1}} \dDirac_{a_t^{k_{t+1}}}(\diff k_t).
 \end{align}
  Simple algebra then verifies that
 \begin{align}
  \MoveEqLeft \xi_T(\particle_{1:T}, \diff j_{1:T})\varXi_T((\particle_{1:T}, j_{1:T}), \diff a_{1:T-1} \times \diff k_{1:T})\\
  & = \xi_T(\particle_{1:T}, \diff k_{1:T})\widehat{\varXi}_T((\particle_{1:T}, k_{1:T}), \diff a_{1:T-1} \times \diff j_{1:T}).
 \end{align}
 This completes the proof. \qedwhite
\end{namedproof}

\subsection{Relationship with `unconditional' SMC algorithms}
\label{app:subsec:relationship_with_unconditional_smc}

In this section, we expand on the observation made in Remark~\ref{rem:relationship_with_unconditional_smc} that whilst standard \gls{CSMC} methods are closely related with a corresponding `unconditional' \gls{SMC} algorithm, no such `unconditional' counterpart exists for the \gls{IRWCSMC} algorithm.
 
As explained in \citet{andrieu2010particle}, standard \gls{CSMC} methods are closely linked to the justification of a corresponding `unconditional' \gls{SMC} algorithm in the sense that the law of all the particles and parent indices generated by the latter, $\mathbb{Q}_{T,D}^{N,\star}$, is:
 \begin{align}
  \MoveEqLeft \frac{\E[\ExpectationCsmcKernel{T,D,\State_{1:T}}{N,\star}(\ind\{(\Particle_{1:T}, A_{1:T-1}) \in  \diff \particle_{1:T} \times \diff a_{1:T-1} \})]}{\prod_{t=1}^T \frac{1}{N+1}\sum_{n=0}^N \Potential_t(\particle_t^n)}\\
  & \propto \smashoperator{\prod_{\particleIndex = 0}^\nParticles} \Mutation_1(\diff \particle_1^\particleIndex) \prod_{\smash{t = 2}}^{\smash{T}}\biggl[\,\smashoperator{\prod_{\particleIndex = 0}^{\smash{\nParticles}}} \ResamplingKernel{t-1,\nDimensions}{\nParticles}(\particle_{t - 1}, \diff a_{t - 1}^\particleIndex)
  \Mutation_t(\particle_{t - 1}^{a_{t - 1}^\particleIndex}, \diff \particle_t^\particleIndex)\biggr]\\
  & \eqqcolon \mathbb{Q}_{T,D}^{N,\star}(\diff \particle_{1:T} \times \diff a_{1:T-1}),
 \end{align}
 where $\State_{1:T} \sim \Target_{T,D}$, and where -- to avoid complications arising from the division by zero -- we assume that $\mutation_t$ and $\Potential_t$ are strictly positive.
 
 However, for the \gls{IRWCSMC} algorithm, no such `unconditional' \gls{SMC} algorithm exists. To see this, let $\lambda$ denotes a suitable version of the Lebesgue measure. Then by \eqref{eq:rwcsmc_symmetric_proposal}, the measure
  \begin{align}
  \MoveEqLeft \frac{\E[\ExpectationCsmcKernelRandomWalk{T,D,\State_{1:T}}{N,\star}(\ind\{(\Particle_{1:T}, A_{1:T-1}) \in  \diff \particle_{1:T} \times \diff a_{1:T-1} \})]}{\prod_{t=1}^T \frac{1}{N+1}\sum_{n=0}^N \mutation_{t}(\particle_{t-1}^{a_{t-1}^n},\particle_t^n) \Potential_t(\particle_t^n)}\\
  & \propto \biggl[ \prod_{t = 1}^T \lambda(\diff \particle_t^0) \SamplingKernelRandomWalk{t,\nDimensions}{\nParticles}(\particle_t^0, \diff \particle_t^{1:N})\biggr]\\
  & \quad \times \biggl[ \prod_{\smash{t = 2}}^{\smash{T}} 
  \smashoperator{\prod_{\smash{\particleIndex = 0}}^{\smash{N}}}
  \ResamplingKernelRandomWalk{t-1,\nDimensions}{\nParticles}((\particle_{t-2:t-1}, a_{\timeIndex-2}), \diff a_{\timeIndex - 1}^\particleIndex)\biggr]\\
  & \eqqcolon \widebar{\mathbb{Q}}_{T,D}^{N,\star}(\diff \particle_{1:T} \times \diff a_{1:T-1}),
 \end{align}
 is not finite and hence there does not exist an algorithm that samples from it.

\subsection{Formal definition of the limiting law}
\label{app:subsec:formal_definition_of_the_limiting_law}

In this section, we give a more formal definition of the limiting law of the genealogies (i.e.\ of the ancestor indices, $A_t^n$) and of the particle indices of the new reference path, $K_t$, under the \gls{IRWCSMC} algorithm that appears in Section~\ref{subsec:rwcsmc_scaling}. 
\begin{align}
  \MoveEqLeft\IteratedCsmcKernelRandomWalk{\nTimeSteps}{\nParticles}(\diff v_{1:\nTimeSteps} \times \diff w_{1:\nTimeSteps} \times \diff a_{1:\nTimeSteps-1} \times \diff k_{1:T})\\
  & \coloneqq \biggl[ \prod_{\timeIndex = 1}^{\smash{\nTimeSteps}} \dN(\diff [v_t, w_t]^\T; \bar{\mu}_{t|T}, \widebar{\varSigma}_{t|T}) \biggr]\\
  & \quad \times \biggl[\prod_{\timeIndex = 1}^{\smash{\nTimeSteps-1}} 
  \dDirac_{0}(\diff a_{\timeIndex}^0) 
  \smashoperator{\prod_{\particleIndex = 1}^{\smash{N}}}
  \ResamplingKernelRandomWalk{\timeIndex|\nTimeSteps}{\nParticles}((v_t, w_{t-1}, a_{\timeIndex-1}), \diff a_{\timeIndex}^\particleIndex)\biggr]\\
  & \quad \times \ResamplingKernelRandomWalk{\nTimeSteps|\nTimeSteps}{\nParticles}((v_T, w_{T-1}, a_{\nTimeSteps-1}), \diff\outputParticleIndex_\nTimeSteps)\\
  & \quad \times 
  \begin{cases}
  \biggl[\,
   \smashoperator{\prod_{\timeIndex = 1}^{\smash{\nTimeSteps-1}}} \dDirac_{a_t^{k_{t+1}}}(\diff k_t)\biggr] & \text{[without backward sampling]\!\!\!\!\!\!\!\!\!\!}\\
   \biggl[\,\smashoperator{\prod_{\timeIndex = 1}^{\nTimeSteps-1}} \BackwardKernelRandomWalk{\timeIndex|\nTimeSteps}{\nParticles}((v_t, w_{t-1:t}, a_{\timeIndex-1}), \diff\outputParticleIndex_\timeIndex)\biggr]. & \text{[with backward sampling]}
  \end{cases}
\end{align}
Here, we have defined the following quantities.
\begin{itemize}
 \item \textbf{Asymptotic resampling kernels.} For any $\particleIndex \in [\nParticles]_0$ and any $\timeIndex \in [\nTimeSteps]$,
 \begin{align}
  \ResamplingKernelRandomWalk{\timeIndex|\nTimeSteps}{\nParticles}((v_t, w_{t-1}, a_{\timeIndex-1}), \{\particleIndex\})
  & \coloneqq 
  \selectionFunctionBoltzmann{\particleIndex}(\{ 
   v_t^m + w_{t-1}^{a_{t-1}^m}
  \}_{\particleIndexAlt = 1}^\nParticles),
\end{align}
where we recall the convention that $w_{t-1}^0 \equiv 0$. As usual, when using the forced-move extension, we replace $\selectionFunctionBoltzmann{\particleIndex}$ by $\selectionFunctionRosenbluth{\particleIndex}$ in the definition above at time $t = T$.

 \item \textbf{Asymptotic backward kernels.} For any $\particleIndex \in [\nParticles]_0$ and any $\timeIndex \in [\nTimeSteps-1]$,
\begin{align}
  \BackwardKernelRandomWalk{\timeIndex|\nTimeSteps}{\nParticles}((v_t, w_{t-1:t}, a_{\timeIndex-1}), \{\particleIndex\})
  & \coloneqq 
  \selectionFunctionBoltzmann{\particleIndex}(\{
   v_t^m + w_t^m + w_{t-1}^{a_{t-1}^m}
  \}_{\particleIndexAlt = 1}^\nParticles).
\end{align}
\end{itemize}

\subsection{Proof of Proposition~\ref{prop:limiting_rwcsmc_algorithm}}
\label{app:subsec:prop:limiting_rwcsmc_algorithm}

In this section, we prove Proposition~\ref{prop:limiting_rwcsmc_algorithm}. The proof can be viewed as an extension of the proof of \citet[][Lemma~10]{bedard2012scaling} to the case that $T > 1$. It relies on a Taylor-series expansion and a few technical lemmata which we state first.

\begin{lemma}\label{lem:selection_functions_lipschitz}
 For any $\nParticles \in \naturals$ and $n \in [N]_0$ Boltzmann selection function $\selectionFunctionBoltzmann{n}$ and Rosenbluth--Teller selection function $\selectionFunctionRosenbluth{n}$ are Lipschitz-continuous. \tendmark
\end{lemma}
\begin{proof}
 This can be verified by checking that the absolute value of the gradient is almost everywhere bounded. 
\end{proof}

\paragraph*{Main decomposition.}

Throughout the remainder of this subsection, we fix some $N, T \in \naturals$. For any $\state_{1:T} \in \spaceState_{T,D}$ and any $t \in [T]$, define
\begin{align}
  \calV_{t|T} & \coloneqq \pi_T([\partial_t \logWeightRandomWalkSingle_t]^2),\\
  \calW_{t|T} & \coloneqq \pi_T([\partial_t \logWeightRandomWalkSingle_{t+1}]^2),\\
  \calS_{t|T} & \coloneqq \pi_T([\partial_t \logWeightRandomWalkSingle_t][\partial_t \logWeightRandomWalkSingle_{t+1}]),\\
  \calV_{t|T}(\state_{t-1:t}) & \coloneqq \textstyle D^{-1} \sum_{d=1}^D \{[\partial_t \logWeightRandomWalkSingle_t](x_{t-1:t,d})\}^2,\\
  \calW_{t|T}(\state_{t:t+1}) & \coloneqq \textstyle D^{-1} \sum_{d=1}^D \{[\partial_t \logWeightRandomWalkSingle_{t+1}](x_{t:t+1,d})\}^2,\\
  \calS_{t|T}(\state_{t-1:t+1}) & \coloneqq \textstyle D^{-1}\sum_{d=1}^D \{[\partial_t \logWeightRandomWalkSingle_t](x_{t-1:t,d})\}\{[\partial_t \logWeightRandomWalkSingle_{t+1}](x_{t:t+1,d})\},\\
  \widetilde{\calV}_{t|T}(\state_{t-1:t}) & \coloneqq \textstyle [D \calV_{t|T}(\state_{t-1:t})]^{-1/2} \sup_{d \in [D]}\lvert [\partial_t \logWeightRandomWalkSingle_t](x_{t-1:t,d})\rvert ,\\
  \widetilde{\calW}_{t|T}(\state_{t:t+1}) & \coloneqq \textstyle [D \calW_{t|T}(\state_{t:t+1})]^{-1/2} \sup_{d \in [D]}\lvert [\partial_t \logWeightRandomWalkSingle_{t+1}](x_{t:t+1,d})\rvert,
\end{align}
where we recall the convention that $\logWeightRandomWalkSingle_{T+1} \equiv 0$. With this notation, for some $0 \leq \eta < 1/4$, define the following family of Borel sets:
\begin{align}
 \ConcentrationSet_{T,D}
 & \coloneqq 
 \left\{ \state_{1:T} \in \spaceState_{T,D} \;\middle|\; \sup_{t \in [T]}
 \begin{bmatrix}
  \lvert \calV_{t|T}(\state_{t-1:t}) - \calV_{t|T}\rvert \vee \\
  \lvert \calW_{t|T}(\state_{t:t+1}) - \calW_{t|T}\rvert \vee \\\
  \lvert \calS_{t|T}(\state_{t-1:t+1}) - \calS_{t|T}\rvert \vee\\
  \widetilde{\calV}_{t|T}(\state_{t-1:t}) \vee \widetilde{\calW}_{t|T}(\state_{t:t+1})
 \end{bmatrix}
 \leq D^{-\eta} \right\}.
 \label{eq:concentration_set_rwcsmc}
\end{align}


\begin{lemma}\label{lem:concentration_set_rwcsmc}
 For any $T \in \naturals$, $\lim_{D \to \infty} \Target_{T,D}(\ConcentrationSet_{T,D}) = 1$. \tendmark
\end{lemma}
\begin{proof}
 The results 
 \begin{gather}
  \Target_{T,D}(\{\state_{1:T} \in \spaceState_{T,D} \mid \lvert \calV_{t|T}(\state_{t-1:t}) - \calV_{t|T}\rvert \leq D^{-\eta}\})  \to 1,\label{eq:lem_concentration_set_rwcsmc:1}\\
  \Target_{T,D}(\{\state_{1:T} \in \spaceState_{T,D} \mid \lvert \calW_{t|T}(\state_{t:t+1}) - \calW_{t|T}\rvert \leq D^{-\eta}\})  \to 1,\label{eq:lem_concentration_set_rwcsmc:2}\\
  \Target_{T,D}(\{\state_{1:T} \in \spaceState_{T,D} \mid \lvert \calS_{t|T}(\state_{t-1:t+1}) - \calS_{t|T}\rvert \leq D^{-\eta}\})  \to 1,
 \end{gather}
 follow from the law of the iterated logarithm since $\eta < 1/2$. To prove
  \begin{align}
  \Target_{T,D}(\{\state_{1:T} \in \spaceState_{T,D} \mid \widetilde{\calV}_{t|T}(\state_{t-1:t}) \leq D^{-\eta}\}) & \to 1,
 \end{align}
 we argue as in \citet{bedard2012scaling} that, by \eqref{eq:lem_concentration_set_rwcsmc:1}, it suffices to show that for any $c > 0$, 
 \begin{align}
  \MoveEqLeft \textstyle \Target_{T,D}(\{\state_{1:T} \in \spaceState_{T,D} \mid \sup_{d \in [D]}\lvert [\partial_t \logWeightRandomWalkSingle_t](x_{t-1:t,d})\rvert \leq c D^{1/2-\eta}\})\\
  & = \textstyle \Prob(\{\sup_{d \in [D]}\lvert [\partial_t \logWeightRandomWalkSingle_t](X_{t-1:t,d})\rvert \leq c D^{1/2-\eta}\})\\
  & = \textstyle [1 - \Prob(\{\lvert [\partial_t \logWeightRandomWalkSingle_t](X_{t-1:t,1})\rvert > c D^{1/2-\eta}\})]^D\\
  & \to 1,
 \end{align}
 where $(X_{1:T,d})_{d \geq 1}$ be \gls{IID} samples from $\TargetSingle_T$. But this holds since Markov's inequality along with the fact that $\eta < 1/4$ and \ref{as:moments_bounded} ensures that
 \begin{align}
  \Prob(\{\lvert [\partial_t \logWeightRandomWalkSingle_t](X_{t-1:t,1})\rvert > c D^{1/2-\eta}\})
  & \leq \textstyle \pi_T(\lvert \partial_t \logWeightRandomWalkSingle_t \rvert^4) c^{-1} D^{4(\eta-1/2)} = \lo(D^{-1}).
 \end{align}
 The result $\smash{\Target_{T,D}(\{\state_{1:T} \in \spaceState_{T,D} \mid \widetilde{\calW}_{t|T}(\state_{t:t+1}) \leq D^{-\eta}\}) \to 1}$ follows by the same arguments. 
\end{proof}

In the remainder of this subsection, we let $\smash{(\state_{1:T,D})_{D \geq 1}}$ be some sequence in $(\spaceState_{T,D})_{D \geq 1}$, i.e.\ $\state_{t,D} = x_{t,1:D,D} \in \reals^{D}$, for any $D \geq 1$. We shall also often use the shorthand $\sup_{\ConcentrationSet_{\nTimeSteps, \nDimensions}}$ for $\sup_{\state_{1:\nTimeSteps,D} \in \ConcentrationSet_{\nTimeSteps, \nDimensions}}$. We then set
 \begin{align}
  \Particle_{t,D}^n \coloneqq
  \begin{cases}
    \state_{t,D}, & \text{if $n = 0$,}\\
    \state_{t,D} + \sqrt{\tfrac{\ell_{t}}{D}}\mathbf{U}_{t+1,D}^n, & \text{if $n \in [N]$,}
  \end{cases}
 \end{align}
 where $\smash{\mathbf{U}_{t,D}^n \coloneqq U_{t,1:D,D}^n}$, with $\smash{U_{t,d,D}^{1:N} \sim \dN(\zeroMat_{N}, \varSigma)}$ for $\varSigma \coloneqq \tfrac{1}{2}(\iMat_N + \unitMat_N \unitMat_N^\T)$ and where $\smash{U_{t,d,D}^{1:N}}$ and $\smash{U_{t,e,D}^{1:N}}$ are independent whenever $s \neq t$ or $d \neq e$. We also fix some $a_{t}^n \in [N]_0$ for all $(t,n) \in [T-1] \times [N]$ and some $k_t \in [N]_0$ for all $t \in [T]$. 
 
 A second-order Taylor-series expansion then gives 
 \begin{align}
  \MoveEqLeft \logWeightRandomWalk_t(\Particle_{t-1,D}^{a_{t-1}^n}, \Particle_{t,D}^n) - \logWeightRandomWalk_t(\Particle_{t-1,D}^0, \Particle_{t,D}^0)\\
  & = \textstyle V_{t,D}^n + W_{t-1,D}^{a_{t-1}^n}  + \sum_{i=1}^4 (R_{t,D}^{n,i} + S_{t-1,D}^{a_{t-1}^n}) + \sum_{i=1}^2 T_{t-1,t,D}^{a_{t-1}^n,n,i},
 \end{align}
 as well as (for $t < T$),
 \begin{align}
  \MoveEqLeft \logBackwardWeightRandomWalk_t(\Particle_{t-1,D}^{a_{t-1}^n}, \Particle_{t,D}^n, \Particle_{t+1,D}^{k_{t+1}}) - \logBackwardWeightRandomWalk_t(\Particle_{t-1,D}^0, \Particle_{t,D}^0, \Particle_{t+1,D}^{k_{t+1}})\\
  & = \textstyle V_{t,D}^n + W_{t-1,D}^{a_{t-1}^n} + \sum_{i=1}^4 (R_{t,D}^{n,i} + S_{t-1,D}^{a_{t-1}^n}) + \sum_{i=1}^2 T_{t-1,t,D}^{a_{t-1}^n,n,i}\\
  & \qquad + \textstyle W_{t,D}^n + \sum_{i=1}^4 S_{t,D}^{n,i},
 \end{align}
 where
 \begin{align}
  V_{t,D}^n & \coloneqq \sqrt{\frac{\calV_{t|T}}{\calV_{t|T}(\state_{t-1:t,D})}}\sqrt{\frac{\ell_t}{D}} \sum_{d=1}^D [\partial_t \logWeightRandomWalkSingle_t](x_{t-1:t,d,D})U_{t,d,D}^n + \frac{\ell_t}{2}\pi_T(\partial_t^2 \logWeightRandomWalkSingle_t),\\
  W_{t,D}^{m} & \coloneqq \sqrt{\frac{\calW_{t|T}}{\calW_{t|T}(\state_{t:{t+1},D})}}\sqrt{\frac{\ell_t}{D}}\sum_{d=1}^D [\partial_t \logWeightRandomWalkSingle_{t+1}](x_{t:t+1,d,D})U_{t,d,D}^m + \frac{\ell_t}{2}\pi_T(\partial_t^2 \logWeightRandomWalkSingle_{t+1}),
  \end{align}
 and with
 \begin{align}
  R_{t,D}^{n,1} & \coloneqq \textstyle \Bigl\{1 - \sqrt{\frac{\calV_{t|T}}{\calV_{t|T}(\state_{t-1:t,D})}}\Bigr\} \sqrt{\tfrac{\ell_t}{D}} \sum_{d=1}^D [\partial_t \logWeightRandomWalkSingle_t](x_{t-1:t,d,D})U_{t,d,D}^n\\
  R_{t,D}^{n,2} & \coloneqq \textstyle \tfrac{\ell_t}{2D}\sum_{d=1}^D [\partial_t^2 \logWeightRandomWalkSingle_t](x_{t-1:t,d,D}) - \tfrac{\ell_t}{2}\pi_T(\partial_t^2 \logWeightRandomWalkSingle_t)\\
  R_{t,D}^{n,3} &\coloneqq \textstyle \tfrac{\ell_t}{2D}\sum_{d=1}^D [\partial_t^2 \logWeightRandomWalkSingle_t](x_{t-1:t,d,D})\{(U_{t,d,D}^n)^2 - 1\}\\
  R_{t,D}^{n,4} & \coloneqq \textstyle \tfrac{\ell_t}{2D}\sum_{d=1}^D \{[\partial_t^2 \logWeightRandomWalkSingle_t](x_{t-1:t,d,D} + \xi_{t,d,D}^n \sqrt{\ell_t}U_{t,d,D}^n)\\
  & \qquad \qquad \qquad - [\partial_t^2 \logWeightRandomWalkSingle_t](x_{t-1:t,d,D})\}(U_{t,d,D}^n)^2\\
  S_{t,D}^{m,1} & \coloneqq \textstyle \Bigl\{1 - \sqrt{\frac{\calW_{t|T}}{\calW_{t|T}(\state_{t:t+1,D})}}\Bigr\} \sqrt{\tfrac{\ell_t}{D}} \sum_{d=1}^D [\partial_t \logWeightRandomWalkSingle_{t+1}](x_{t:t+1,d,D})U_{t,d,D}^m\\
  S_{t,D}^{m,2} & \coloneqq \textstyle \tfrac{\ell_t}{2D}\sum_{d=1}^D [\partial_t^2 \logWeightRandomWalkSingle_{t+1}](x_{t:t+1,d,D}) - \tfrac{\ell_t}{2}\pi_T(\partial_t^2 \logWeightRandomWalkSingle_{t+1})\\
  S_{t,D}^{m,3} & \coloneqq \textstyle \tfrac{\ell_t}{2D}\sum_{d=1}^D [\partial_t^2 \logWeightRandomWalkSingle_{t+1}](x_{t:t+1,d,D})\{(U_{t,d,D}^m)^2 - 1\}\\
  S_{t,D}^{m,4} & \coloneqq \textstyle \tfrac{\ell_t}{2D}\sum_{d=1}^D \{[\partial_t^2 \logWeightRandomWalkSingle_{t+1}](x_{t,d,D} + \eta_{t,d,D}^m \sqrt{\ell_t}U_{t,d,D}^m, x_{t+1,d,D})\\
  & \qquad \qquad \qquad - [\partial_t^2 \logWeightRandomWalkSingle_{t+1}](x_{t:t+1,d,D})\}(U_{t,d,D}^m)^2\\
 T_{t,t+1,D}^{m,n,1} & \coloneqq \textstyle \tfrac{\sqrt{\ell_t \ell_{t+1} }}{D} \sum_{d=1}^D [\partial_t\partial_{t+1} \logWeightRandomWalkSingle_{t+1}](x_{t:t+1,d,D})U_{t,d,D}^m U_{t+1,d,D}^n\\
  T_{t,t+1,D}^{m,n,2} & \coloneqq \textstyle \tfrac{\sqrt{\ell_t\ell_{t+1} }}{D} \sum_{d=1}^D \{ [\partial_t \partial_{t+1} \logWeightRandomWalkSingle_{t+1}](x_{t:t+1,d,D} + \eta_{t,d,D}^{m} \sqrt{\ell_t}U_{t,d,D}^m,\\
  & \qquad\qquad\qquad\qquad\qquad\qquad\qquad x_{t+1,d,D} + \xi_{t+1,d,D}^n \sqrt{\ell_{t+1}}U_{t+1,d,D}^n)\\
  & \qquad\qquad \qquad\qquad - [\partial_t \partial_{t+1} \logWeightRandomWalkSingle_{t+1}](x_{t:t+1,d,D})\}U_{t,d,D}^m U_{t+1,d,D}^n,
 \end{align}
 for some $\eta_{t,d,D}^m, \xi_{t,d,D}^n \in [0, D^{-1/2}]$, and with the usual convention that $\smash{V_{t,D}^n}$, $\smash{W_{t,D}^n}$, $\smash{R_{t,D}^{n,i}}$, and $\smash{S_{t,D}^{n,i}}$ are $0$ if $t = 0$ or $n = 0$. Similarly $\smash{T_{t-1,t,D}^{m,n,i} = 0}$ whenever $m=0$, $n=0$ or $t = 1$.

\begin{lemma}\label{lem:remainder_terms}
 Assume~\ref{as:iid_model} and \ref{as:moments_bounded}. 
 For any $t \in [T]$, $(m,n) \in [N]^2$, $i \in [4]$ and $j \in [2]$,
 \begin{enumerate}
  \item $\smash{\lim_{D \to \infty} \sup_{\ConcentrationSet_{T,D}} \E[\lvert R_{t,D}^{n,i}\rvert] = 0}$,
  \item $\smash{\lim_{D \to \infty} \sup_{\ConcentrationSet_{T,D}} \E[\lvert S_{t,D}^{m,i}\rvert] = 0}$,
  \item $\smash{\lim_{D \to \infty} \sup_{\ConcentrationSet_{T,D}} \E[\lvert T_{t,t+1,D}^{m,n,j}\rvert] = 0}$. \tendmark
 \end{enumerate}
\end{lemma}
\begin{proof}
 By definition of $\ConcentrationSet_{T,D}$, using that $x \mapsto \sqrt{x}$ is concave and increasing so that $\smash{\lvert \sqrt{x}-\sqrt{y}\rvert \leq \sqrt{\lvert x - y\rvert}}$, and since by Jensen's inequality, $\E[\lvert X \rvert] \leq \E[X^2]^{1/2}$, 
  \begin{align}
   \E[\lvert R_{t,D}^{n,1}\rvert] 
   &\leq \textstyle \sqrt{\tfrac{\lvert \calV_{t|T}(\state_{t-1:t,D}) - \calV_{t|T}\rvert }{\calV_{t|T}(\state_{t-1:t,D})}}\\
   & \qquad \times \textstyle \sqrt{\tfrac{\ell_t}{D}\sum_{d=1}^D \{[\partial_t \logWeightRandomWalkSingle_t](x_{t-1:t,d,D})\}^2 \E[(U_{t,d,D}^n)^2]}\\
   & \leq D^{-\eta/2} \ell_l^{1/2} \to 0,\\
   \E[\lvert S_{t,D}^{m,1}\rvert] 
   &\leq \textstyle \sqrt{\tfrac{\lvert \calW_{t|T}(\state_{t:t+1,D}) - \calV_{t|T}\rvert }{\calW_{t|T}(\state_{t:t+1,D})}}\\
   & \qquad \times \textstyle \sqrt{\tfrac{\ell_t}{D}\sum_{d=1}^D \{[\partial_t \logWeightRandomWalkSingle_{t+1}](x_{t:t+1,d,D})\}^2 \E[(U_{t,d,D}^m)^2]}\\
   & \leq D^{-\eta/2} \ell_l^{1/2} \to 0.
 \end{align}
 From the definition of $\ConcentrationSet_{T,D}$,
 \begin{align}
   \E[\lvert R_{t,D}^{n,2}\rvert] &\leq \tfrac{\ell_t}{2} D^{-\eta} \to 0,\\
   \E[\lvert S_{t,D}^{m,2}\rvert] &\leq \tfrac{\ell_t}{2} D^{-\eta} \to 0.
 \end{align}
 By Jensen's inequality, $\E[\lvert X \rvert] \leq \E[X^2]^{1/2}$ and hence
 \begin{align}
  \E[\lvert R_{t,D}^{n,3}\rvert] 
   & \textstyle \leq \sqrt{\tfrac{\ell_t^2}{4D^2} \sum_{d=1}^D \{[\partial_t^2 \logWeightRandomWalkSingle_t](x_{t-1:t,d,D})\}^2\var[(U_{t,d,D}^n)^2]}\\ 
   & \leq \tfrac{\ell_t}{2 D^{1/2}} \lVert \partial_t^2 \logWeightRandomWalkSingle_t \rVert_\infty \to 0,\\
  \E[\lvert S_{t,D}^{m,3}\rvert]
   & \textstyle \leq \sqrt{\tfrac{\ell_t^2}{4D^2} \sum_{d=1}^D \{[\partial_t^2 \logWeightRandomWalkSingle_{t+1}](x_{t:t+1,d,D})\}^2\var[(U_{t,d,D}^m)^2]}\\
   & \leq \tfrac{\ell_t}{2 D^{1/2}} \lVert \partial_t^2 \logWeightRandomWalkSingle_{t+1}\rVert_\infty \to 0,\\
  \!\!\!\!\!\E[\lvert T_{t,t+1,D}^{m,n,1}\rvert]
   & \textstyle \leq \sqrt{\tfrac{\ell_t \ell_{t+1}}{D^2} \sum_{d=1}^D \{[\partial_t \partial_{t+1} \logWeightRandomWalkSingle_{t+1}](x_{t:t+1,d,D})\}^2\E[U_{t,d,D}^m U_{t+1,d,D}^n]} = 0.\!\!\!\!
 \end{align}
 Since $\eta_{t,d,D}^m, \xi_{t,d,D}^n \in [0, D^{-1/2}]$ and since $\partial_t^2 \logWeightRandomWalkSingle_t$, $\partial_t^2 \logWeightRandomWalkSingle_{t+1}$ and $\partial_t \partial_{t+1} \logWeightRandomWalkSingle_{t+1}$ are Lipschitz-continuous 
  \begin{align}
  \E[\lvert R_{t,D}^{n,4}\rvert] 
   & \textstyle \leq  \tfrac{\ell_t^{3/2}}{2\sqrt{D}}  [\partial_t^2 \logWeightRandomWalkSingle_t]_\lip \E[\lvert U_{t,d,D}^n\rvert^3] \to 0,\\
  \E[\lvert S_{t,D}^{m,4}\rvert] 
   & \textstyle \leq  \tfrac{\ell_t^{3/2}}{2\sqrt{D}}  [\partial_t^2 \logWeightRandomWalkSingle_{t+1}]_\lip \E[\lvert U_{t,d,D}^m\rvert^3] \to 0,\\
  \!\!\!\!\!\!\!\E[\lvert T_{t,t+1,D}^{m,n,2}\rvert] 
   & \textstyle \leq  \tfrac{\sqrt{\ell_t \ell_{t+1}}}{\sqrt{D}}  [\partial_t \partial_{t+1} \logWeightRandomWalkSingle_{t+1}]_\lip \E[\lvert (\ell_t^{1/2} U_{t,d,D}^m +  \ell_{t+1}^{1/2} U_{t+1,d,D}^n) U_{t,d,D}^m U_{t+1,d,D}^n\rvert] \to 0,\!\!\!\!\!\!\!
 \end{align}
 where $[f]_\lip$ denotes the Lipschitz constant of $f$ w.r.t.\ the 1-norm. 
\end{proof}

\begin{lemma}\label{lem:clt}
 Assume~\ref{as:iid_model} and \ref{as:moments_bounded} and let $(V_t^n)_{t \in [T], n\in [N]}$ and $(W_t^n)_{t \in [T-1], n\in [N]}$ be the families of Gaussian random variables defined in Section~\ref{subsec:non-degnerate_limiting_genealogies}. Then, if $\state_{1:T,D} \in \ConcentrationSet_{T,D}$, as $D \to \infty$,
 \begin{align}
  \smash{Y_D \coloneqq (V_{1,D}^{1:N}, \dotsc, V_{T,D}^{1:N}, W_{1,D}^{1:N}, \dotsc, W_{T-1,D}^{1:N})^\T},
 \end{align}
 converges in distribution to
 \[
  \smash{Y \coloneqq (V_{1}^{1:N}, \dotsc, V_{T}^{1:N}, W_{1}^{1:N}, \dotsc, W_{T-1}^{1:N})^\T}. \mendmark
 \]
\end{lemma}
\begin{proof}
 Since $\E[V_{t,D}^n] = \E[V_t^n]$ and $\E[W_{t,D}^m] = \E[W_t^m]$, we only need to show convergence in distribution of the centred random vector
 \begin{align}
  \widetilde{Y}_D \coloneqq Y_D - \E[Y_D] = (\widetilde{V}_{1,D}^{1:N}, \dotsc, \widetilde{V}_{T,D}^{1:N}, \widetilde{W}_{1,D}^{1:N}, \dotsc, \widetilde{W}_{T-1,D}^{1:N})^\T
 \end{align}
 to
 \begin{align}
  \widetilde{Y} \coloneqq Y - \E[Y] = (\widetilde{V}_{1}^{1:N}, \dotsc, \widetilde{V}_{T}^{1:N}, \widetilde{W}_{1}^{1:N}, \dotsc, \widetilde{W}_{T-1}^{1:N})^\T.
 \end{align}
 By the Cram\'er--Wold theorem, it then suffices to prove that $\smash{\lambda^\T \widetilde{Y}_D \convergesInDistribution \lambda^\T \widetilde{Y}}$ for any $\smash{\lambda = (\lambda_1^{1:N}, \dotsc, \lambda_T^{1:N}, \bar{\lambda}_1^{1:N}, \dotsc, \bar{\lambda}_{T-1}^{1:N}) \in \reals^{N(2T-1)}}$. Equivalently, we must show that $\smash{\lambda^\T \widetilde{Y}_D \convergesInDistribution \dN(0, \tau^2)}$, where 
 \begin{align}
  \tau^2 
  & \coloneqq \var[\lambda^\T \widetilde{Y}]\\
  & = \sum_{t=1}^T \sum_{n=1}^N \sum_{m=1}^N \lambda_t^n \lambda_t^m \cov[V_t^n, V_t^m] + \bar{\lambda}_t^n \bar{\lambda}_t^m \cov[W_t^n, W_t^m] + 2 \lambda_t^n \bar{\lambda}_t^m \cov[V_t^n, W_t^m],
 \end{align}
  where we recall the convention that $W_T^n \equiv 0$.
  
 Let $\calF_{d,D} \coloneqq \sigma(\{U_{t,e,D}^n \mid t \in [T], e \in [d], n \in [N]\})$ as well as
 \begin{align}
  \!\!\calU_{d,D} & \coloneqq  \sum_{t=1}^T \biggl[ \sqrt{\frac{\calV_{t|T}}{\calV_{t|T}(\state_{t-1:t,D})}} \sqrt{\tfrac{\ell_t}{D}} [\partial_t \logWeightRandomWalkSingle_t](x_{t-1:t,d,D}) \sum_{n=1}^N \lambda_t^n U_{t,d,D}^n\\
  & \qquad \! + \ind\{t < T\} \sqrt{\frac{\calW_{t|T}}{\calW_{t|T}(\state_{t:{t+1},D})}}\sqrt{\tfrac{\ell_t}{D}} [\partial_t \logWeightRandomWalkSingle_{t+1}](x_{t:t+1,d,D}) \sum_{m=1}^N \bar{\lambda}_t^m U_{t,d,D}^m \biggr],\!\!
 \end{align}
 then $\calU_{d,D}$ is $\calF_{d,D}$-measurable. Therefore, for $\state_{1:T,D} \in \ConcentrationSet_{T,D}$, by the central limit theorem for triangular arrays \citep{dvoretzky1972asymptotic},  $\smash{\sum_{d=1}^D \calU_{d,D} = \lambda^\T \widetilde{Y}_D}$ converges in distribution to a zero-mean Gaussian random variable with variance $\tau^2$ if, as $D \to \infty$,
 \begin{enumerate}
  \item $\smash{\sum_{d=1}^D \E[\calU_{d,D}^2 |\calF_{d-1,D}] - \E[\calU_{d,D}|\calF_{d-1,D}]^2 \convergesInProbability \tau^2}$,
  \item $\smash{\sum_{d=1}^D \sum_{d=1}^D \E[\calU_{d,D}^2 \ind\{\lvert \calU_{d,D} \rvert \geq \varepsilon \} |\calF_{d-1,D}] \to 0}$, for any $\varepsilon > 0$.
 \end{enumerate}
 To verify the first condition, we note that
 \begin{align}
  \MoveEqLeft \sum_{d=1}^D \E[\calU_{d,D}^2 |\calF_{d-1,D}] - \E[\calU_{d,D}|\calF_{d-1,D}]^2\\
  & = \tau^2 + 2 \sum_{t=1}^{T-1} \ell_t H_{t|T,D} \sum_{n=1}^N \sum_{m=1}^N \lambda_t^n \bar{\lambda}_t^m \varSigma_{n,m},
 \end{align}
 where
 \begin{align}
  H_{t|T,D} & \coloneqq
  \sqrt{\frac{\calV_{t|T}\calW_{t|T}}{\calV_{t|T}(\state_{t-1:t,D})\calW_{t|T}(\state_{t:t+1,D})}} \calS_{t|T}(\state_{t-1:t+1,D}) - \calS_{t|T},
 \end{align}
 so that $\smash{\lim_{D \to \infty} \sup_{\ConcentrationSet_{T,D}} \lvert H_{t|T,D} \rvert = 0}$ by definition of $\ConcentrationSet_{T,D}$.

 It remains to check the second condition. Let $\varepsilon > 0$ and set
 \begin{align}
  a^{(1)} & \coloneqq \textstyle \sup_{t \in [T]} \ell_t^2 \calV_{t|T} < \infty,\\
  a^{(2)} & \coloneqq \textstyle \sup_{t \in [T-1]} \ell_t^2 \calW_{t|T} < \infty,\\
  a^{(3)} & \coloneqq \textstyle 2 \sup_{t \in [T-1]} \ell_t \ell_{t+1} \sqrt{\calV_{t|T} \calW_{t|T}} < \infty,
 \end{align}
 and $a \coloneqq \sup_{i \in [3]} a^{(i)}$,
 \begin{align}
   b_{t,d,D}^{(i)} & \coloneqq \textstyle [\widetilde{\calV}_{t|T}(\state_{t-1:t,D})]^2,\\
   b_{t,d,D}^{(2)} & \coloneqq \textstyle [\widetilde{\calW}_{t|T}(\state_{t:t+1,D})]^2,\\
   b_{t,d,D}^{(3)} & \coloneqq \textstyle \widetilde{\calV}_{t|T}(\state_{t-1:t,D})\widetilde{\calW}_{t|T}(\state_{t:t+1,D}),
 \end{align}
 as well as
 \begin{align}
  \calM_{t,d}^{(1)} & \coloneqq \textstyle \sum_{n=1}^N \sum_{m=1}^N \lambda_t^n \lambda_t^m U_{t,d,D}^n U_{t,d,D}^m,\\
  \calM_{t,d}^{(2)} & \coloneqq \textstyle \ind\{t < T\} \sum_{n=1}^N \sum_{m=1}^N \bar{\lambda}_t^n \bar{\lambda}_t^m U_{t,d,D}^n U_{t,d,D}^m,\\
  \calM_{t,d}^{(3)} & \coloneqq \textstyle \ind\{t < T\} \sum_{n=1}^N \sum_{m=1}^N \lambda_t^n \bar{\lambda}_t^m U_{t,d,D}^n U_{t,d,D}^m.
 \end{align}
 Then, since $b_{t,d,D}^{(i)} \leq D^{-2\eta}$ for all $i \in [3]$ by definition of $\ConcentrationSet_{T,D}$, 
 \begin{align}
  \{\lvert \calU_{d,D} \rvert \geq \varepsilon \}
  & \subseteq \textstyle  \{ a \sum_{i=1}^3  \sum_{t=1}^T  b_{t,d,D}^{(i)}  \lvert \calM_{t,d}^{(i)}\rvert \geq \varepsilon^2 \}\\
  & \subseteq \textstyle  \{ a \sum_{i=1}^3  \sum_{t=1}^T \lvert \calM_{t,d}^{(i)} \rvert \geq D^{2\eta}\varepsilon^2 \},
 \end{align}
 and thus
 \begin{align}
  \MoveEqLeft \sum_{d=1}^D \E[\calU_{d,D}^2 \ind\{\lvert \calU_{d,D} \rvert \geq \varepsilon \} |\calF_{d-1,D}]\\
  & \leq \sum_{d=1}^D \E\biggl[   a \sum_{i=1}^3 \sum_{t=1}^T \lvert \calM_{t,d}^{(i)}\rvert\ind\biggl\{a \sum_{j=1}^3 \sum_{s=1}^T \lvert \calM_{s,d}^{(j)} \rvert \geq D^{2\eta} \varepsilon^2 \biggr\} \biggr] \to 0.
 \end{align}
 This completes the proof. 
\end{proof}

\begin{namedproof}[of Proposition~\ref{prop:limiting_rwcsmc_algorithm}]
 We present the proof for the general case \emph{with} backward sampling. This immediately implies the proof for the case \emph{without} backward sampling. Likewise, we omit the proof in the case of the forced-move extension.
 
 We fix some $N, T \in \naturals$ and define $\ConcentrationSet_{T,D}$ as in \eqref{eq:concentration_set_rwcsmc}. By Lemma~\ref{lem:concentration_set_rwcsmc}, we then have $\smash{\lim_{D \to \infty} \Target_{T,D}(\ConcentrationSet_{T,D}) = 1}$. We also fix some $a_{t}^n \in [N]_0$ for all $(t,n) \in [T-1] \times [N]$ and some $k_t \in [N]_0$ for all $t \in [T]$. 

 The proof of the statement is then complete upon verifying that, as $D \to \infty$,
 \begin{align}
  \MoveEqLeft \textstyle \sup_{\ConcentrationSet_{\nTimeSteps, \nDimensions}} \lvert \E[\varUpsilon_D(\Particle_D)] - \E[\varUpsilon(Y)]\rvert\\
  & \leq  \textstyle\sup_{\ConcentrationSet_{\nTimeSteps, \nDimensions}} \lvert \E[\varUpsilon_D(\Particle_D)] - \E[\varUpsilon(Y_D)]\rvert +  \textstyle\sup_{\ConcentrationSet_{\nTimeSteps, \nDimensions}} \lvert \E[\varUpsilon(Y_D)] - \E[\varUpsilon(Y)]\rvert \label{eq:main_decomposition}\\
  & \to 0,
 \end{align}
 where $\Particle_D \coloneqq \Particle_{1,D}^{1:N}, \dotsc, \Particle_{T,D}^{1:N}$ and
 \begin{align}
  \!\!\!\!\!\!\!\!\!\! \MoveEqLeft \varUpsilon_D(\Particle_D)\\
  & \coloneqq \textstyle \prod_{t=1}^{T-1} \prod_{n=1}^N \selectionFunctionBoltzmann{a_t^\particleIndex}(\{\logWeightRandomWalk_t(\Particle_{t-1,D}^{a_{t-1}^\particleIndexAlt}, \Particle_{t,D}^\particleIndexAlt) - \logWeightRandomWalk_t(\Particle_{t-1,D}^0, \Particle_{t,D}^0)\}_{\particleIndexAlt = 1}^\nParticles)\\
  & \quad \times \textstyle\selectionFunctionBoltzmann{k_T}(\{\logWeightRandomWalk_T(\Particle_{T-1,D}^{a_{T-1}^\particleIndexAlt}, \Particle_{T-1,D}^\particleIndexAlt) - \logWeightRandomWalk_T(\Particle_{T-1,D}^0, \Particle_{T,D}^0)\}_{\particleIndexAlt = 1}^\nParticles)\\
  & \quad \times \textstyle\prod_{t=1}^{T-1} \selectionFunctionBoltzmann{\outputParticleIndex_t}(\{
  \logBackwardWeightRandomWalk_t(\Particle_{t-1,D}^{\mathrlap{a_{t-1}^\particleIndexAlt}}, \Particle_{t,D}^\particleIndexAlt, \Particle_{t+1,D}^{\outputParticleIndex_{t+1}}) - \logBackwardWeightRandomWalk_t(\Particle_{t-1,D}^0, \Particle_{t,D}^0, \Particle_{t+1,D}^{\outputParticleIndex_{t+1}})
  \}_{\particleIndexAlt = 1}^\nParticles), \!\!\!\!\!\!\!\!\!\!
 \end{align}
 and where $Y$ and $Y_D$ are as in Lemma~\ref{lem:clt} as well as
  \begin{align}
  \MoveEqLeft\varUpsilon((v_1^{1:N}, \dotsc, v_T^{1:N}, w_1^{1:N}, \dotsc, w_{T-1}^{1:N})^\T)\\
  & \coloneqq \textstyle\prod_{t=1}^{T-1} \prod_{n=1}^N \selectionFunctionBoltzmann{a_t^\particleIndex}(\{v_t^m + w_{t-1}^{a_{t-1}^m}\}_{\particleIndexAlt = 1}^\nParticles)\\
  & \quad \textstyle\times \selectionFunctionBoltzmann{k_T}(\{v_T^m + w_{T-1}^{a_{T-1}^m}\}_{\particleIndexAlt = 1}^\nParticles) \prod_{t=1}^{T-1} \selectionFunctionBoltzmann{\outputParticleIndex_t}(\{v_t^m + w_t^m + w_{t-1}^{a_{t-1}^m}\}_{\particleIndexAlt = 1}^\nParticles).
 \end{align} 

 We now consider the two terms on the r.h.s.\ of \eqref{eq:main_decomposition}. For the first term, a standard telescoping-sum decomposition, and using the fact that the selection functions are Lipschitz (see Lemma~\ref{lem:selection_functions_lipschitz}) and bounded above by $1$, gives
 \begin{align}
 \!\!\!\!\!\!\!\!\!\!  \MoveEqLeft \E[\varUpsilon_D(\Particle_D)] - \E[\varUpsilon(Y_D)]\rvert\\
  & \leq \textstyle \bigl[\sup_{m \in [N_0]}[\selectionFunctionBoltzmann{m}]_\lip\bigr]\\
  & \quad \times \textstyle \Bigl[ N \sum_{t=1}^{T-1} \sum_{n=1}^N [\{ \sum_{i=1}^4 \lvert R_{t,D}^{n,i} \rvert + \lvert S_{t-1,D}^{a_{t-1}^n,i}\rvert\} + \sum_{j=1}^2 \lvert T_{t-1,t,D}^{a_{t-1}^n, n,j} \rvert]\\
  & \qquad\textstyle + \sum_{n=1}^N [\{ \sum_{i=1}^4 \lvert R_{T,D}^{n,i} \rvert + \lvert S_{T-1,D}^{a_{T-1}^n,i}\rvert\} + \sum_{j=1}^2 \lvert T_{T-1,T,D}^{a_{t-1}^n, n,j} \rvert]\\
  & \qquad \textstyle + \sum_{t=1}^{T-1} \sum_{n=1}^N [\{ \sum_{i=1}^4 \lvert R_{t,D}^{n,i} \rvert + \lvert S_{t,D}^{n,i} \rvert + \lvert S_{t-1,D}^{a_{t-1}^n,i}\rvert\} + \sum_{j=1}^2 \lvert T_{t-1,t,D}^{a_{t-1}^n, n,j} \rvert]\Bigr],\!\!\!\!\!\!\!\!\!\!
 \end{align}
 so that $\smash{\lim_{D \to \infty} \sup_{\ConcentrationSet_{\nTimeSteps, \nDimensions}} \lvert \E[\varUpsilon_D(\Particle_D)] - \E[\varUpsilon(Y_D)]\rvert \to 0}$, by Lemma~\ref{lem:remainder_terms}.
 
 For the second term on the r.h.s.\ of \eqref{eq:main_decomposition}, Lemma~\ref{lem:clt} and the continuous mapping theorem ensure that $\varUpsilon(Y_D) \convergesInDistribution \varUpsilon(Y)$. Since $0 \leq \varUpsilon \leq 1$, this implies $\smash{\lim_{D \to \infty} \sup_{\ConcentrationSet_{\nTimeSteps, \nDimensions}} \lvert \E[\varUpsilon(Y_D)] - \E[\varUpsilon(Y)]\rvert \to 0}$. \qedwhite
\end{namedproof}

\subsection{Proof of Proposition~\ref{prop:stability_of_acceptance_rates}}
\label{app:subsec:prop:stability_of_acceptance_rates_rwcsmc}

In this section, we prove Proposition~\ref{prop:stability_of_acceptance_rates}. The proof relies on a few lemmata which we state first.

\begin{lemma}\label{lem:convex_ordering_of_lognormals}
 Let $\sigma_2 \geq \sigma_1 > 0$ and $(X_1, X_2) \sim \dN(0_2,\iMat_2)$. Then
 \[
  \smash{\eul^{\sigma_1 X_1 - \sigma_1^2/2} \leq_{\mathrm{cx}} \eul^{\sigma_2 X_2 - \sigma_2^2/2}.} \mendmark
 \]
\end{lemma}
\begin{proof}
 Let $\standardNormalPdf\colon \reals \to (0,\infty)$ denote the probability density function of a standard normal distribution, let $\mathrm{\standardNormalCdf}\colon \reals \to  (0,1)$ denote the associated cumulative distribution function and $X \sim \dN(0,1)$. Then for any $a>0$, $b \in \reals$ and $c \in \reals$, writing $l(a,c) \coloneqq \log(c)/a + a/2$,
 \begin{align}
  \E[(\eul^{a X + b}-c)_+]
  & = \int_{l(a,c)}^\infty (\eul^{ax + b}-c) \standardNormalPdf(x) \intDiff x\\
  & = \eul^{a^2/2 + b} \int_{l(a,c)}^\infty \standardNormalPdf(x-a) \intDiff x - c \standardNormalCdf(-l(a,c))\\
  & = \eul^{a^2/2 + b} \standardNormalCdf(-l(a,c)+a) \intDiff x - c \standardNormalCdf(-l(a,c)). \label{eq:lem:convex_ordering_of_lognormals}
 \end{align}
 Let $Y_i \coloneqq  \eul^{\sigma_i X_i - \sigma_i^2/2}$ for $i \in \{1,2\}$. Then by \eqref{eq:lem:convex_ordering_of_lognormals}, for any $d \in \reals$,
 \begin{align}
  \E[(Y_2 - d)_+] - \E[(Y_1 - d)_+]
  & = \standardNormalCdf\Bigl(\frac{\sigma_2}{2} - \frac{\log(d)}{\sigma_2}\Bigr) - \standardNormalCdf\Bigl(\frac{\sigma_1}{2} - \frac{\log(d)}{\sigma_1}\Bigr)  \geq 0.
 \end{align}
 Finally, $\E[Y_1] = \E[Y_2]$ by the properties of the log-normal distribution. This completes the proof. 
\end{proof}

\begin{lemma}\label{lem:lower_bound_on_expectation_of_selection_function}
 Let $X \coloneqq (X_1, \dotsc, X_N) \sim \dN(\mu, \sigma^2 \varSigma)$ and $Y \coloneqq (Y_1, \dotsc, Y_N) \sim \dN(\nu, \tau^2 \varSigma)$, where $[\varSigma]_{i,i} = 1$ and $[\varSigma]_{i,j} = 1/2$ for $i \neq j$ and where $\mu \coloneqq -a\bm{1}_N$, $\nu \coloneqq -b\bm{1}_N$ for $a \in \reals$ and $b \geq \tau^2/2$. Assume also that $X$ and $Y$ are independent. Then for any binary vector $\delta \coloneqq \delta_{1:N} \in \{0,1\}^N$,
 \begin{align}
  \E\biggl[\frac{\sum_{i=1}^N \eul^{X_i + \delta_i Y_i}}{1 + \sum_{i=1}^N \eul^{X_i + \delta_i Y_i}}\biggr]
  & \geq \biggl(1 + \frac{\eul^{\sigma^2/2 + a + \tau^2/2 + b}}{N}\biggr)^{\mathrlap{-1}}.
 \end{align}
 In particular, for $\delta = (0,\dotsc,0)$ we have the tighter bound
 \[
  \E\biggl[\frac{\sum_{i=1}^N \eul^{X_i}}{1 + \sum_{i=1}^N \eul^{X_i}}\biggr]
   \geq \biggl(1 + \frac{\eul^{\sigma^2/2 + a}}{N}\biggr)^{\mathrlap{-1}}. \mendmark
 \]
\end{lemma}
\begin{proof}
  We begin by proving the bound in the special case $\delta = (0,\dotsc,0)$:
  \begin{align}
    \E\biggl[\frac{\sum_{i=1}^N \eul^{X_i}}{1 + \sum_{i=1}^N \eul^{X_i}}\biggr]
    & \geq \E\biggl[\frac{1}{1 + \eul^{-X_1}/N}\biggr]\\
    & \geq \frac{1}{1 + \E[\eul^{-X_1}/N]}
    = \biggl(1 + \frac{\eul^{\sigma^2/2 + a}}{N}\biggr)^{\mathrlap{-1}},
  \end{align}
  where the first line follows since $t \mapsto t/(1 + t)$ is concave and $\sum_{i=1}^N \eul^{X_i} \leq_{\mathrm{cx}} N\eul^{X_1}$; the second step is due to Jensen's inequality and the fact that $t \mapsto 1/(1 + t)$ is convex; the last step follows from the properties of the log-normal distribution.
  
  We now extend the approach to arbitrary $\delta \in \{0,1\}^N$. Since $\tau^2/2 - b \leq 0$,
  \begin{align}
   \frac{\sum_{i=1}^N \eul^{X_i + \delta_iY_i}}{1 + \sum_{i=1}^N \eul^{X_i + \delta_i Y_i}}
   & \geq \frac{\sum_{i=1}^N \eul^{X_i + \delta_i Y_i + (1-\delta_i)(\tau^2/2 - b)}}{1 + \sum_{i=1}^N \eul^{X_i + \delta_i Y_i + (1-\delta_i)(\tau^2/2 - b)}}.
  \end{align}
  Furthermore, by Lemma~\ref{lem:convex_ordering_of_lognormals} and \citet[][Theorem~5]{dhaene2000comonotonicity},
  \begin{align}
    \sum_{i=1}^N \eul^{X_i + \delta_i Y_i + (1-\delta_i)(\tau^2/2 - b)} 
    \leq_{\mathrm{cx}} N \eul^{X_1 + Y_1},
  \end{align}
  and since $t \mapsto t/(1 + t)$ is concave,
   \begin{align}
   \E\biggl[\frac{\sum_{i=1}^N \eul^{X_i + \delta_iY_i}}{1 + \sum_{i=1}^N \eul^{X_i + \delta_i Y_i}}\biggr]
   & \geq \E\biggl[\frac{\sum_{i=1}^N \eul^{X_i + \delta_i Y_i + (1-\delta_i)(\tau^2/2 - b)}}{1 + \sum_{i=1}^N \eul^{X_i + \delta_i Y_i + (1-\delta_i)(\tau^2/2 - b)}}\biggr]\\
   & \geq \E\biggl[\frac{\sum_{i=1}^N \eul^{X_i + Y_i}}{1 + \sum_{i=1}^N \eul^{X_i + Y_i}}\biggr]\\
   & \geq \E\biggl[\frac{N \eul^{X_1 + Y_1}}{1 + N \eul^{X_1 + Y_1}}\biggr]\\
   & \geq \frac{1}{1 + \E[\eul^{-X_1 - Y_1}]/N}\\
   & = \biggl(1 + \frac{\eul^{\sigma^2/2 + a + \tau^2/2 + b}}{N}\biggr)^{\mathrlap{-1}}.
  \end{align}
  Here, the penultimate line again follows by Jensen's inequality since $t \mapsto 1/(1 + t)$ is convex. \qedwhite 
\end{proof}

\begin{lemma}\label{lem:integration_by_parts_identity}
 Let $\pi(x_{1:T})$ denote some twice differentiable probability density function on $\reals^T$ and write $\smash{\partial_t^i f(x_{1:T})}$ as shorthand for $\smash{\frac{\partial^i}{\partial x_t^i}f(x_{1:T})}$ with $\smash{\partial_t^1 \eqqcolon \partial_t}$. Then
 \[
  \smash{\pi([\partial_t \log \pi]^2) = - \pi(\partial_t^2 \log \pi).} \mendmark
 \]
\end{lemma}
\begin{proof}
 If $T=1$, using integration by parts, 
  \begin{align}
  \pi([(\log \pi)']^2) 
   & = \int \pi'(x) (\log \pi)'(x) \intDiff x\\
   & = \pi'(x)|_{-\infty}^\infty - \pi((\log \pi)'') = - \pi((\log \pi)''). \label{eq:lem:integration_by_parts_identity}
 \end{align}
 For $T > 1$, we let $\pi_{-t}(x_{-t}) \coloneqq \int_{-\infty}^\infty \pi(x_{1:T}) \intDiff x_t$ denote the marginal density of  $x_{-t} \coloneqq (x_{1:t-1}, x_{t+1:T})$ and let $\pi_{t|-t}(x_t|x_{-t}) \coloneqq \pi(x_{1:T}) /\pi_{-t}(x_{-t})$. Since $\partial_t \log \pi(x_{1:T}) = \partial_t \log \pi_{t|-t}(x_t|x_{-t})$,
 \begin{align}
   \pi([\partial_t \log \pi]^2)
   & = \pi([\partial_t \log \pi_{t|-t}]^2)\\
   & = \int \biggl[\int [\partial_t \log \pi_{t|-t}(x_t|x_{-t})]^2 \pi_{t|-t}(x_t|x_{-t}) \intDiff x_t \biggr] \pi_{-t}(x_{-t}) \intDiff x_{-t}\\
   & = - \pi(\partial_t^2 \log \pi_{t|-t})\\
   & = - \pi(\partial_t^2 \log \pi),
 \end{align}
 where the third line follows by integration by parts in the same way as \eqref{eq:lem:integration_by_parts_identity} with $\pi(x)$ replaced by the conditional distribution $\pi_{t|-t}(x_t|x_{-t})$. 
\end{proof}

\begin{namedproof}[of Proposition~\ref{prop:stability_of_acceptance_rates}]
  By Lemma~\ref{lem:peskun}, it suffices to consider the case without forced move. 
  Under Assumption~\ref{as:independent_over_time}, we have $W_t^n \equiv 0$ for any $t \in [T]$ and any $n \in [N]$, so that 
  \begin{align}
  \ResamplingKernelRandomWalk{\timeIndex|\nTimeSteps}{\nParticles}((v_t, w_{t-1}, a_{\timeIndex-1}), \{\particleIndex\})
  & =
  \selectionFunctionBoltzmann{\particleIndex}(\{ 
   v_t^m
  \}_{\particleIndexAlt = 1}^\nParticles),
\end{align}
 does not depend on $a_{t-1}$ (nor on $w_{t-1}$). As a consequence,
\begin{align}
 \acceptanceRateRwCsmc{\nTimeSteps}{\nParticles}(t)
  & = \prod_{s=t}^T \sum_{n \in [N]}\ExpectationCsmcKernelRandomWalk{\nTimeSteps}{\nParticles}[\selectionFunctionBoltzmann{\particleIndex}(\{ 
   V_s^m
  \}_{\particleIndexAlt = 1}^\nParticles)]
 \geq \prod_{s = t}^T \biggl(1 + \frac{\exp(\ell_s \calI_{s|T})}{N}\biggr)^{\mathrlap{-1}},
 \end{align}
 where the last line follows by Lemma~\ref{lem:lower_bound_on_expectation_of_selection_function}. This completes the proof of the first part of the proposition.

We now prove the lower bound in the case that backward sampling is employed. Since $W_t^n \equiv 0$ for any $t \in [T]$ and any $n \in [N]$ due to Assumption~\ref{as:independent_over_time}, we additionally have that
  \begin{align}
  \BackwardKernelRandomWalk{\timeIndex|\nTimeSteps}{\nParticles}((v_t, w_{t:t-1}, a_{\timeIndex-1}), \{\particleIndex\})
  & =
  \selectionFunctionBoltzmann{\particleIndex}(\{ 
   v_t^m
  \}_{\particleIndexAlt = 1}^\nParticles),
\end{align}
 does not depend on $a_{t-1}$ (nor on $w_{t:t-1}$). As a consequence, 
\begin{align}
 \acceptanceRateRwCsmc{\nTimeSteps}{\nParticles}(t) 
  & = \sum_{n \in [N]}\ExpectationCsmcKernelRandomWalk{\nTimeSteps}{\nParticles}[\selectionFunctionBoltzmann{\particleIndex}(\{ 
   V_t^l
  \}_{l = 1}^\nParticles)]
  \geq \biggl(1 + \frac{\exp(\ell_t \calI_{t|T})}{N}\biggr)^{\mathrlap{-1}},
 \end{align}
 where the last line follows by Lemma~\ref{lem:lower_bound_on_expectation_of_selection_function}. This completes the proof of the second part of the proposition. \qedwhite
\end{namedproof}

\section{Additional simulation results}
\label{app:sec:additional_simulation_results}

\glsreset{ESS}


\textbf{Figure~\ref{fig:ess}} displays the $\Target_{T,D}$-averaged \emph{\gls{ESS}} of the `resampling' and 'backward-sampling' weights at time $t$ for each algorithm. More specifically, let
\begin{align}
 \mathit{ESS}(W^{0:N}) \coloneqq \frac{1}{\sum_{n=0}^N (W^n)^2},
\end{align}
denote the \gls{ESS} for self-normalised importance sampling weights $W_t^{0:N}$ \citep{kong1994sequential}. Below, let $\E$ denote expectation w.r.t.\ $\smash{\State_{1:T} \sim \Target_{T,D}}$.

\begin{enumerate}
 \item The first column shows $\smash{\E\{\ExpectationCsmcKernel{T,D,\State_{1:T}}{N}[\mathit{ESS}(W_t^{0:N})]\}}$,
where,
\begin{align}
 \!\!\!\!\!\!\!\!\!\!\!\!\!\!\!\!\!\!\!\!W_t^n \coloneqq 
 \begin{cases}
  \selectionFunctionBoltzmann{n}(\{\logWeight_t(\Particle_t^\particleIndexAlt) - \logWeight_t(\Particle_t^0)\}_{\particleIndexAlt = 1}^\nParticles), & \text{without backward sampling,}\!\!\!\!\\
  \selectionFunctionBoltzmann{n}(\{
  \logBackwardWeight_\timeIndex(\Particle_{\timeIndex}^{\mathrlap{m}}, \Particle_{\timeIndex+1}^{K_{\timeIndex+1}}) - \logBackwardWeight_\timeIndex(\Particle_{\timeIndex}^0, \Particle_{\timeIndex+1}^{K_{\timeIndex+1}})
  \}_{\particleIndexAlt = 1}^\nParticles), & \text{with backward sampling.}\\
 \end{cases}
\end{align}
\item The second column shows $\smash{\E\{\ExpectationCsmcKernelRandomWalk{T,D,\State_{1:T}}{N}[\mathit{ESS}(\widebar{W}_t^{0:N})]\}}$,
where
\begin{align}
  \!\!\!\!\!\!\!\!\!\!\!\!\!\!\!\!\!\!\!\!\widebar{W}_t^n \coloneqq 
 \begin{cases}
  \selectionFunctionBoltzmann{n}(\{\logWeightRandomWalk_t(\Particle_{\timeIndex-1}^{\mathrlap{A_{\timeIndex-1}^\particleIndexAlt}\,}, \Particle_t^m) - \logWeightRandomWalk_\timeIndex(\Particle_{\timeIndex-1}^0, \Particle_t^0)\}_{\particleIndexAlt = 1}^\nParticles), & \text{without backward sampling,}\!\!\!\!\!\!\!\!\!\!\!\!\!\!\!\!\!\!\!\!\!\!\!\!\!\!\!\!\!\\
  \selectionFunctionBoltzmann{n}(\{
  \logBackwardWeightRandomWalk_\timeIndex(\Particle_{\timeIndex-1}^{\mathrlap{A_{\timeIndex-1}^\particleIndexAlt}\,}, \Particle_{\timeIndex}^{\mathrlap{m}}, \Particle_{\timeIndex+1}^{K_{\timeIndex+1}}) - \logBackwardWeightRandomWalk_\timeIndex(\Particle_{\timeIndex-1}^0, \Particle_{\timeIndex}^0, \Particle_{\timeIndex+1}^{K_{\timeIndex+1}})
  \}_{\particleIndexAlt = 1}^\nParticles), & \text{with backward sampling.}\!\!\!\!\!\!\!\!\!\!\!\!\!\!\!\!\!\!\!\!\!\!\!\!\!\!\!\!\!\\
 \end{cases}
\end{align}
\end{enumerate}
By construction, the \gls{ESS} takes values in $[1, N+1]$. The first column shows that for the \gls{ICSMC} algorithm, the \gls{ESS} degenerates to its smallest possible value, $1$, in high dimensions. In contrast, for the \gls{IRWCSMC} algorithm, the \gls{ESS} converges to a non-trivial limit $>1$.

\begin{figure}[H]
 \vspace{0cm}
 \noindent{}
 \centering  
 \includegraphics{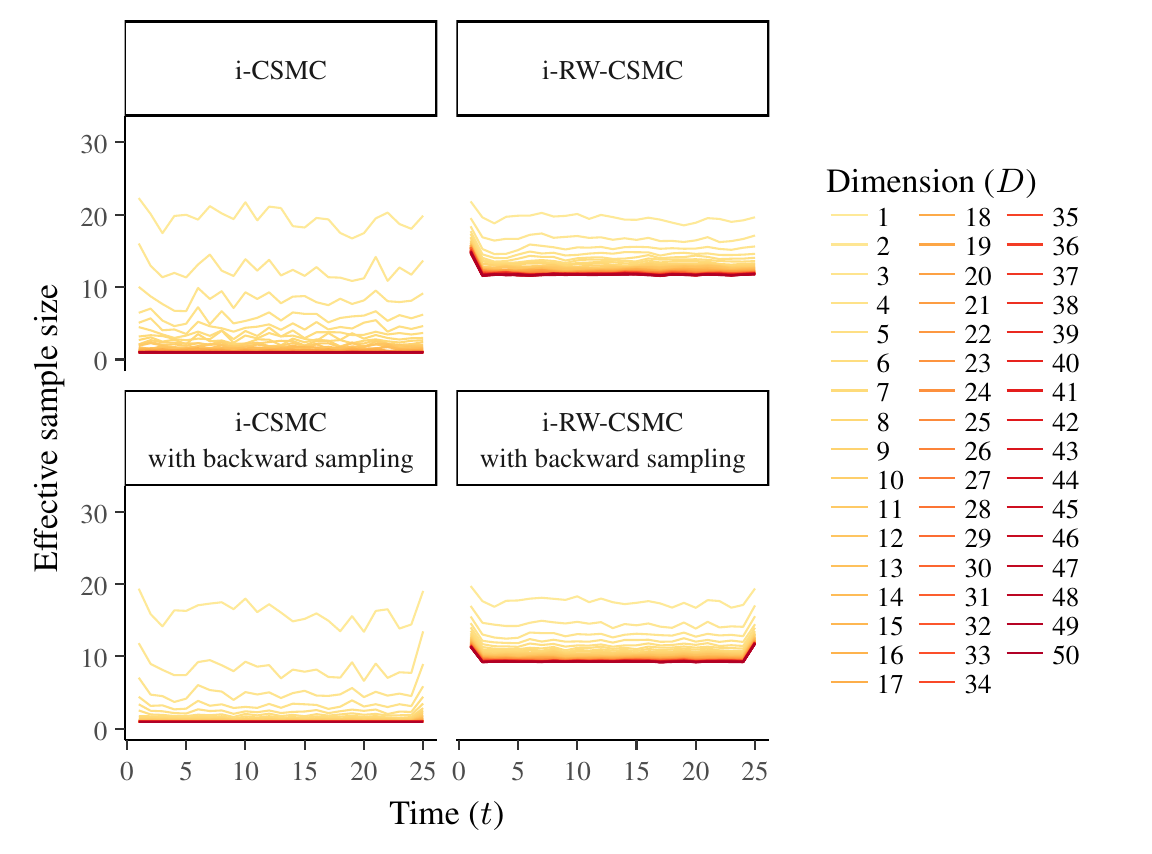}
 \caption{The $\Target_{T,D}$-averaged effective sample sizes of the `resampling weights' (top row) and `backward-sampling weights' (bottom row) as a function of $t$.}
 \label{fig:ess}
\end{figure}

\section{Use for parameter estimation}
\label{app:sec:parameter_estimation}

The Feynman--Kac model is typically specified through a set of parameters $\theta \in \mathrm{\Theta}$, i.e.\ $\Mutation_t = \Mutation_{\theta,t}$, $\mutation_t = \mutation_{\theta,t}$, $\Potential_t = \Potential_{\theta,t}$ and $\Target_{T,D} = \Target_{\theta, T,D}$. In this case, we likewise write \gls{MCMC} kernels induced by Algorithms~\ref{alg:iterated_csmc}, \ref{alg:iterated_rwehmm} and \ref{alg:iterated_rwcsmc}, i.e.\ $S(\ccdot|\ccdot;\theta)$ as $\smash{\InducedIteratedCsmcKernel{T,D}{N} = \InducedIteratedCsmcKernel{\theta, T,D}{N}}$, $\smash{\InducedIteratedEhmmKernelRandomWalk{T,D}{N}=\InducedIteratedEhmmKernelRandomWalk{\theta, T,D}{N}}$ as $\smash{\InducedIteratedCsmcKernelRandomWalk{T,D}{N}=\InducedIteratedCsmcKernelRandomWalk{\theta, T,D}{N}}$ 

If $\theta$ is unknown, then Bayesian inference in this model requires an \gls{MCMC} algorithm that targets the \emph{joint} posterior distribution of the parameters and the latent states which is proportional to $\varpi(\diff\theta \times \diff \state_{1:T}) \propto \mu(\diff \theta) \Target_{\theta,T,D}(\diff \state_{1:T})$, where the probability measure $\mu$ on $\Theta$ is the prior distribution for $\theta$. 

In this section, we discuss two classes pf \gls{MCMC} algorithms which target this joint posterior distribution. The first includes the particle Gibbs sampler from \citet{andrieu2010particle}; the second includes a novel algorithm.

\paragraph*{Particle Gibbs sampler} The first parameter-estimation algorithm is the \emph{particle Gibbs sampler} proposed in \citet{andrieu2010particle}. Its $(l+1)$th iteration is given in Algorithm~\ref{alg:particle_gibbs}, where $R_{\state_{1:T}}(\theta, \diff \vartheta)$ denotes some $\varpi(\diff\theta | \state_{1:T})$-invariant \gls{MCMC} kernel (e.g.\ often a convolution of multiple \gls{MH} updates). 

\noindent\parbox{\textwidth}{
\begin{flushleft}
 \begin{framedAlgorithm}[particle Gibbs sampler] \label{alg:particle_gibbs} Given $(\theta[l], \state_{1:\nTimeSteps}[l]) \in \mathrm{\Theta} \times \spaceState_{T,D}$,
 \begin{enumerate}
 \item \label{alg:particle_gibbs:1} sample $\theta[l+1]  \sim R(\ccdot|\theta[l]; \state_{1:T}[l])$,
 \item \label{alg:particle_gibbs:2} sample $\state_{1:T}[l+1] \sim \smash{\InducedIteratedCsmcKernel{\theta[l+1], T,D}{N}}(\state_{1:T}[l], \ccdot)$. 
 \end{enumerate}
\end{framedAlgorithm}
\end{flushleft}
}

In Step~\ref{alg:particle_gibbs:2} of the particle Gibbs sampler, it is straightforward to instead use the Markov kernel induced by Algorithm~\ref{alg:iterated_rwehmm} or \ref{alg:iterated_rwcsmc}, i.e.\ $\smash{\InducedIteratedEhmmKernelRandomWalk{\theta, T,D}{N}}$ or $\smash{\InducedIteratedCsmcKernelRandomWalk{\theta, T,D}{N}}$. To see this, note that these kernels, too, leave $\Target_{\theta,T,D}(\diff \state_{1:T}) = \varpi(\diff \state_{1:T}|\theta)$ invariant.


\paragraph*{Alternative algorithm} For the \gls{RWEHMM} and \gls{IRWCSMC} an alternative type of parameter-estimation algorithm, in which the $\theta$-updates make use of the information contained in \emph{all} particles $\smash{\Particle_t^n}$ is possible. 

The \gls{RWEHMM}-based algorithm is outlined in Algorithm~\ref{alg:rwehmm_based_parameter_estimation}, where $q_{\particle_{1:T}}(\theta, \diff \theta')$ is some proposal kernel for the parameters which may depend on the values of the particles. It can be viewed as a version of the parameter-estimation algorithms based around embedded \gls{HMM} methods proposed in \citet{shestopaloff2013mcmc} who argued that averaging over multiple particles may allow for larger steps to be taken in the $\theta$-direction compared to conditioning on a particular sequence of latent states. 

\noindent\parbox{\textwidth}{
\begin{flushleft}
 \begin{framedAlgorithm}[alternative \gls{RWEHMM}-based parameter estimation] \label{alg:rwehmm_based_parameter_estimation} Given $(\theta, \state_{1:T}) \coloneqq (\theta[l], \state_{1:\nTimeSteps}[l]) \in \mathrm{\Theta} \times \spaceState_{T,D}$,
 \begin{enumerate}
 \item \label{alg:rwehmm_based_parameter_estimation:1} sample $\Particle_{1:T} = \particle_{1:T}$ via Step~\ref{alg:iterated_rwehmm:1} of Algorithm~\ref{alg:iterated_rwehmm}, 
 
 \item \label{alg:rwehmm_based_parameter_estimation:2} sample $\varTheta' = \theta' \sim q_{\particle_{1:T}}(\theta, \ccdot)$ and set 
 \begin{align}
  r \coloneqq \frac{q_{\particle_{1:T}}(\theta\mathrlap{'}, \theta) \mu(\theta') \sum_{n_{1:T} \in [N]_0^T} \Target_{\theta\mathrlap{'}, T,D}(\particle_1^{n_1}, \dotsc, \particle_T^{n_T})}{q_{\particle_{1:T}}(\theta, \theta') \mu(\theta) \sum_{n_{1:T} \in [N]_0^T} \Target_{\theta, T,D}(\particle_1^{n_1}, \dotsc, \particle_T^{n_T})},
 \end{align}
  
 \item \label{alg:rwehmm_based_parameter_estimation:3} sample $U = u \sim \dUnif_{[0,1]}$,  
 \item \label{alg:rwehmm_based_parameter_estimation:4} if $u \leq r$,
 \begin{itemize}
   \item sample $K_{1:T} = k_{1:T} \sim \xi_{\theta\mathrlap{'},T}(\particle_{1:T}, \ccdot)$,
  \item return $(\theta[l+1], \state_{1:T}[l+1]) \coloneqq (\theta\mathrlap{'}, (\particle_1^{k_1}, \dotsc, \particle_T^{k_T}))$;
 \end{itemize}
 else,
  \begin{itemize}
   \item sample $K_{1:T} = k_{1:T} \sim \xi_{\theta,T}(\particle_{1:T}, \ccdot)$,
  \item return $(\theta[l+1], \state_{1:T}[l+1]) \coloneqq (\theta, (\particle_1^{k_1}, \dotsc, \particle_T^{k_T}))$.
 \end{itemize}


 \end{enumerate}
\end{framedAlgorithm}
\end{flushleft}
}

Since Algorithm~\ref{alg:rwehmm_based_parameter_estimation} relies on the \gls{RWEHMM} scheme, its computational cost again grows quadratically in $N$. This motivates us to propose Algorithm~\ref{alg:rwcsmc_based_parameter_estimation} which only requires $\bo(N)$ operations. To our knowledge, Algorithm~\ref{alg:rwcsmc_based_parameter_estimation} is novel. 
For simplicity, we only state the version of the algorithm with the backward-sampling but without the forced-move extension. Here,  $q_{\particle_{1:T}, a_{1:T-1}}(\theta, \diff \theta')$ is some proposal kernel for the parameters which may depend on the values of the particles and ancestor indices. Likewise, we have used the following notation for the probability of resampling the $n$th particle at time $t$ which was already introduced in Appendix~\ref{app:subsec:joint_law_rwcsmc}:
\begin{align}
  \ResamplingKernelRandomWalk{\theta,\timeIndex,\nDimensions}{\nParticles}((\particle_{\timeIndex-1:\timeIndex}, a_{\timeIndex-1}), \{\particleIndex\})
  & \coloneqq 
   \selectionFunctionBoltzmann{\particleIndex}(\{\logWeightRandomWalk_{\theta, \timeIndex}(\particle_{\timeIndex-1}^{\mathrlap{a_{\timeIndex-1}^\particleIndexAlt}}, \particle_{\timeIndex}^\particleIndexAlt) - \logWeightRandomWalk_{\theta, \timeIndex}(\particle_{\timeIndex-1}^{\mathrlap{a_{\timeIndex-1}^0}}, \particle_{\timeIndex}^0)\}_{\particleIndexAlt = 1}^\nParticles)\\
   & =
      \dfrac{\mutation_{\theta, t}(\particle_{t-1}^{a_{t - 1}^n}, \particle_t^n)\Potential_{\theta, t}(\particle_t^n)}{\sum_{\particleIndexAlt=0}^\nParticles \mutation_{\theta, t}(\particle_{t-1}^{a_{t-1}^\particleIndexAlt}, \particle_t^\particleIndexAlt) \Potential_{\theta, t}(\particle_t^\particleIndexAlt)}.
\end{align}
In addition, $\smash{a_t' \coloneqq {a_{\mathrlap{t}}'}^{0:N} \in [N]_0^{N+1}}$ denote values of a second set of proposed time-$t$ ancestor indices $\smash{A_t' \coloneqq {A_{\mathrlap{t}}'}^{0:N}}$.
 
\noindent\parbox{\textwidth}{
\begin{flushleft}
 \begin{framedAlgorithm}[alternative \gls{IRWCSMC}-based parameter estimation] \label{alg:rwcsmc_based_parameter_estimation} Given $(\theta, \state_{1:T}) \coloneqq (\theta[l], \state_{1:\nTimeSteps}[l]) \in \mathrm{\Theta} \times \spaceState_{T,D}$,
 \begin{enumerate}
 \item \label{alg:rwcsmc_based_parameter_estimation:1} sample $(\Particle_{1:T}, A_{1:T-1}) = (\particle_{1:T}, a_{1:T-1})$ via Step~\ref{alg:iterated_rwcsmc:1} of Algorithm~\ref{alg:iterated_rwcsmc}, 
 
 \item \label{alg:rwcsmc_based_parameter_estimation:2} sample $\varTheta' = \theta' \sim q_{\particle_{1:T}, a_{1:T-1}}(\theta, \ccdot)$, $\smash{A_{1:T-1}' = a_{1:T-1}' \sim \prod_{t=1}^{T-1} \prod_{n=0}^N \ResamplingKernelRandomWalk{\theta\mathrlap{'},t,\nDimensions}{\nParticles}((\particle_{t-1:t}, a_{t-1}'), \{{a_{\mathrlap{t}}'}^n\})}$,
 and set 
 \begin{align}
  r \coloneqq \frac{q_{\particle_{1:T}, a_{1:T-1}'}(\theta\mathrlap{'}, \theta) \mu(\theta') \prod_{t=1}^T \sum_{n=0}^N \mutation_{\theta\mathrlap{'},t}(\particle_{t-1}^{{a_{\mathrlap{t-1}}'}^n}, \particle_t^n) \Potential_{\theta\mathrlap{'},t}(\particle_t^n)}{q_{\particle_{1:T}, a_{1:T-1}}(\theta, \theta') \mu(\theta) \prod_{t=1}^T \sum_{n=0}^N \mutation_{\theta\mathrlap{'},t}(\particle_{t-1}^{a_{t-1}^n}, \particle_t^n) \Potential_{\theta\mathrlap{'},t}(\particle_t^n)},
 \end{align}
  
 \item \label{alg:rwcsmc_based_parameter_estimation:3}  sample $U = u \sim \dUnif_{[0,1]}$,  
 \item \label{alg:rwcsmc_based_parameter_estimation:4} if $u \leq r$,
 \begin{itemize}
 
    \item sample $\smash{\OutputParticleIndex_\nTimeSteps =\outputParticleIndex_\nTimeSteps \in [\nParticles]_0}$ with probability $\smash{\ResamplingKernelRandomWalk{\theta\mathrlap{'},T,\nDimensions}{\nParticles}((\particle_{T-1:T}, a_{T-1}'), \{k_T\})}$,
  
   \item  for $t = T-1,\dotsc,1$, set $\smash{\OutputParticleIndex_\timeIndex =\outputParticleIndex_\timeIndex \coloneqq {a_{\mathrlap{\timeIndex}}'}^{\outputParticleIndex_{\timeIndex+1}}}$,
  

  \item return $(\theta[l+1], \state_{1:T}[l+1]) \coloneqq (\theta\mathrlap{'}, (\particle_1^{k_1}, \dotsc, \particle_T^{k_T}))$;
 \end{itemize}
 else,
  \begin{itemize}
    \item sample $\smash{\OutputParticleIndex_\nTimeSteps =\outputParticleIndex_\nTimeSteps \in [\nParticles]_0}$ with probability $\smash{\ResamplingKernelRandomWalk{\theta,T,\nDimensions}{\nParticles}((\particle_{T-1:T}, a_{T-1}), \{k_T\})}$,
  \item for $t = T-1,\dotsc,1$, sample $K_t = k_t \in [N]_0$ with probability
  \begin{align}
   \frac{\mutation_{\theta,t}(\particle_{\timeIndex - 1}^{a_{\timeIndex - 1}^{\smash{\outputParticleIndex_\timeIndex}}}, \particle_\timeIndex^{\outputParticleIndex_\timeIndex}) \Potential_{\theta,t}(\particle_\timeIndex^{\outputParticleIndex_\timeIndex}) \mutation_{\theta,t+1}(\particle_\timeIndex^{\outputParticleIndex_\timeIndex}, \particle_{\timeIndex+1}^{\outputParticleIndex_{\timeIndex+1}})}{\sum_{\particleIndexAlt=0}^\nParticles \mutation_{\theta,t}(\particle_{\timeIndex - 1}^{\mathrlap{a_{\timeIndex - 1}^\particleIndexAlt}}, \particle_\timeIndex^\particleIndexAlt) \Potential_{\theta,t}(\particle_\timeIndex^\particleIndexAlt) \mutation_{\theta, \timeIndex+1}(\particle_\timeIndex^\particleIndexAlt, \particle_{\timeIndex+1}^{\outputParticleIndex_{\timeIndex+1}})},
  \end{align}

  \item return $(\theta[l+1], \state_{1:T}[l+1]) \coloneqq (\theta, (\particle_1^{k_1}, \dotsc, \particle_T^{k_T}))$.
 \end{itemize}


 \end{enumerate}
\end{framedAlgorithm}
\end{flushleft}
}

Algorithm~\ref{alg:rwcsmc_based_parameter_estimation} could potentially be improved by employing (conditional) systematic rather than multinomial resampling. In this case, the ancestor indices in both the numerator and denominator can be drawn based on the same uniform random number at each time step. Correlating the averages in the numerator and denominator of the acceptance ratio in this fashion may lead to better scaling with $T$ as in correlated pseudo-marginal methods \citep{deligiannidis2018correlated}.

%
%
%

\end{appendix}
\end{document}